\documentclass[11pt]{article}

\usepackage{graphicx}
\usepackage{color}
\usepackage{rotating}
\usepackage{amsmath}
\usepackage{amssymb}
\usepackage{amsthm}
\usepackage{float}
 \usepackage{url}
\usepackage{subfigure}

\usepackage{natbib} 

\usepackage{enumerate}
\usepackage{enumitem}

\usepackage{dsfont}
 
 \usepackage{algorithm}
\usepackage{algorithmic}
 
\usepackage{array,arydshln}

\newtheorem{theorem}{Theorem}[section]
\newtheorem{lemma}[theorem]{Lemma}

\newtheorem{remark}{Remark}[section]

\usepackage{soul} 

 \usepackage{xcolor}

\usepackage{hyperref} 

\hypersetup{
    colorlinks,
    linkcolor= {blue!90!black}, 
    citecolor={blue!60!black},
    urlcolor={blue!80!black}
}

\newcommand\bb {\mathbf b}
\newcommand\bc {\mathbf c}

\newcommand\bg {\mathbf g}
\newcommand\bh {\mathbf h}

\newcommand\bese {\mathbf s}

\newcommand\bu {\mathbf u}

\newcommand\bx {\mathbf x}

\newcommand\bz {\mathbf z}

\newcommand\bA {\mathbf A}
\newcommand\bB {\mathbf B}

\newcommand\bV {\mathbf V}
\newcommand\bW {\mathbf W}
\newcommand\bX {\mathbf X}

\newcommand\bZ {\mathbf Z}

\newcommand{\ELE}{\mathbb{L}}

\newcommand{\wELE}{\widehat{\ELE}}

\newcommand\wbc {\widehat{\bc}}

\newcommand\wbu {\widehat{\bu}}

\newcommand\wA {\widehat{A}}

\newcommand\wB {\widehat{B}}
\newcommand\wC {\widehat{C}}

\newcommand\wW {\widehat{W}}
\newcommand\wbW {\widehat{\bW}}
\newcommand\wY {\widehat{Y}}

\newcommand\wtbb {\widetilde{\bb}}

\newcommand\wtf {\widetilde{f}}

\newcommand\itA {{\mathcal{A}}}
\newcommand\itB {{\mathcal{B}}}
\newcommand\itC {{\mathcal{C}}}
\newcommand\itD {{\mathcal{D}}}
\newcommand\itE {{\mathcal{E}}}
\newcommand\itF {{\mathcal{F}}}
\newcommand\itG {{\mathcal{G}}}

\newcommand\itI {{\mathcal{I}}}
\newcommand\itJ {{\mathcal{J}}}

\newcommand\itL {{\mathcal{L}}}

\newcommand\itS {{\mathcal{S}}}

\newcommand\itV {{\mathcal{V}}}

\newcommand\bbe {\mbox{\boldmath $\beta$}}

\newcommand\bbech {\mbox{\scriptsize${\bbe}$}}

\newcommand\bla {\mbox{\boldmath $\lambda$}}

\newcommand\blach {\mbox{\footnotesize$\bla$}}

\newcommand\bSi {\mbox{\boldmath $\Sigma$}}
\newcommand\bUpsi {\mbox{\boldmath $\Upsilon$}}

\newcommand\wbeta {\widehat{\beta}}
\newcommand\wbbe {\widehat{\bbe}}

\newcommand\weta {\widehat{\eta}}

\newcommand\wbla {\widehat{\bla}}
\newcommand\wblach {\mbox{\scriptsize\boldmath $\widehat{\lambda}$}}

\newcommand\wmu {\widehat{\mu}}

\newcommand\wsigma {\widehat{\sigma}}

\newcommand\wtbbe {\widetilde{\bbe}}

\newcommand\wtbeta {\widetilde{\beta}}

\newcommand{\tuk}{\mbox{\scriptsize \sc t}}

\newcommand{\ini}{\mbox{\footnotesize \sc ini}}

\def\real{\mathbb{R}}
\def\natu{\mathbb{N}}
\def\qu{\mathbb{Q}}

\newcommand{\esp}{\mathbb{E}}
\newcommand{\prob}{\mathbb{P}}

\newcommand{\var}{\mbox{\sc Var}}

\newcommand{\convpp}{ \buildrel{a.s.}\over\longrightarrow}

\newcommand{\convprob  }{ \buildrel{p}\over\longrightarrow}

\newcommand{\trasp}{^{\mbox{\footnotesize \sc t}}}

\newcommand\bcero {{\bf{0}}}
\newcommand\buno {{\bf{1}}}

\def\dst{\displaystyle}

\def\median{\mathop{\mbox{median}}}

\def\argmin{\mathop{\mbox{argmin}}}

\newcommand{\cero}{\textbf{0}}

\newcommand\noi{\noindent}
\def\dst{\displaystyle}

\def\square{\ifmmode\sqr\else{$\sqr$}\fi}
\def\sqr{\vcenter{
         \hrule height.1mm
         \hbox{\vrule width.1mm height2.2mm\kern2.18mm
\vrule width.1mm}
         \hrule height.1mm}}

\newcommand{\act}{\mbox{\footnotesize\sc i}}
\newcommand{\noact}{\mbox{\footnotesize\sc ii}}

\begin{document}
 
\title{Robust variable selection for partially linear additive models}
\author{Graciela Boente$^1$ and  Alejandra Mercedes Mart\'{\i}nez$^2$\\ 
\small $^1$  CONICET and Universidad de Buenos Aires,  Argentina\\
\small $^2$   CONICET and Universidad Nacional de Luj\'an,   Argentina 
}
\date{}

\maketitle
\begin{abstract}
Among semiparametric regression models, partially linear additive models provide a useful tool to include additive nonparametric
components as well as a parametric component, when explaining   the  relationship between the response  and a set of explanatory variables. This paper concerns such models under sparsity assumptions for the covariates included in the linear component. Sparse covariates are frequent in  regression problems where the task of variable selection is usually of interest. As in other settings, outliers either in the residuals or in the covariates involved in the linear component have a  harmful effect. To simultaneously achieve  model selection for the parametric component of the model and resistance to outliers, we combine  preliminary robust estimators of the additive component, robust linear $MM-$regression estimators with a penalty such as SCAD on the coefficients in the parametric part. Under mild assumptions,  consistency results and rates of convergence for the proposed estimators are derived.  A Monte Carlo study is carried out to compare, under different models and contamination schemes, the performance of  the robust proposal  with its classical counterpart. The obtained results show the advantage of using the robust approach.  Through the analysis of a real data set, we also illustrate   the benefits of the proposed procedure.  
\end{abstract}

\noi \textbf{Keywords:} Partially Linear Additive Models; Penalties;  Robust Estimation; Sparse Regression Models 

\noi\textbf{AMS Subject Classification:}  62F35; 62G25

\normalsize
\newpage

\section{Introduction}\label{sec:intro}
Partial linear model and partially linear additive regression models (\textsc{plam}) were introduced to   deal with  the  \lq \lq curse of dimensionality\rq \rq present in fully nonparametric regression models. In both cases, we intend to model a response variable $Y$ using some covariates which are split in two subsets of variables where one subset enters into the model through a linear regression and the other   is included through unknown smooth functions.  Partial linear models are useful when the dimension of the latter subset is one or smaller than 3, since otherwise, it still suffers from the \lq \lq curse of dimensionality\rq \rq. Partially linear additive regression models   provide an   attempt to solve this problem by assuming an additive structure in the nonparametric component. More precisely, partially linear additive models   assume  that $(Y_i,\bZ_i\trasp,\bX_i\trasp)\trasp\in\real^{1+q+p}$, $1\leq i\leq n$, are independent and identically distributed vectors with the same distribution as $(Y,\bZ\trasp,\bX\trasp)\trasp$ such that 
\begin{equation}\label{eq:plam}
Y= \mu+\bbe\trasp\bZ+\sum_{j=1}^p \eta_j(X_j)+u\,=\, \mu+\bbe\trasp\bZ+\sum_{j=1}^p \eta_j(X_j)+ \sigma\varepsilon\,,
\end{equation}
 where the constant $\mu\in\real$, the vector $\bbe\in\real^q$, the univariate functions $\eta_j:\itI_j\to\real$, $1\leq j\leq p$, and the scale parameter $\sigma>0$ are the quantities to be estimated and the errors $\varepsilon$ are   independent from the  covariates $(\bZ\trasp,\bX\trasp)\trasp$. When second moments exist,  it is  assumed that $\esp(\varepsilon) =0$ and $\var(\varepsilon)=1$, so $\sigma>0$ stands for the unknown   scale parameter. In the context of robust regression,  these  two requirements are avoided by requiring that the error $\varepsilon$ has a symmetric distribution $F(\cdot)$ where the scale parameter of $F(\cdot)$ equals 1 to identify $\sigma$.

In order to ensure the identifiability of the additive components $\eta_j$,   we require  that   $\int_{\itI_j}\, \eta_j(x)\,dx=0$, for $1\leq j\leq p$.  Furthermore, without loss of generality we assume that $\itI_j=[0,1]$, for $1\le j\le p$.
 
The partially linear additive regression model \eqref{eq:plam} includes as particular cases  the partial lineal model when $p=1$  and  the additive one  when the regression parameter equals zero and the linear regression model when $\eta_j \equiv 0$, for $1\le j\le p$.  It is worth mentioning that partially linear additive models  allow to include in the linear  component  covariates which are discrete, dummy or unbounded. Besides, as mentioned in \citet{Ma:2012} and \citet{Opsomer:Ruppert:1999}, they provide a parsimonious model which can be easily interpreted and they are suitable   when the user is   quite  certain of  the linear relation  between the response and some subset of covariates, but not about the shape of the relationship with the other ones.  

In contrast to kernel methods, splines and methods based on series allow to provide simultaneous estimators of the parameter $\bbe$ and the  nonparametric
components. Among others, we can mention the papers by  \citet{Stone:1985}  for additive models, by \citet{Li:2000} who considered general series  and by \citet{Ma:Yang:2011} who combined   spline estimation and kernel methods. Robust procedures based on splines were studied in    \citet{He:Shi:1996} and \citet{he:zhu:fung:2002} for  the particular case of partial linear models, that is, when $p=1$, while \citet{boente:martinez:2023} proposed $MM-$estimators for partial linear additive models.

In practice, in a first step of modelling, researchers frequently introduce all possible variables  in the model based on their own experience. In this sense, some variables that have a null impact on the response variable will reduce the prediction capability of the model.   As it is well known, sparse statistical models correspond to situations where there are only a small number of non--zero parameters and for that reason, they are much easier to interpret than dense ones, see  \citet{Hastie:etal:book:2015}.  In this way, variable selection plays an important role during modelling.

Sparse models have raised a paradigm shift in statistical modelling, since the traditional estimating approaches to regression do not impose any restrictions on the parameters.  It is worth mentioning that one of the main goals under a sparse setting is variable selection, that is, to identify variables related to non--null coefficients. A possible and useful way to perform automatic variable selection is by including a penalty term in the optimization problem that defines the estimators. Some of the advantages of these methods are their strong interpretability and the low computing cost, see \citet{efron:hastie:2016} for an overview on penalized methods. One extended procedure to perform variable selection is to consider LASSO estimators  introduced in \citet{tibshirani:1996}. These estimators, which add to the least squares loss an   $\ell_1$ regularization,   are effective for variable selection,   but tend to choose too many features. \citet{zou:hastie:2005} and \citet{zou:2006}  considered  alternative regularizations. The first authors include a penalty which combines both $\ell_1$ and $\ell_2$ norms and is known as  the Elastic Net penalty. In contrast to these deterministic penalties, \citet{zou:2006} considered  a random penalty,  the adaptive LASSO denoted from now on ADALASSO, which is defined from an initial consistent estimator. For a regression model $Y=  \bbe\trasp\bZ+ \sigma\varepsilon$ with $\bbe\in\real^q$, the  ADALASSO penalty is defined by $ \iota \,  \sum_{j=1}^q  {|\beta_j|}/{|\wtbeta_j|^{\gamma}}$,
for some $\gamma>0$, where we understand that $|\beta_j|/|\wtbeta_j|^{\gamma}=\infty$ if $|\wtbeta_j|=0$ but $|\beta_j|\ne 0$, while $|\beta_j|/|\wtbeta_j|^{\gamma}=0$ if $|\wtbeta_j|=|\beta_j|= 0$.  Elastic Net preserves the sparsity of LASSO and  maintains  some of the desirable predictive properties of Ridge regression, while   ADALASSO gives a more realistic scenario.     \citet{fan:li:2001} and \citet{zhang:2010} proposed alternative penalties, the  SCAD and MCP penalties, respectively, which are bounded and lead to sparse estimators.

In partially linear additive regression models, sparse models for the regression component were considered in \citet{liu:etal:2011}, \citet{du:etal:2012} and  \citet{lian:2012} who  developed   variable selection procedures   based on least squares regression, spline approximations for the nonparametric components and  SCAD or  ADALASSO penalizations on the regression parameter $\bbe$.     
However, these estimators are based on a least squares approach  and assume that  the error   has finite variance, so, as in partial linear models, a small proportion  of atypical data  may seriously  affect the estimations. A more resistant approach based on quantile regression and spline approximation was suggested in \citet{Guo:etal:2013} and  \citet{sherwood:wang:2016} who also considered variable selection in sparse models. An   approach based on penalized splines was discussed in \citet{Koenker:2011}. As mentioned in \citet{boente:martinez:2023} quantile estimators are related to an unbounded loss  function and, for that reason, as in linear regression models, they may be affected  by high--leverage outliers.  

Our motivating example corresponds to  the plasma beta-carotene level data set, collected by \citet{Nierenberg:etal:1989} and available at \url{http://lib.stat.cmu.edu/datasets/Plasma_Retinol}.  This data set, that consists of 315 observations, corresponds to a cross--sectional study developed to investigate the relationship between personal characteristics, dietary factors, ,plasma concentrations of retinol, beta--carotene and other carotenoids. The interest on this data set  arises from the fact that some studies     suggested that low levels of beta--carotene may be associated with an increased risk of cancers such as lung, colon, breast, and prostate cancer, see \citet{harrell:2002}. 

The data  set was also considered in \citet{liu:etal:2011} and \citet{Guo:etal:2013} who proposed a partially linear additive model.  \citet{liu:etal:2011}  estimated the  parameters using a least squares approach combined with  the SCAD penalty, while  \citet{Guo:etal:2013}  used   composite quantile regression and adaptive LASSO. Another difference  between these two papers is the chosen response, which is the  logarithm of the plasma beta--carotene, labelled \textsc{betaplasma},   in  \citet{liu:etal:2011}, and the   \textsc{betaplasma} without any transformation in   \citet{Guo:etal:2013}.  Both authors modelled the response using a \textsc{plam} with  covariates associated to the linear component being the sex,  smoking type,  body mass index, the ingestion or not of vitamin complements, the number of calories consumed per day, the grams of fat consumed per day, the grams of fiber consumed per day,  the dietary  beta-carotene consumed and the number of alcoholic drinks consumed per week, while the age and the cholesterol consumed were  included in the model through additive nonparametric components. In \citet{liu:etal:2011}, the covariate  the grams of fiber consumed per day was included in the linear component, while \citet{Guo:etal:2013} included it in the additive component. Both  \citet{liu:etal:2011} and \citet{Guo:etal:2013} observed the presence of one extremely high leverage point in alcohol consumption  which was deleted prior to the analysis, then only 314 observations were used. 

The effect of vertical and high--leverage outliers in sparse linear regression models when considering unbounded loss functions was discussed in \citet{smucler:yohai:2017} and it is expected that the distortion created by atypical data when considering penalized  least squares or penalized quantile estimators will appear also in  partially linear additive models. This motivates the need of robust procedures that combine a bounded loss function with a penalization to   select significant variables without any prior analysis to discard observations. The procedure may also be useful to identify atypical observations.

The rest of the paper is organized as follows. Section \ref{sec:estimadores} introduces the robust penalized estimators,  a robust procedure to select the penalty parameter is described in Section \ref{sec:RBIC}. The algorithm used to compute the estimators and   possible robust initial estimators of the scale and additive functions are described in Sections \ref{sec:algo} and \ref{sec:Prelim-est}, respectively. The asymptotic properties of the proposed estimators including  consistency results and variable selection properties are stated in Section  \ref{sec:asympt}. Section \ref{sec:monte} reports the results of a Monte Carlo study  conducted to examine the small sample properties of the proposed procedure under different contamination schemes. The usefulness of the proposed methodology is illustrated in Section 
\ref{sec:realdata} on the   real data set  described above, while some final comments are presented in Section \ref{sec:coments}. All proofs are relegated to the Appendix.

\section{The robust penalized estimators}{\label{sec:estimadores}}
As mentioned in the Introduction, variable selection is an important issue when too many variables are introduced in the model even when only a small group of them are relevant. Penalized regression procedures bet  on the sparsity principle and have shown to be effective in variable selection, when considering appropriate penalties.  Among them, the SCAD penalty, proposed by \citet{fan:li:2001}, has advantages over the $\ell_q$ and the hard thresholding penalties due to the sparsity and continuity properties of the resulting estimators. For $\lambda>0$, the SCAD penalty function is defined as $ \sum_{s=1}^q p_{\lambda}(|b_s|)$ where $p_\lambda:\real_{\ge 0} \to \real_{\ge 0}$ is given by 
\begin{equation}
\label{eq:SCADpenalidad}
p_\lambda(b)=
\lambda b \buno_{\{ b \leq \lambda\}} - \frac{1}{2(a-1)}
\left(   b^2-2\, a \,\lambda\, b + \lambda^2 \right)  \buno_{ \{  \lambda  < b\leq a \lambda\}} + 
\frac{\lambda ^2 (a+1)}{2 } \buno_{\{ b > a \lambda\}}\,,
\end{equation}
  with $\buno_A$  the indicator of the set $A$.  Another well known penalty is the minimax concave penalty (MCP) proposed by \citet{zhang:2010}, which corresponds to the choice
\begin{equation}
p_\lambda(b)  =   
\left(\lambda  b -  \frac{b^2}{2a}\right) \buno_{\{b \leq a\lambda \}}+
  \frac{a \lambda ^2}{2}  \buno_{\{b > a\lambda\}} \,.   
\label{eq:MCPpenalidad}
  \end{equation}
For both penalties, the  positive constant $a$, which is larger than 2 for SCAD, is selected by the user. 

In the non--sparse setting, $MM-$estimators for partially linear additive regression models were defined in  \citet{boente:martinez:2023}. Their approach uses $B-$spline approximations to provide a set of candidates for the estimators of the additive components. As in linear regression models,  to obtain a  scale estimator, they first compute an $S-$estimator using a $\rho-$function $\rho_0$. In a second step, using $\rho-$function  $\rho_{1}$ such that $\rho_{1} \le \rho_{0}$ and the scale estimator, they  define the final $M-$estimator.  As for the least squares or the quantile estimators, the $MM-$estimators do not lead to sparse estimators.
This entails that they do not allow to make variable selection on covariates related to the linear component and may have a bad performance regarding robustness and efficiency. Hence, to improve the behaviour of the robust estimators of $\bbe$, we will to include a regularization term that penalizes candidates without few non--zero components.

To define the penalized estimators, let $\wsigma$, $\wmu$ and $\weta_j$, for $1\leq j\leq p$,  be preliminary strong consistent estimators of $\sigma$, $\mu$ and $\eta_j$, respectively. In Section \ref{sec:Prelim-est}, we discuss a possible choice for the initial estimators.

The penalized estimators will be defined using  a bounded $\rho-$function  $\rho_{1}$ as defined in \citet{maronna:etal:2019libro}. One possible choice for $\rho_1$ is the bisquare Tukey's function defined as $\rho_{1}=\rho_{\,\tuk,\,c_1}$, where $\rho_{\,\tuk,\,c}(t) =\min\left(1 - (1-(t/c)^2)^3, 1\right)$. It is worth mentioning that   the scale estimator $\wsigma$ corresponds to an $S-$estimator obtained using a  $\rho-$function $\rho_0$, then to ensure good robustness properties,   $\rho_1$ must satisfy $\rho_1\leq \rho_0$,  a property which holds when considering the Tukey's loss function if $c_1\geq c_0$. 

To simplify the notation denote as
\begin{align}{\label{eq:funcionLn}}
L_n(a,g_1,\dots,g_p,\varsigma,\bb)&=\frac{1}{n}\sum_{i=1}^n \rho_1\left(\frac{Y_i-a-\bb\trasp\bZ_i-\sum_{j=1}^p g_j(X_{ji})}{\varsigma}\right)\,,
\end{align}
and as 
\begin{equation}{\label{eq:funcionL}}
L(a,g_1,\dots,g_p,\varsigma,\bb)=\esp\rho_1\left(\frac{Y_1-a-\bb\trasp\bZ-\sum_{j=1}^p g_j(X_{j})}{\varsigma}\right)\,,
\end{equation}
its population counterpart. Furthermore,   $\itJ_{\blach}(\bb)$ will stand for a penalty function, chosen by the user,
depending on a tuning parameter $\bla=(\lambda_1,\dots,\lambda_q)\trasp\in \real^q$, which measures the   model complexity, for instance, $\itJ_{\blach}(\bb)=\sum_{s=1}^q p_{\lambda_{s}}(|b_s|)$ with $p_\lambda$ a univariate penalty such as those defined in \eqref{eq:SCADpenalidad} and \eqref{eq:MCPpenalidad}. Note that we allow for different parameters  $\lambda_{s}$  controlling the sparsity of the parametric component for each component of $\bb$.

We define  the penalized robust estimators of $\bbe$ as
\begin{equation}\label{eq:opt-plam}
\wbbe=\argmin_{\bb\in\real^q}PL_{n,\blach}(\bb)\,,
\end{equation}
where
\begin{equation}\label{eq:pl}
PL_{n,\blach}(\bb)=L_n\left(\wmu,\weta_1,\dots,\weta_p,\wsigma, \bb\right)+\itJ_{\blach}(\bb)\,.
\end{equation}

\subsection{Selection of the penalty parameter}{\label{sec:RBIC}}
As discussed for instance in \citet{Efron:etal:2004} and \citet{Meinshausen}, the selection of the penalty parameter plays an important role  when fitting sparse models, since  it tunes the complexity of the model. In this paper, we propose a robust $BIC$ criterion  used to select the penalty parameter. To be more precise, let $\Lambda\subset \real^q$  be the set of possible values for  $\bla$ to be considered. From now on,   $\wsigma$, $\wmu$, $\weta_j$ are the preliminary estimators of $\sigma$, $\mu$ and $\eta_j$, respectively, that do not depend on the penalty parameter. The robust  criterion selects the penalty parameter by minimizing over $\Lambda$  the following $RBIC$ criteria 
\begin{equation}
\label{eq:RBIC}
RBIC(\bla)= \log\left(\wsigma^2\sum_{i=1}^n \rho\left(\frac{r_i(\wbbe_{\blach})}{\wsigma}\right)\right)+ df_{\blach} \frac{\log(n)}{n}\,,
\end{equation}
where $r_i(\bb)=Y_i-\wmu-\sum_{j=1}^p \weta_j(X_{ij})-\bZ_i\trasp\bb$,  $\wbbe_{\blach}$ is the estimator obtained when considering the penalty parameter $\bla$ and $df_{\blach}$ is the number of non--zero  components in $\wbbe_{\blach}$.

\subsection{Algorithm}{\label{sec:algo}}
In this section, we present an algorithm to numerically obtain the estimators defined through  \eqref{eq:opt-plam}. We present a general algorithm that uses a bounded loss function and the $RBIC$ criterion to select $\bla$. The particular case of the penalized least squares estimators combined with  the $BIC$ criterion to select the penalty parameters is easily obtained taking the square loss function.

\begin{algorithm}[ht!]
\caption{General algorithm}
\label{alg1}
\begin{algorithmic}[1]
\STATE Obtain preliminary estimators $\wmu$, $\weta_j$, for $1\leq j\leq p$, and a scale estimator $\wsigma$. A possible choice is to compute the robust estimators obtained with the non-penalized procedure proposed in \citet{boente:martinez:2023}. Another one is to use the preliminary estimators of $\mu$ and $\weta_j$ (and $\wbbe$) described in Section \ref{sec:Prelim-est} and then to consider a robust scale estimator, such as the MAD, over the initial residuals $r_{\ini,i}=Y_i-\wmu-\wbbe\trasp\bZ_i-\sum_{j=1}^p \weta_j(X_{ij})$ to obtain an estimator $\wsigma$ of $\sigma$ or the $S-$scale defined through \eqref{eq:s-est}. 
\STATE Consider a grid for $\bla$ of $N$ elements: $\bla_1,\dots,\bla_N$. 
\FOR{$j = 1$ \TO $N$}  
\STATE Compute the regression estimator defined through  \eqref{eq:opt-plam} using the penalty parameter $\bla_j$, that is, 
\begin{equation*}
\wbbe_{\blach_j}=\argmin_{\bb\in\real^q}PL_{n,\blach_j}(\bb)\,,
\end{equation*}
\STATE Obtain $RBIC(\bla_j)$ defined in  \eqref{eq:RBIC} with $\wbbe_{\blach}=\wbbe_{\blach_j}$.
\ENDFOR
\STATE Select $\widehat{j}$ as
$$\widehat{j}=\argmin_{j=1,\dots,N}\, RBIC(\bla_j) \,.$$
\STATE Set $\wbla=\bla_{\widehat{j}}$ and $\wbbe_{\wblach}$ as the penalized estimator of $\bbe$.
\end{algorithmic}
\end{algorithm}
 
We  still have to describe a procedure to compute the estimators defined through   \eqref{eq:opt-plam} for each fixed value of the penalty parameter.
To simplify the notation, let  $Y_i^{\star}=Y_i-\wmu-\sum_{j=1}^p \weta_j(X_{ij})$, where $\wmu$, $\weta_j$, for $1\leq j\leq p$ are the preliminary estimators of $\mu$ and $\eta_j$, respectively.  We will derive the algorithm for the case where $\itJ(\bb)=\sum_{s=1}^q p_{\lambda_s}(|b_s|)$ and since $\bla$ is fixed, we will simply write $p_{\lambda_s}(\cdot)=p_s(\cdot)$ to avoid burden notation.

In the sequel, we will consider penalties, such as SCAD or MCP, which are such that $p_{\lambda}$ is twice continuously differentiable at $(0, +\infty)$ and $p_{\lambda}(0)=0$. This ensures, that the penalty  can be locally approximated by a quadratic function as done, for instance, in \citet{fan:li:2001}. Let  $t_0$ be an initial point close to $0$, then there exists a quadratic function $q(t)$ such that, $q$ is even,  $q(t_0)=p(|t_0|)$, and $q^{\prime}(t_0)=p^\prime(|t_0|)$.
Using the first condition, we have that $q(t)=a+bt^2$. Then, the two last conditions imply that
$  a+b|t_0|^2= p(|t_0|)$ and $2b|t_0|= p^\prime(|t_0|)$. Using the latter equation, we get that $b=p^\prime(|t_0|)/(2|t_0|)$ and replacing in the former one, we conclude that $a=p(|t_0|)-p^\prime(|t_0|)\,t_0^2/(2|t_0|)$. Therefore, the quadratic approximation of the penalty function equals
$$
q(t)=p(|t_0|)-\frac{p^\prime(|t_0|)}{2|t_0|}t_0^2+\frac{p^\prime(|t_0|)}{2|t_0|}t^2=p(|t_0|)+\frac{p^\prime(|t_0|)}{2|t_0|}(t^2-t_0^2)
$$
that is, for $t$ close to $t_0$, we have that 
$$p(|t|)\approx p(|t_0|)+\frac{p^\prime(|t_0|)}{2|t_0|}(t^2-t_0^2)\,.$$
In this way, taking $\bb_0$ close to the minimizer of \eqref{eq:opt-plam} and $\bb$ close to $\bb_0$ with $|b_{0s}|>0$ for $s=1,\dots, q$, we have that
$$p_s(|b_s|)\approx p_s(|b_{0s}|)+\frac{p_s^\prime(|b_{0s} )}{2|b_{0s}|}(b_s^2-b_{0s}^2)\,.$$
Define the diagonal matrix $\bUpsi_{\bb_0}\in\real^{q\times q}$   as  
$$\bUpsi_{\bb_0}=\mbox{diag}\left\{\frac{p_1^\prime(|b_{01}|)}{2|b_{01}|},\dots,\frac{p_q^\prime(|b_{0q}|)}{2|b_{0q}|}\right\}\,.$$ 
The objective function given in the right hand side of \eqref{eq:opt-plam} can then be approximated as
\begin{equation}\label{eq:quadratic-pl}
\frac{1}{n}\sum_{i=1}^n\rho\left(\frac{r_i(\bb)}{\wsigma}\right)+\bb\trasp \bUpsi_{\bb_0} \bb\,,
\end{equation}
with $r_i(\bb)=Y_i^{\star}- \bb\trasp \bZ_i$ and $\wsigma$ is a preliminary estimator of $\sigma$. Hence, the estimators obtained by minimizing \eqref{eq:quadratic-pl} will be close to the minimizers of \eqref{eq:opt-plam} for $\bla$ fixed.

Finally, the minimization in \eqref{eq:quadratic-pl} can be obtained by a reweighted procedure as it is usual for $M-$estimators. 
Denote $\bb^{(m)}$ the estimator obtained in the $m-$step, $w(t)=\psi(t)/t$ and $r_i(\bb)=Y_i^{\star}-\bb\trasp\bZ_i$. Differentiating \eqref{eq:quadratic-pl} with respect to $\bb$ and then multiplying and dividing by $(Y_i^{\star}-\bZ_i\trasp\bb)/\wsigma$, we easily obtain
\begin{align*}
0&= \frac{1}{n}\sum_{i=1}^n \psi\left(\frac{Y_i^{\star}-\bb\trasp\bZ_i}{\wsigma}\right)\left(-\frac{\bZ_i}{\wsigma}\right)+2\bUpsi_{\bb^{(m)}} \bb  \\
&=\frac{1}{n}\sum_{i=1}^n w\left(\frac{r_i(\bb)}{\wsigma}\right)\frac{Y_i-\bZ_i\trasp\bb}{\wsigma}\left(-\frac{\bZ_i}{\wsigma}\right)+2\,\bUpsi_{\bb^{(m)}}  \bb  \,.
\end{align*}
The estimator  $\bb^{(m+1)}$ of the $(m+1)-$step may be defined as the solution of
$$\frac{1}{n}\sum_{i=1}^n w_{i,m}\frac{Y_i^{\star}-\bb\trasp\bZ_i}{\wsigma}\left(-\frac{\bZ_i\trasp}{\wsigma}\right)+2\,\bUpsi_{\bb^{(m)}}  \bb =0$$
where $w_{i,m}=w(r_i(\bb^{(m)})/\wsigma)$ are the weights obtained with the estimator computed in the $m-$step. Noticing that the latter equation is equivalent to
$$-\frac{1}{n}\sum_{i=1}^n \frac{w_{i,m}}{\wsigma^2}  Y_i^{\star}\, \bZ_i+\left(\frac{1}{n}\sum_{i=1}^n \frac{w_{i,m}}{\wsigma^2}\bZ_i\bZ_i\trasp+2 \bUpsi_{\bb^{(m)}} \right)\bb=0\,,$$
we get that
 $$\bb^{(m+1)}=\left(\frac{1}{n}\sum_{i=1}^n \frac{w_{i,m}}{\wsigma^2}\bZ_i\bZ_i\trasp+2\, \bUpsi_{\bb^{(m)}} \right)^{-1}\frac{1}{n}\sum_{i=1}^n \frac{w_{i,m}}{\wsigma^2} Y_i^{\star} \bZ_i\,.$$
The following Algorithm \ref{alg2} summarizes the described procedure to compute the regression estimators.

\begin{algorithm}[ht!]
\caption{Optimization problem}
\label{alg2}
\begin{algorithmic}[1]
\STATE Let $m = 0$ and $\bb^{(0)}$ be an initial estimator of $\bbe$ and $\wsigma$ an estimator of the scale parameter $\sigma$ as before.
\REPEAT
\STATE $m  \leftarrow m + 1$
\STATE Compute 
$$\bUpsi_{\bb^{(m)}} =\mbox{diag}\left\{\frac{p_1^\prime(|b_{m1}|)}{2|b_{m1}|},\dots,\frac{p_q^\prime(|b_{mq}|)}{2|b_{mq}|}\right\}\,,$$
with $\bb^{(m)}=(b_{m1},\dots,b_{mq})\trasp$.
\STATE Compute $w_{i,m}= w(r_i(\bb^{(m)})/\wsigma)$.
\STATE Define 
$$\bb^{(m+1)}=\left(\frac{1}{n}\sum_{i=1}^n \frac{w_{i,m}}{\wsigma^2}\bZ_i\bZ_i\trasp+2\, \bUpsi_{\bb^{(m)}} \right)^{-1}\frac{1}{n}\sum_{i=1}^n \frac{w_{i,m}}{\wsigma^2} Y_i^{\star} \bZ_i \,.$$
\UNTIL convergence
\end{algorithmic}
\end{algorithm}

\subsection{Preliminary estimates of $\mu$ and $\eta_j$}{\label{sec:Prelim-est}}
In order to obtain a scale estimator $\wsigma$ and preliminary estimators $\wmu$ and $\weta_j$ of $\mu$ and $\eta_j$, for $j=1,\dots,p$,   we will first consider penalized  $S-$estimators. This initial estimators will use a $\rho-$function $\rho_0$ and $B-$splines to approximate each additive function. As in \citet{boente:martinez:2023}, for each $j=1,\dots,p$, once fixed the spline order $\ell_j$, the number of internal knots $N_{n,j}$ and their location, a basis of centered $B-$splines $\{B_s^{(j)}\,:\,1\leq s\leq k_{j}\}$ with $k_j=k_{n,j}=N_{n,j}+\ell_j$ can be used to approximate $\eta_j(x)$ as $\sum_{s=1}^{k_j}\lambda_s^{(j)}B_s^{(j)}(x)$. Taking into account that $\sum_{s=1}^{k_j}B_s^{(j)}(x)=0$, the approximation may be rewritten as
$$\sum_{s=1}^{k_j}\lambda_{s}^{(j)} B_s^{(j)}(x)=\sum_{s=1}^{k_j-1} \left(\lambda_{s}^{(j)} -\lambda_{k_j}^{(j)}\right) B_s^{(j)}(x)\,.$$
Denote $K=\sum_{j=1}^p k_j -p$ the effective dimension of the considered space used to approximate the nonparametric additive components. If we define $\bc^{(j)}=(c_1^{(j)}, \dots, c_{k_j-1}^{(j)})\trasp \in \real^{k_j-1}$ with $c_s^{(j)}=\lambda_{s}^{(j)}-\lambda_{k_j}^{(j)}$ and, for   $1\le i\le n$,  the residuals of the partially linear additive model are
  \begin{equation} \label{eq:residuals}
r_i(a,{\bb},  \bc)=  \, Y_i -a-\bb\trasp\bZ_i-\sum_{j=1}^p \sum_{s=1}^{k_j-1} c_{s}^{(j)} B_s^{(j)}(X_{ij})=  \, Y_i -a-\bb\trasp\bZ_i-\bc\trasp\bV_i\, ,
\end{equation}
where   $\bV_i=(\bV^{(1)}(X_{i1})\trasp,\dots,\bV^{(p)}(X_{ip})\trasp)\trasp$, $\bV^{(j)}(t)=(B_1^{(j)}(t),\dots,B_{k_j-1}^{(j)}(t))\trasp$, for $1\le j\le p$ and $\bc=( \bc^{(1)\mbox{\footnotesize{\sc t}}},\dots,\bc^{(p)\mbox{\footnotesize{\sc t}}})\trasp\in \real^K$. From now on, $\eta_{\bc}$ stands for
 $\eta_{j,\bc}(t)=\sum_{s=1}^{k_j-1}c_{s}  B_s^{(j)}(x)$.

Let now define the optimization problem that will define the initial $S-$estimators. Let $s_n(a,\bb, \bc)$  be the $M-$scale estimator of the residuals related to $\rho_0$, that is,  $s_n(a,\bb, \bc)$ is the solution on $s$ of the implicit equation
$
(1/n)\sum_{i=1}^n \rho_{0}\left( {r_i(a, \bb,\bc)}/{s}\right) \, = \, b $,
where $0<b<1$. Hence,  $s_n(a,\bb, \bc)$  satisfies
\begin{equation*} \label{eq:s-est}
\frac{1}{n}\sum_{i=1}^n \rho_{0}\left(\frac{r_i(a, \bb,\bc)}{s_n(a,\bb, \bc)}\right) \, = \, b\,  .
\end{equation*}
Then, the \lq\lq ridge\rq\rq  ~$S-$estimators are  defined as
$$(\wmu_{\ini}, \wbbe_{\ini},\wbc_{\ini})=   \ \argmin_{a\in \real, \bb\in \real^q, \bc\in \real^K}  \, s^2_n(a,\bb, \bc)+\lambda_1\|\bb\|^2+\lambda_2\sum_{j=1}^p \|\bc^{(j)}\|_{H_j}^2\,,$$
where $\lambda_1=\lambda_{1,n}$ and $\lambda_2=\lambda_{2,n}$ are regularization parameters,  $\|\cdot\|$ is the Euclidean norm in $\real^q$, $\|\bc^{(j)}\|_{H_j}=(\bc^{(j)\mbox{\footnotesize\sc t}}H_j\bc^{(j)})^{1/2}$ with $H_j$ of dimension $(k_j-1)\times (k_j-1)$ and its $(s,s')$ element given by
$\int_0^1 B_s^{(j)}(t)B_{s'}^{(j)}(t)\,dt$, meaning that $\|\bc^{(j)}\|_{H_j}^2= \int_0^1  \eta_{j,\bc^{(j)}}^2(t)dt$. Then, the residual scale estimator is defined as $\wsigma=s_n(\wmu_{\ini}, \wbbe_{\ini},\wbc_{\ini})$ and the estimators of the regression functions $\eta_j$ are given by $\weta_j=\eta_{j,\wbc_{\ini}^{(j)}}$.
It is worth mentioning that these estimators do not select variables, but allow for the estimators computation when $q+K$ is large.
 
\section{Asymptotic results}{\label{sec:asympt}}
\subsection{Consistency results}{\label{sec:consist}}
From now on,  for $1\le j\le p$,  $\itS_{j}$ stands for some finite--dimensional space of dimension $k_j$, that is,  
\begin{equation}
\label{eq:itSj}
\itS_{j}=\left\{  \sum_{s=1}^{k_j} c_s \,   B_s^{(j)}(x) \, ,\, \bc \in \real^{k_j}\right\} \, ,
\end{equation}
where $\{B_s^{(j)}: 1\leq s\leq k_{j}\}$ stands for the linear space basis. Assumption \ref{ass:weta}  below states that, for $1\le j\le p$, $\weta_j\in \itS_j$. As mentioned in  \citet{boente:martinez:2023},  besides the  centered $B-$splines bases of order $\ell_j$ and size $k_j+1$, other basis may be considered. In the sequel, $\|\cdot\|$ refers to the Euclidean norm in $\real^q$ and for any continuous function $v\,:\,\real\to\real$,  $\|v\|_{\infty}=\sup_t|v(t)|$.

In order to derive the asymptotic properties of the penalized estimator, we will assume that
\begin{enumerate}[label=\textbf{C\arabic*}]
\setcounter{enumi}{0}
\item\label{ass:rho_bounded_derivable}
 \begin{itemize}
\item[\textbf{(a)}] The function $\rho_1 \,:\,\real\to [0;+\infty)$ is bounded, continuous, even, non-decreasing in $[0;+\infty)$ and such that $\rho_1(0)=0$. Furthermore, $\lim_{t\to+\infty}\rho_1(t)\neq 0$ and if $0\leq u<v$ with $\rho_1(v)<\sup_t\rho_1(t)$ then $\rho_1(u)<\rho_1(v)$. Without loss of generality, since $\rho_1$ is bounded, we assume that $\sup_t \rho_1(t)=1$.
\item[\textbf{(b)}]  $\rho_1$ is continuously differentiable with bounded derivative $\psi_1$. Moreover, the function $\zeta_1\,:\,\real\to\real$ defined as $\zeta_1(t)=t\psi_1(t)$  is bounded.
 \end{itemize}
\item\label{ass:densidadepsilon} The random variable $\varepsilon$ is distributed as $F_0$ and has density function $f_0(t)$ that is even, monotone non-decreasing in $|t|$, and strictly decreasing for $|t|$ in a neighbourhood of $0$.
\item\label{ass:probaZ}   $\prob(\bZ\trasp\bbe=0)<c $ for all non-zero $\bbe\in\real^q$, for some constant $0<c\le 1-b_{\rho_1}$,  where $b_{\rho_1}=\esp \rho_1(\varepsilon)$.
\item\label{ass:wsigma} $\wsigma$ is a strong consistent estimator of $\sigma$. 
\item\label{ass:weta} For $1\leq j\leq p$, $\weta_j\in\itS_j$, where $\itS_j$ is defined in \eqref{eq:itSj}. Furthermore, $ \sum_{j=1}^p \|\weta_j- \eta_j\|_{\infty}\convpp 0$ and  $\wmu\convpp\mu$.
\end{enumerate}
 
\begin{remark}[\textbf{Comments on assumptions}]{\label{remark:comentarios}}
Assumptions \ref{ass:rho_bounded_derivable} and \ref{ass:densidadepsilon}  are standard conditions in regression models, respectively. The latter is a condition usually required jointly with  \ref{ass:probaZ}  to ensure Fisher--consistency.
 It is worth mentioning that  assumptions \ref{ass:rho_bounded_derivable}, \ref{ass:densidadepsilon} and \ref{ass:probaZ} together with Lemma 3.1 in \citet{yohai:1985} imply that, for any $\varsigma>0$,  the function $\ELE_{\varsigma}:\real^q\to \real$ defined as
 $$\ELE_{\varsigma}(\bb)=L(\mu,\eta_1,\dots,\eta_p,\varsigma,\bb)=\esp\rho_1\left(\frac{Y-\mu-\sum_{j=1}^p \eta_j(X_j)-\bZ\trasp \bb}{\varsigma}\right)\,,$$ 
 has a unique minimum at $\bb=\bbe$.
Strong consistency of the preliminary scale  estimator is stated in \ref{ass:wsigma}, while  in \ref{ass:weta} we require consistency to the estimators of $\mu$ and $\eta_j$, $1\le j\le p$.
 \end{remark}
 
 It is worth noticing that, in Theorem \ref{teo:1} below, the parameter $\bla=\bla_n$ may be deterministic or random and in the
latter situation, the only requirement is that $\itJ_{\blach}(\bbe)\convpp 0$. In particular, when $\itJ_{\blach}(\bb)=\sum_{s=1}^q p_{\lambda_{s}}(|b_s|)$ and $p_{\lambda_{s}}$ are the functions related to the penalties   SCAD or MCP defined through \eqref{eq:SCADpenalidad} or \eqref{eq:MCPpenalidad}, respectively,  this condition holds when $\lambda_{s}=\lambda_{n,s}\convpp 0$, for $1 \le s \le q$. Furthermore, since different penalty parameters are allowed for each coordinate, our results include the  ADALASSO, $\itJ_{\blach}(\bb)=\iota_n \sum_{s=1}^q |b_s|/ |\wbeta_{\ini,s}|$, where $\wbbe_{\ini}$ a preliminary consistent estimator of $\bbe$. Taking $\lambda_{s}=\lambda_{n,s}= \iota_n/|\wbeta_{\ini,s}|$ in Theorem \ref{teo:1}, we get consistency of the  ADALASSO estimator when $ \iota_n\to 0$ as $n\to \infty$ as well as that of the adaptive SCAD.

\begin{theorem}\label{teo:1} 
Let $(Y_i,\bZ_i\trasp,\bX_i\trasp)\trasp$ be i.i.d. observations satisfying \eqref{eq:plam} with the errors $\varepsilon_i$ independent from the vector of covariates $(\bZ_i\trasp,\bX_i\trasp)\trasp$. Let $\rho_1$ be a function satisfying \ref{ass:rho_bounded_derivable}. Let $\wbbe$ be the penalized estimator defined in  \eqref{eq:opt-plam}. Asumme \ref{ass:densidadepsilon} to \ref{ass:weta} hold and that $\itJ_{\blach}(\bbe)\convpp 0$. Then, $\wbbe\convpp \bbe$.
\end{theorem}

 In order to obtain   rates of convergence for the estimator of $\bbe$, some additional conditions will be needed. 
 
\begin{enumerate}[resume,label=\textbf{C\arabic*}]
\item\label{ass:psi-dos-der} $\rho_1$ is twice continuously differentiable with second derivative  $\psi_1^{\prime}$   Lipschitz. Furthermore, $\varphi_1(t)=t\,\psi_1^{\prime}(t)$ is bounded and $\esp \psi_1^{\prime}(\epsilon)>0$.
\item\label{ass:zmom2} $\esp\|\bZ\|^2<\infty$  and the matrix $\bV_{\bz}=\esp\bZ\bZ\trasp$ is non-singular.
\item  \label{ass:zmean0} $\esp(\bZ|\bX)=\bcero_q$. 
\item \label{ass:wetatasa} There exists $M>0$ such that $\lim_{n\to \infty} \prob\left(\max_{1\le j\le p} \|\weta_j-\eta_j\|_{\itL_1}\le M\right)=1$, where for brevity,   $\itL_1=\itC^{1}[0,1]$ stands for the space of continuously differentiable functions on $[0,1]$ with   norm
$\|\eta\|_{\itL_1}=\max(\|\eta\|_{\infty},\|\eta^{\prime}\|_{\infty})$.  
\end{enumerate}

\begin{remark}{\label{remark:comentarios2}}
The condition $\esp \psi_1^{\prime}(\epsilon)>0$ in assumption \ref{ass:psi-dos-der} and the non--singularity of $\bV_{\bz}$ required in assumption \ref{ass:zmom2} ensure that a root-$n$ rate may be achieved. 
Regarding assumption \ref{ass:wetatasa}, the requirement  that $\lim_{n\to \infty} \prob\left(\max_{1\le j\le p} \|\weta_j-\eta_j\|_{\itL_1}\le M\right)=1$, for some positive constant $M$, is fulfilled by the estimators proposed in \citet{boente:martinez:2023} (see Lemma A.5 therein).
\end{remark}

\begin{theorem}\label{teo:rate} 
Let $(Y_i,\bZ_i\trasp,\bX_i\trasp)\trasp$ be i.i.d. observations satisfying \eqref{eq:plam} with the errors $\varepsilon_i$ independent from the vector of covariates $(\bZ_i\trasp,\bX_i\trasp)\trasp$. Let $\rho_1$ be a function satisfying \ref{ass:rho_bounded_derivable} and \ref{ass:psi-dos-der} and assume that  \ref{ass:densidadepsilon} to \ref{ass:weta},  \ref{ass:zmom2}, \ref{ass:zmean0} and \ref{ass:wetatasa} hold.

Let $\wbbe$ be the penalized estimator defined in \eqref{eq:opt-plam} with  $\itJ_{\blach}(\bb)=\itJ_{\blach_n}(\bb)=\sum_{s=1}^q p_{\lambda_{n,s}}(|b_s|)$. 
\begin{enumerate}[label={(\alph*)}]
\item Assume that   $p_{\lambda}(\cdot)$ is twice continuously differentiable in $(0, \infty)$, $p_{\lambda}(s)\ge 0$, for any $s\ge 0$, $p_{\lambda}^{\prime} (|\beta_{\ell}|)\ge 0 $ and $p_{\lambda}(0) = 0$. Let  
\begin{align*}
a_n   =\max_{1\le s\le q: \beta_s\neq 0}\{p_{\lambda_{n,s}}^{\prime}(|\beta_{s}|) \}\quad \mbox{and} \quad
 b_n(\nu)  =\mathop{\sup_{1\le s\le q: \beta_s\neq 0}}_{ \tau \in [-1,1]}\{|p_{\lambda_n,s}^{\prime\,\prime}(|\beta_{s}| + \tau \nu)|   \}\,.
 \end{align*}
If for some $\nu$, $b_n(\nu)\convprob 0$ then $\|\wbbe-\bbe\|=O_{\prob}(n^{-1/2}+a_n)$. 
\item Assume that $\itJ_{\blach_n}(\bb)=\iota_n \sum_{s=1}^q |b_s|/ |\wbeta_{\ini,s}|$,  where $\wbbe_{\ini}$ is a preliminary consistent estimator of $\bbe$. Then if $\iota_n\convprob 0$ and $\sqrt{n}\iota_n=O_{\prob}(1)$, we  have that $\|\wbbe-\bbe\|=O_{\prob}(n^{-1/2}).$
\end{enumerate}
 \end{theorem}

\begin{remark}
 Theorem \ref{teo:rate} implies that when using the hard thresholding or the MCP or  SCAD penalty functions, the penalized estimator is root$-n$ consistent requiring only that $\lambda_{n,s}\to 0$. In contrast and as in regression models, when considering the  ADALASSO a rate for the penalty parameter is needed. 
\end{remark}

\subsection{Variable selection property}{\label{sec:seleccion}}
Let consider the last $q-k$ coordinates of $\bbe$ equal to zero, that is, $\bbe=(\bbe_{\act}\trasp,\cero_{q-k}\trasp)\trasp$ where $\bbe_{\act}\in\real^k$ and let  $\wbbe_{\act}$ stand for the first $k$ coordinates of $\wbbe$ and $\wbbe_{\noact}$ for the remaining $q-k$. The following theorem states that the robust estimator has the sparsity property. Again in Theorem \ref{teo:rate2}, we distinguish the case of the  ADALASSO from that of  penalties that achieve a root-$n$ rate requiring only that the penalty parameter converges to $0$.

\begin{theorem}\label{teo:rate2} 
Let $(Y_i,\bZ_i\trasp,\bX_i\trasp)\trasp$ be i.i.d. observations satisfying \eqref{eq:plam} with the errors $\epsilon_i$ independent from   $(\bZ_i\trasp,\bX_i\trasp)\trasp$.  Let $\wbbe$ be the penalized estimator defined in \eqref{eq:opt-plam}, where the function $\rho_1$  satisfies \ref{ass:rho_bounded_derivable} and \ref{ass:psi-dos-der} and that $\itJ_{\blach}(\bb)=\itJ_{\blach_n}(\bb)=\sum_{s=1}^q p_{\lambda_{n,s}}(|b_s|)$ with    $p_{\lambda}(s)\ge 0$, for any $s\ge 0$,  $p_{\lambda}^{\prime} (|\beta_{\ell}|)\ge 0 $ and $p_{\lambda}(0) = 0$.   Assume that $\sqrt{n}\|\wbbe-\bbe\|=O_{\prob}(1)$ and  \ref{ass:densidadepsilon} to \ref{ass:weta},  \ref{ass:zmom2}, \ref{ass:zmean0}  and   \ref{ass:wetatasa} hold.
\begin{enumerate}[label={(\alph*)}]
\item Assume that    $\lambda_{n,s}=\lambda_s\convprob 0$, $\sqrt{n}\min_{1\le s\le q}\lambda_{n,s} \convprob +\infty$ and that for each $C>0$, there exists a constant $K=K_C$ such that, for $ 1\leq s\leq q$,
\begin{equation}
\label{eq:condicionpl}
 \lim_{n\to \infty}\prob\left(p_{\lambda_s}\left(\frac{|u|}{\sqrt{n}}\right)\ge K \lambda_s \frac{|u|}{\sqrt{n}}\;\quad \mbox{for any $|u|\le C$}\right)=1\,.
\end{equation} 
Then, we have that $ \prob\left(\wbbe_{\noact}=\bcero_{q-k}\right)\to 1\,.$ 

In particular, if, for $1\le s\le q$, $\{\lambda_{n,s}\}_{n\ge 1}$ are deterministic sequences of penalty parameters such that $\lambda_{n,s}\to 0$, $\sqrt{n} \lambda_{n,s} \to +\infty$ and $p_{\lambda}(\cdot)$ is   continuously differentiable in $(0, \infty)$ and 
$$\liminf_{n\to\infty}\liminf_{\theta\to 0^+}p_{\lambda_{n,s}}^\prime(\theta)/\lambda_{n,s}>0\,,$$
for $1\leq s\leq q$, we also have $ \prob\left(\wbbe_{\noact}=\bcero_{q-k}\right)\to 1$.

\item Assume that $\itJ_{\blach_n}(\bb)=\iota_n \sum_{s=1}^q |b_s|/ |\wbeta_{\ini,s}|$,  where $\wbbe_{\ini}$ is a preliminary $\sqrt{n}-$consistent estimator of $\bbe$. Then if  $\sqrt{n}\iota_n=O_{\prob}(1)$ and ${n}\,\iota_n \convprob +\infty$ , we  have that $ \prob\left(\wbbe_{\noact}=\bcero_{q-k}\right)\to 1$.
\end{enumerate}
\end{theorem}

\begin{remark}
In the proof of Theorem \ref{teo:rate2}, only the convergence $\sqrt{n}\min_{k+1\le s\le q}\lambda_{n,s} \convprob +\infty$  instead of $\sqrt{n}\min_{1\le s\le q}\lambda_{n,s} \convprob +\infty$  is needed. However, since the number of non--null components is unknown, we require the condition for all  the penalty parameters.

It is worth mentioning that SCAD and MCP penalties satisfy condition \eqref{eq:condicionpl}. Effectively, let us first consider the SCAD penalty. Given $C>0$, as in the proof of Corollary 1 in \citet{Bianco:Boente:Chebi:2022}, taking into account that  $\sqrt{n}\min_{1\le s\le q}\lambda_s \convprob +\infty$, we have that $\prob(\sqrt{n}\min_{1\le s\le q}\lambda_s  >C) \to 1$. Then, if $\itA_{n,C}$ stands for the set $\itA_{n,C}=\{\sqrt{n}\min_{1\le s\le q}\lambda_s  >C\}$, for any $|u|\le C$ we have that  in  $\itA_{n,C}$,  $|u|/\sqrt{n}\le C/\sqrt{n}<\lambda_s$, then
$p_{\lambda_s}\left( |u|/\sqrt{n} \right) =\lambda_s  |u|/\sqrt{n}$, so  \eqref{eq:condicionpl} holds with $K=1$.

When considering the MCP penalty, we use that   $\prob(\itA_{n,C/a}) \to 1$, where $a$ is the fixed constant used in the MCP penalty. Then,  for any $|u|\le C$, in  $\itA_{n,C/a}$,  $|u|/\sqrt{n}\le C/\sqrt{n}< a\, \lambda_s$, so 
$$p_{\lambda_s}\left( |u|/\sqrt{n} \right) =\lambda_s  \frac{|u|}{\sqrt{n}}- \frac{u^2}{2\, a\,n}= \lambda_s  \frac{|u|}{\sqrt{n}}\left(1- \frac{|u|}{2\, a\,\lambda_s\,\sqrt{n}}\right)\,.$$
Using that $C/\sqrt{n}< a\, \lambda_s$, we get that ${|u|}/({2\, a\,\lambda_s\,\sqrt{n}})<1/2$, implying that,  in  $\itA_{n,C/a}$, 
$$p_{\lambda_s}\left( |u|/\sqrt{n} \right) \ge \frac{1}{2}\lambda_s  \frac{|u|}{\sqrt{n}}\,,$$
which entails that \eqref{eq:condicionpl} holds with $K=1/2$.

 For the  ADALASSO penalty, $\itJ_{\blach_n}(\bb)=\iota_n \sum_{s=1}^q |b_s|/ |\wbeta_{\ini,s}|$,  where $\wbbe_{\ini}$ a preliminary consistent estimator of $\bbe$. Thus, with our notation $\lambda_{n,s}= \iota_n / |\wbeta_{\ini,s}|$. Note that
$$p_{\lambda_s}\left( \frac{|u|}{\sqrt{n}} \right) =\frac{\iota_n}{|\wbeta_{\ini,s}|} \frac{|u|}{\sqrt{n}}\,,$$
implies that  condition \eqref{eq:condicionpl}  holds with $K=1$. Besides, to get root-$n$ consistent estimators Theorem \ref{teo:rate} require that $\sqrt{n}\iota_n=O(1)$ in contrast to the requirement   $\sqrt{n}\lambda_{n,s}\to \infty$, for all $1\le s\le q$,  stated in part (a)  of Theorem \ref{teo:rate2}. However, as mentioned above, in the proof the condition $\sqrt{n}\min_{k+1\le s\le q}\lambda_{n,s} \convprob +\infty$ is only needed. For that reason, a root-$n$ consistent estimator is needed when considering the  ADALASSO. In that case, for any $k+1\le s\le q$, we will have that $\sqrt{n}|\wbeta_{\ini,s}|=O_{\prob}(1)$, so  $\sqrt{n}\lambda_{n,s}= n\,\iota_n / |\sqrt{n}\,\wbeta_{\ini,s}|$ will converge to $\infty$ if $n\,\iota_n\to \infty$ as required in (b).  Note also that our statements allow for a random parameter $\iota_n$.
\end{remark}

\section{Monte Carlo Study}{\label{sec:monte}}
This section contains the results of a numerical study designed to compare the robust proposal given in this paper with the corresponding estimator based on least squares, that is, when using $\rho(t)=t^2$ in \eqref{eq:opt-plam}.   For the robust estimator, we considered   the Tukey's bisquare loss function $\rho_c(t)=\min\{1-(1-(t/c)^2)^3,1\}$. The tuning constant $c>0$ balances the robustness and efficiency properties of the associated estimators. For the reported simulation study, we selected the tuning constant as $c=4.685$, which in regression models provides estimators with an 85\% efficiency for Gaussian errors. From now on, we denote the penalized robust procedure proposed in this paper as $\textsc{rob}$, while $\textsc{ls}$ will be used when referring to the approach based on least squares. All computations were carried out in {\texttt{R}} and the code used is available at {\url{https://github.com/alemermartinez/rplam-vs}}.
 
For the preliminary estimators, we considered the proposal of \citet{boente:martinez:2023}. Their approach is based on $B-$splines and as in that paper,  we used  cubic splines and the same number of terms to approximate each additive function. For the robust proposal, the initial estimators also used the Tukey's loss function while for the classical estimator  the squared loss function $\rho(t)=t^2$ is used to compute the initial estimators.

The samples $\{(Y_i,\bZ_i\trasp,\bX_i\trasp)\trasp\}_{i=1}^n$ are generated with the same distribution as $(Y,\bZ\trasp,\bX\trasp)\trasp$, $\bZ=(Z_1,\dots,Z_q)\in\real^q$, $\bX=(X_1,\dots,X_p)\trasp\in\real^p$, with $q=6$ and $p=3$ and sample size $n=400$.  The covariates $\bZ_i=(Z_{i1},\dots,Z_{iq})\trasp$ are generated from a multivariate normal distribution with zero mean and  covariance matrix $\bSi$ with diagonal elements equal to 1 and non--diagonal ones equal to $(\bSi)_{i,\ell}=\mbox{Corr}(Z_{ik},Z_{i\ell})=0.5^{|k-\ell|}$, for $1\leq k,\ell\leq q$. The covariates $X_{ij}$, for $j=1,\dots,p$,  involved in the nonparametric components are uniformly distributed on the interval $(0,1)$   independent from each other and from the vector $\bZ_i$. The distribution of the vector of covariates is similar to the one proposed in Example 2 of \citet{lv:yang:guo:2017}.  The response and the covariates satisfy the partially linear additive model \eqref{eq:plam} with $\mu=0$, $\sigma=1$, $\bbe=(3,-1.5,2,0,0,0)\in\real^6$ and the additive functions are
$$\eta_1(x)=5x-\frac{5}{2},\qquad \eta_2(x)=3(2x-1)^2-1 \quad \mbox{ and }\quad
\eta_3(x)= 60x^3-90x^2+30x\,.$$
In all cases the additive functions are such that $\int_0^1 \eta_j(x)\, dx=0$ for   $j=1,\dots,3$. For clean samples, denoted from now on as $C_0$, the error's distribution is $\varepsilon\sim N(0,1)$.

In order to study the effect of atypical data on the estimators, four different contamination schemes  were considered.  They are characterized  in terms  of the errors $ \epsilon$ and the regression covariates as follows:
\begin{itemize}
\item $C_1$\,:\, $\varepsilon \sim t_3$
\item $C_2$\,:\, $\varepsilon \sim 0.95 N(0,1)+0.05 N(0,100)$.  This contamination corresponds to inflating the error's variance and will only affect the variability of the estimators. 
\item $C_3$\,:\, $\varepsilon \sim 0.85 N(0,1)+0.15 N(15,1)$. This contamination corresponds to vertical outliers. 
\item $C_4$\,:\, In this contamination scheme, which corresponds to bad leverage points, 5\% of the covariates $\bZ_i$ were randomly replaced by $(20,\dots,20)\in\real^6$ without changing the responses obtained under $C_0$.
\end{itemize}
 Contamination $C_1$ corresponds to the case of heavy-tailed errors. As mentioned above, $C_2$ is a variance contamination setting, while $C_3$ is a bias contamination scheme where 15\% of the errors has another normal distribution with  center shifted giving raise to vertical outliers. In scenario $C_4$ the high-leverage points are introduced aiming to affect the estimation of the regression parameter.

Different measures were used for determining the effectiveness in the variable selection results. More precisely, for each sample, we calculated
\begin{itemize}
\item[] \textsc{CN}$_0$: the number of zero components correctly estimated to be zero,
\item[] \textsc{IN}$_0$: the number of non-zero components incorrectly estimated to be zero, and
\item[] \textsc{CF}: that takes 1 if the correct variables are selected and $0$ otherwise.
\end{itemize}
The averages over replications are reported in the Tables and Figures below. Let observe that, for each replication, the measure \textsc{CN}$_0$  belongs to the set $\{0,1,2,3\}$ since there are three zero components and so the closer  the average of this measure is to 3, the better. Similarly, for each replication, the \textsc{IN}$_0$ number also belongs to the set $\{0,1,2,3\}$, since there are also three non-zero components, but, in this case, the closer this measure is to 0, the better. Finally, the \textsc{CF} measure reveals the proportion of times the correct model is selected.

To evaluate the performance of the parametric components, as in \citet{lv:yang:guo:2017}, for each replication, we computed the generalized mean square error (GMSE) defined as $\mbox{GMSE}=(\wbbe-\bbe)\trasp \bSi (\wbbe-\bbe)$, where $\bSi\in\real^{q\times q}$ is the true  covariance matrix of $\bZ$. To avoid possible large values of this measure at some replications, instead of reporting the mean over replications, we report the median over the 500 replications.

This last measure was also considered for the oracle estimators. The oracle estimator is the estimator that \lq\lq knows which are the true non-zero components\rq\rq\, and so it does not need to select variables, meaning that we are considering  the  partially linear additive model \eqref{eq:plam} as above, but the regression covariates correspond to the first three covariates  $(Z_{i1},Z_{i2},Z_{i3})\in\real^3$ and the regression parameter equals $(3,-1.5,2)\trasp$. Hence, for this model the true  covariance matrix of the covariates  belong to $\real^{3\times 3}$ and corresponds to a block of the matrix $\bSi$ defined above. The robust oracle estimators were computed using the approach in \citet{boente:martinez:2023} with the Tukey's loss function with tuning constant $c=4.685$, while the classical ones correspond to choosing  $\rho(t)=t^2$. From now on, we denote OGMSE the results of the GMSE measure for the oracle estimators.

In this numerical study, we select as penalty function  the SCAD penalty taking $a=3.7$, which is the usual value for $a$. At each replication, the penalty parameters $\bla$ for the robust penalized estimator were obtained minimizing  the $RBIC$ criterion defined in Section \ref{sec:RBIC} over the grid $\Lambda=\{\bla=(\lambda_1,\dots,\lambda_6)\trasp\,:\, \lambda_j\in\{0,0.2,0.4,0.6\}, 1\leq j\leq 6\}$ that contains $4096$ vectors.  Over the same grid but with the square loss, the $BIC$ criteria was minimized to obtain the penalty parameters for the classical estimator.

\begin{table}[ht]
\centering
\begin{tabular}{ccccc}
  \hline
 & \textsc{Method} & \textsc{CN}$_0$ & \textsc{IN}$_0$ & \textsc{CF} \\ 
  \hline
$C_0$ & \textsc{ls} & 2.98 \textcolor{gray}{(0.13)} & 0.00 \textcolor{gray}{(0.00)} & 0.98 \textcolor{gray}{(0.13)}\\ 
 & \textsc{rob} & 2.99 \textcolor{gray}{(0.08)} & 0.00 \textcolor{gray}{(0.06)} & 0.99 \textcolor{gray}{(0.09)} \\ \hline
 $C_1$ & \textsc{ls} & 2.85 \textcolor{gray}{(0.35)} & 0.00 \textcolor{gray}{(0.00)} & 0.85 \textcolor{gray}{(0.35)} \\ 
& \textsc{rob} & 3.00 \textcolor{gray}{(0.06)} & 0.01 \textcolor{gray}{(0.12)} & 0.98 \textcolor{gray}{(0.13)} \\ \hline
$C_2$ & \textsc{ls} & 2.60 \textcolor{gray}{(0.53)} & 0.00 \textcolor{gray}{(0.00)} & 0.61 \textcolor{gray}{(0.49)}\\ 
 & \textsc{rob} & 2.99 \textcolor{gray}{(0.08)} & 0.01 \textcolor{gray}{(0.11)} & 0.99 \textcolor{gray}{(0.10)} \\ \hline
$C_3$ & \textsc{ls} & 2.37 \textcolor{gray}{(0.57)} & 0.00 \textcolor{gray}{(0.00)} & 0.41 \textcolor{gray}{(0.49)}\\ 
 & \textsc{rob} & 3.00 \textcolor{gray}{(0.04)} & 0.01 \textcolor{gray}{(0.13)} & 0.99 \textcolor{gray}{(0.11)} \\ \hline
$C_4$ & \textsc{ls} & 0.61 \textcolor{gray}{(0.50)} & 0.00 \textcolor{gray}{(0.00)} & 0.00 \textcolor{gray}{(0.00)} \\ 
& \textsc{rob} & 3.00 \textcolor{gray}{(0.06)} & 0.02 \textcolor{gray}{(0.13)} & 0.98 \textcolor{gray}{(0.13)} \\ 
   \hline
\end{tabular}
\caption{\label{tab:selection} Mean over replications of the measures for variable selection when considering  the   least squares ($\textsc{ls}$) estimators and its robust counterpart ($\textsc{rob}$).  Standard deviations are reported between brackets.} 
\end{table}
 
Table \ref{tab:selection} reports the mean over replications of the three measures  \textsc{CN}$_0$, \textsc{IN}$_0$ and \textsc{CF}, for both the robust and least-squares estimators. Standard deviations are reported between brackets. Recalling that \textsc{CN}$_0$ is the average number of zero components correctly estimated as zero and so the closer to $3$, the better, under $C_0$ both the robust and classical estimators have a good performance showing similar values close  to $3$. In contrast, for the  contamination schemes $C_1$ to $C_4$, the least squares approach leads to   smaller means of the \textsc{CN}$_0$  than  the robust proposal, which behaves similarly across all the contamination settings. It is worth mentioning that under $C_4$ the least squares estimators break  down and have a poor variable selection capability. The bad behaviour of the classical estimator is also reflected on the larger standard deviations obtained when contaminating data arise. 

Regarding the results for the \textsc{IN}$_0$ measure, which corresponds to the proportion of times the non-zero components are detected as zero,  in all scenarios the obtained results for both estimators are $0$  or close to $0$. It is worth mentioning that,  in a few samples, the robust procedure   detects as zero some components that correspond to the first three elements of $\bbe$. Finally, similar conclusions to those obtained with  \textsc{CN}$_0$ are valid when considering  the summary measure \textsc{CF}, which corresponds to the proportion of times that the true model is selected.  Effectively, for clean samples, both estimators  have values of \textsc{CF} close to $1$ and under all contaminations, the robust proposal still provides reliable results close to 1. In contrast, for the classical method  much smaller averages are obtained. In particular, under $C_4$, the least squares estimator never detects the true model, since the mean  of the \textsc{CF} value equal  $0$.  
This effect also becomes evident in   Figure  \ref{fig:C-measure-barplot} which  displays the barplots for the \textsc{CN}$_0$ and \textsc{CF} measures, in panels (a) and (b) respectively. Red bars correspond to the least-square estimator and blue bars to the robust proposal. Both plots have also an indication of the corresponding threshold. As it was already mentioned, while all the robust approaches show values close to the thresholds, the $\textsc{ls}-$ estimators perform poorly under the contamination settings, especially under $C_4$.

\begin{figure}[ht!]
\begin{center}
\begin{tabular}{cc}
(a) & (b) \\[-4ex]
\includegraphics[scale=0.4]{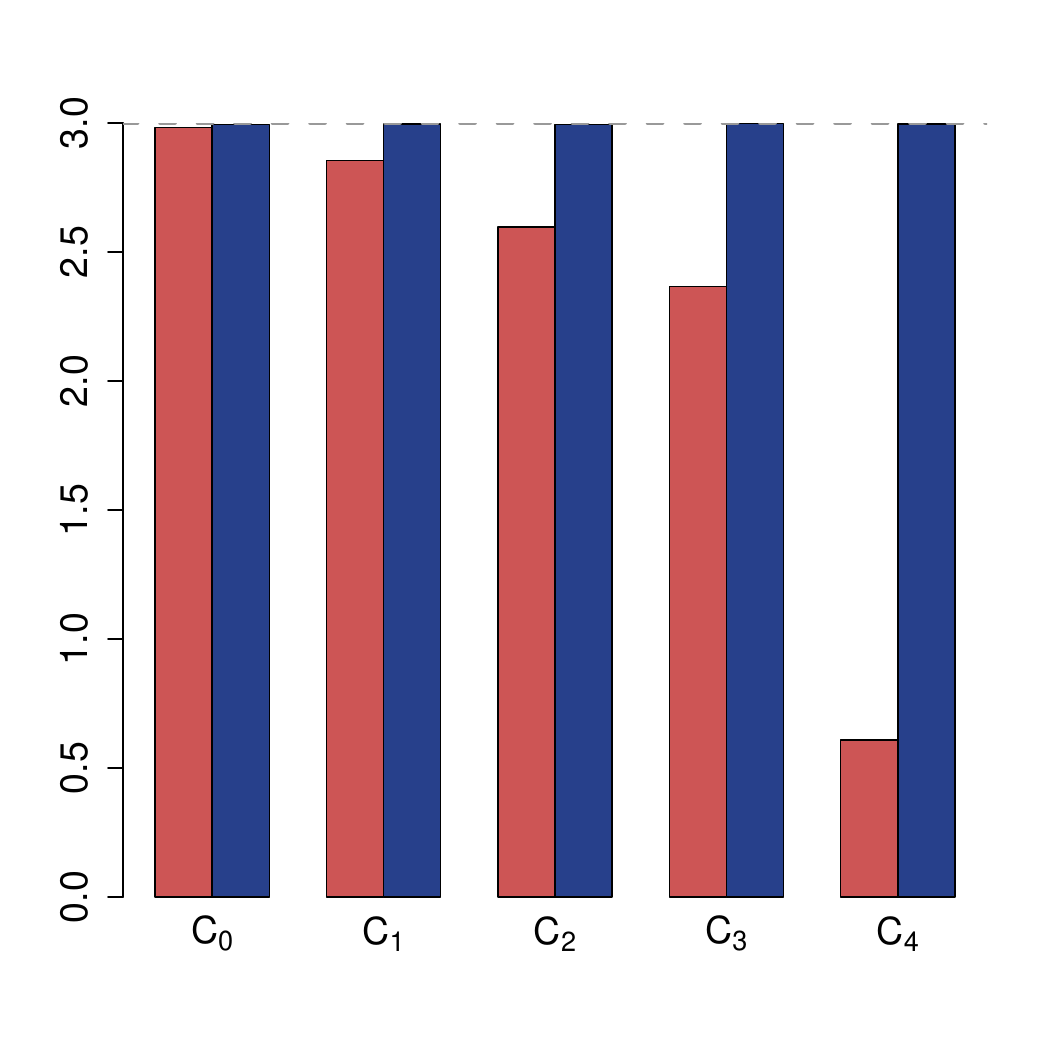} & \includegraphics[scale=0.4]{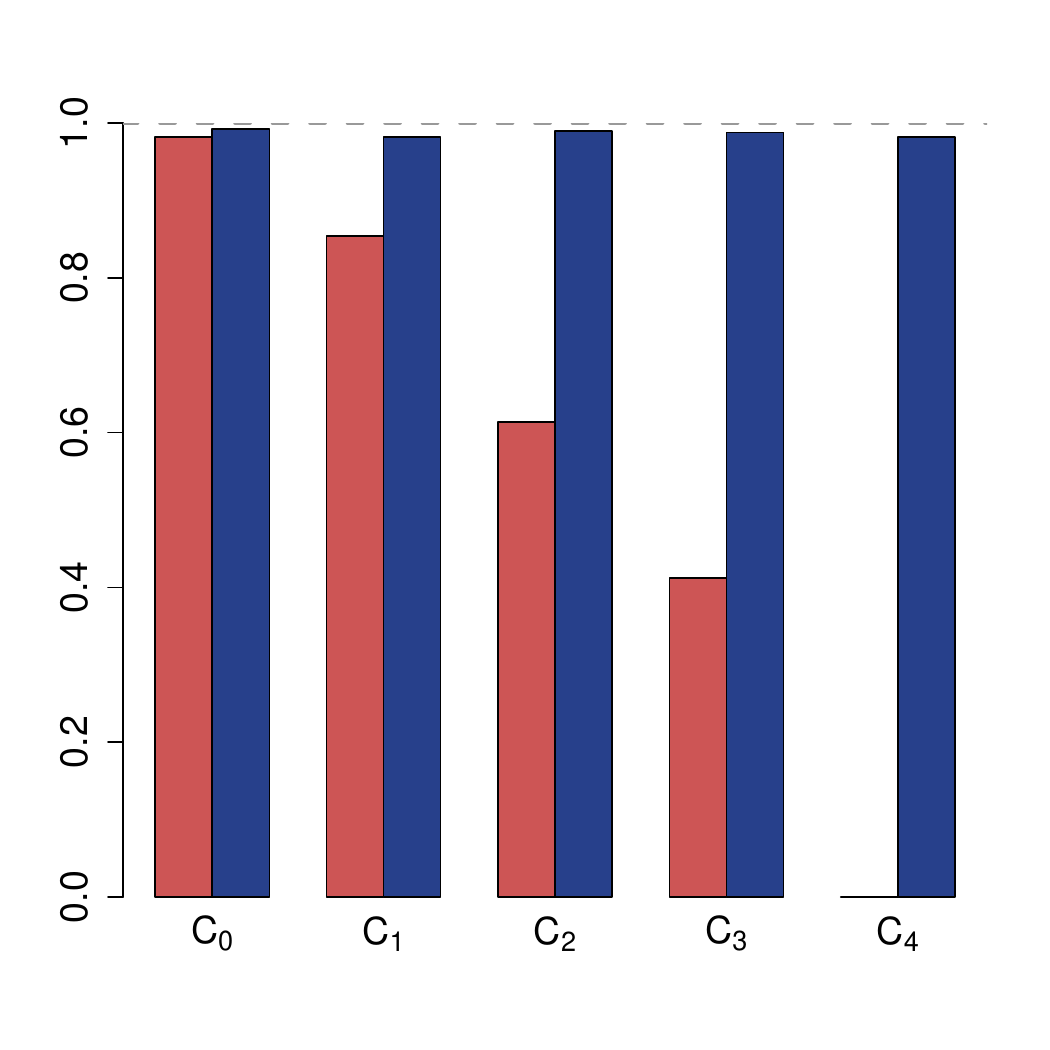}
\end{tabular}
\vskip-0.2in
\caption{\label{fig:C-measure-barplot} Panels (a) and (b) display the plot of the \textsc{CN}$_0$   and \textsc{CF} values,  across all the contamination settings, respectively. The blue bars correspond to the results obtained with the robust proposal and the red ones  to those of the $\textsc{ls}-$estimator.}
\end{center}
\end{figure}

\begin{table}[ht]
\centering
\begin{tabular}{cccc}
  \hline
 & \textsc{Method} & OGMSE  & GMSE \\ 
  \hline
$C_0$ & \textsc{ls} & 0.006 \textcolor{gray}{(0.00)} & 0.011 \textcolor{gray}{(0.01)}\\ 
 & \textsc{rob} & 0.006 \textcolor{gray}{(0.01)} & 0.012 \textcolor{gray}{(0.01)} \\ \hline
 $C_1$ & \textsc{ls} & 0.016 \textcolor{gray}{(0.01)} & 0.032 \textcolor{gray}{(0.03)} \\ 
& \textsc{rob} & 0.010 \textcolor{gray}{(0.01)} & 0.011 \textcolor{gray}{(0.01)} \\ \hline
$C_2$ & \textsc{ls} & 0.036 \textcolor{gray}{(0.03)} & 0.068 \textcolor{gray}{(0.06)} \\ 
 & \textsc{rob} & 0.007 \textcolor{gray}{(0.01)} & 0.011 \textcolor{gray}{(0.01)} \\ \hline
$C_3$ & \textsc{ls} & 0.064 \textcolor{gray}{(0.06)} & 0.132 \textcolor{gray}{(0.12)} \\ 
 & \textsc{rob} & 0.006 \textcolor{gray}{(0.01)} & 0.011 \textcolor{gray}{(0.01)} \\ \hline
$C_4$ & \textsc{ls} & 6.965 \textcolor{gray}{(0.16)} & 4.599 \textcolor{gray}{(0.15)}\\ 
& \textsc{rob} & 0.007 \textcolor{gray}{(0.01)} & 0.009 \textcolor{gray}{(0.01)} \\ 
   \hline
\end{tabular}
\caption{\label{tab:estimation} Median over replications  of the GMSE and OGMSE,  for the classical  ($\textsc{ls}$) and  robust procedures ($\textsc{rob}$). Between brackets,   the \textsc{mad} is reported.} 
\end{table}

Table \ref{tab:estimation} contains summary measures  for the  generalized mean square error (GMSE). More precisely, we report  the median over replications   of the GMSE and     the  \textsc{mad}  between brackets. Similar summary measures  are given for the OGMSE. For clean samples, both  approaches show similar results and the obtained results for the GMSE are approximately twice those obtained for   oracle estimators.  This effect which is due to the penalizing process, might be explained with the grid used for selecting the regularization parameters. It is also worth mentioning that the GMSE of the robust estimator is larger than that of the least squares one, due to its loss of efficiency which is related to the ratio $\esp \psi^2(\epsilon)\,\{\esp \psi^{\prime}(\epsilon)\}^{-2}$ and is not included in the expression of the GMSE.
For the contamination schemes $C_1$ to $C_4$,   similar values of GMSE than  those under $C_0$ are obtained for the the robust proposal. In contrast,  the least squares estimator increases the median over replications  in about 3, 6, 12 and 418 times, respectively. As it is expected, these disastrous performance is also observed for the oracle estimators.  In particular,   high leverage outliers have a damaging effect  on the classical regression estimators. 

\begin{figure}[ht!]
\begin{center}
\includegraphics[scale=0.4]{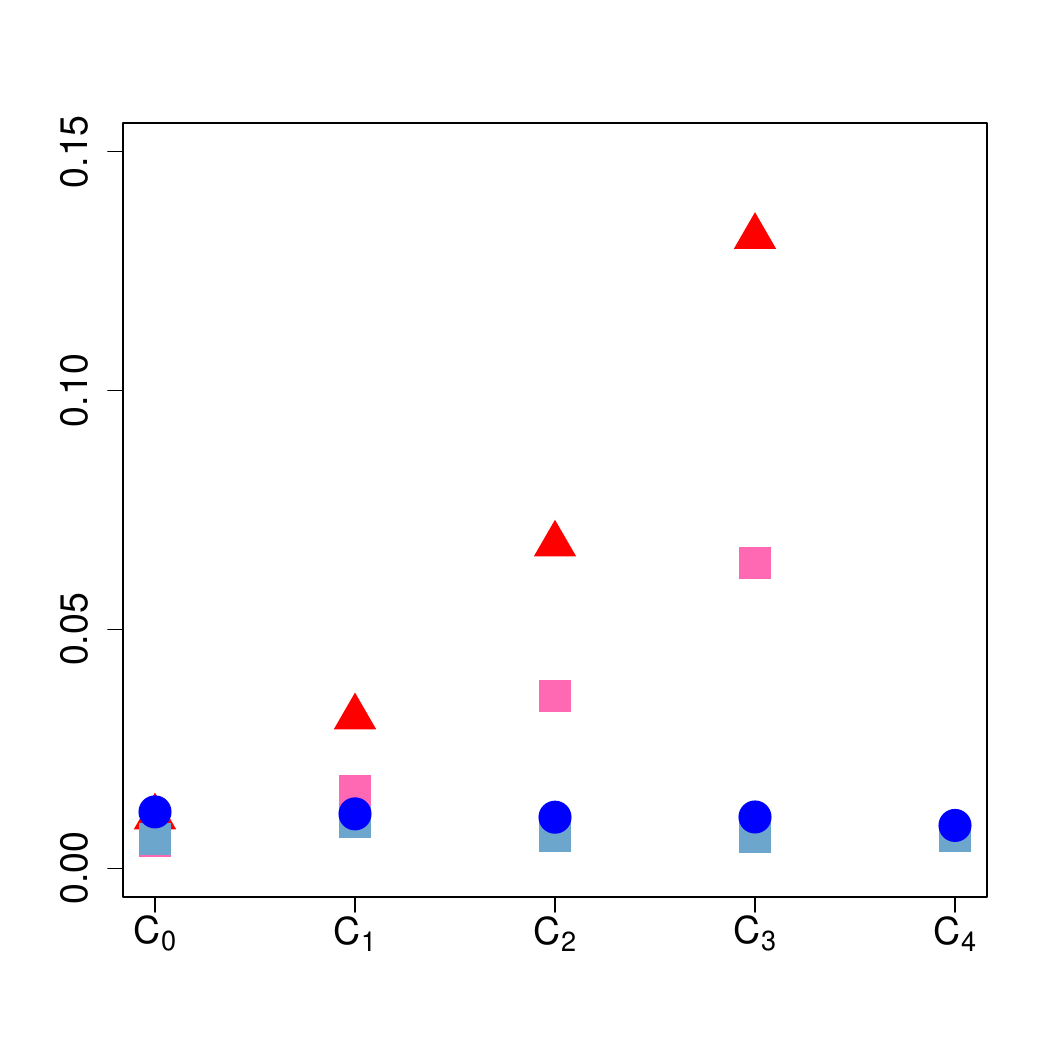}
\vskip-0.2in
\caption{\label{fig:GMSE-measure} Plot of the median over replications of the GMSE  for the penalized and oracle estimators across all the contamination cases. The red triangles and blue circles correspond to the penalized $\textsc{ls}-$estimator  and the robust counterpart, respectively,  while the pink and light blue squares   identify the medians of the  the oracle least-square and   oracle robust estimator, respectively.}
\end{center}
\end{figure}

To visualize the effect of contaminations, Figure \ref{fig:GMSE-measure} displays in blue circles and red triangles the results for the robust and least squares estimators, respectively, together with their corresponding oracle versions in light blue and pink squares. The stable behaviour of the robust approaches becomes evident, since the obtained values are  at the bottom of the plot. In contrast, the least squares estimator presents increased values of the GMSE  under $C_1$ to $C_3$, the values obtained under $C_4$ are beyond the limits of the plot.

\begin{figure}[tp]
\begin{center}
\begin{tabular}{cc}
$C_0$ & $C_1$\\[-2ex]
\includegraphics[scale=0.4]{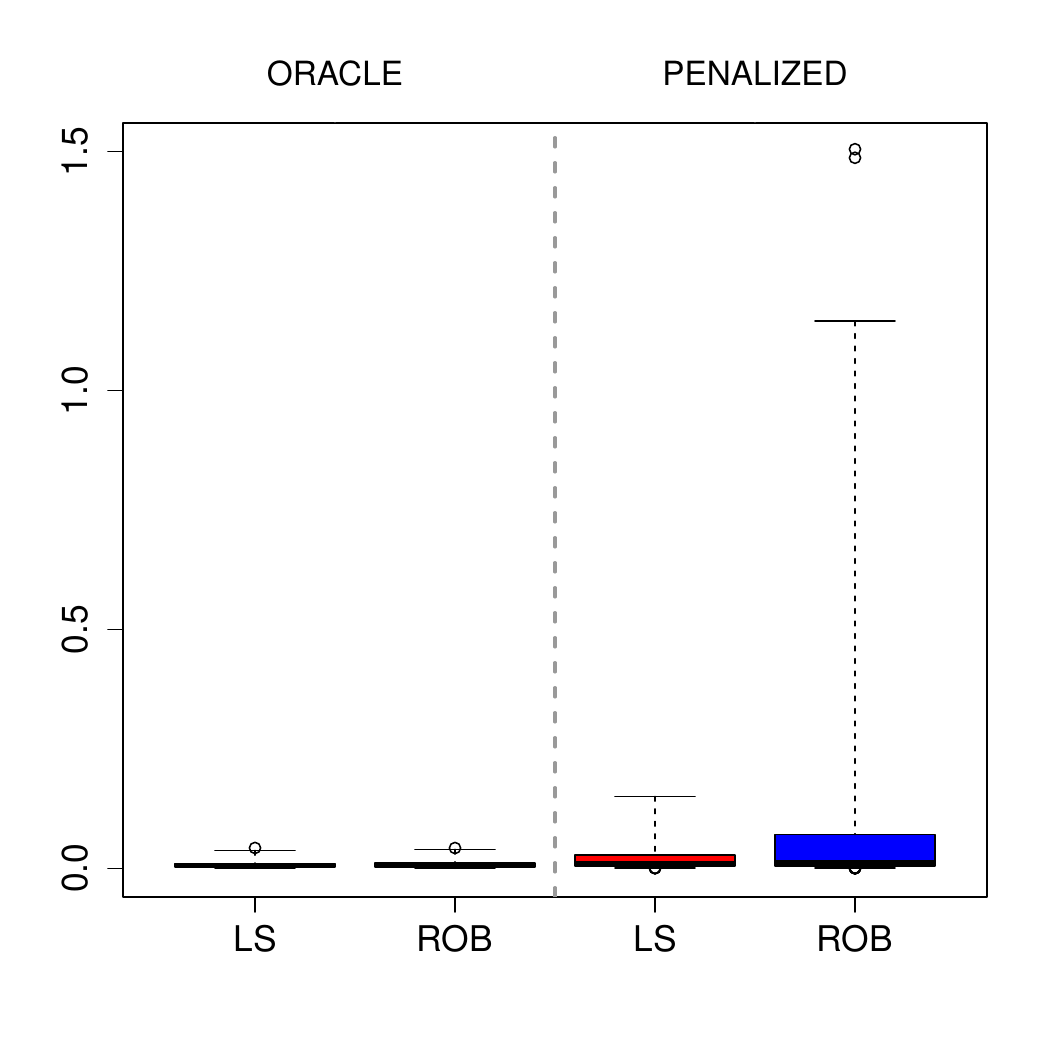} & \includegraphics[scale=0.4]{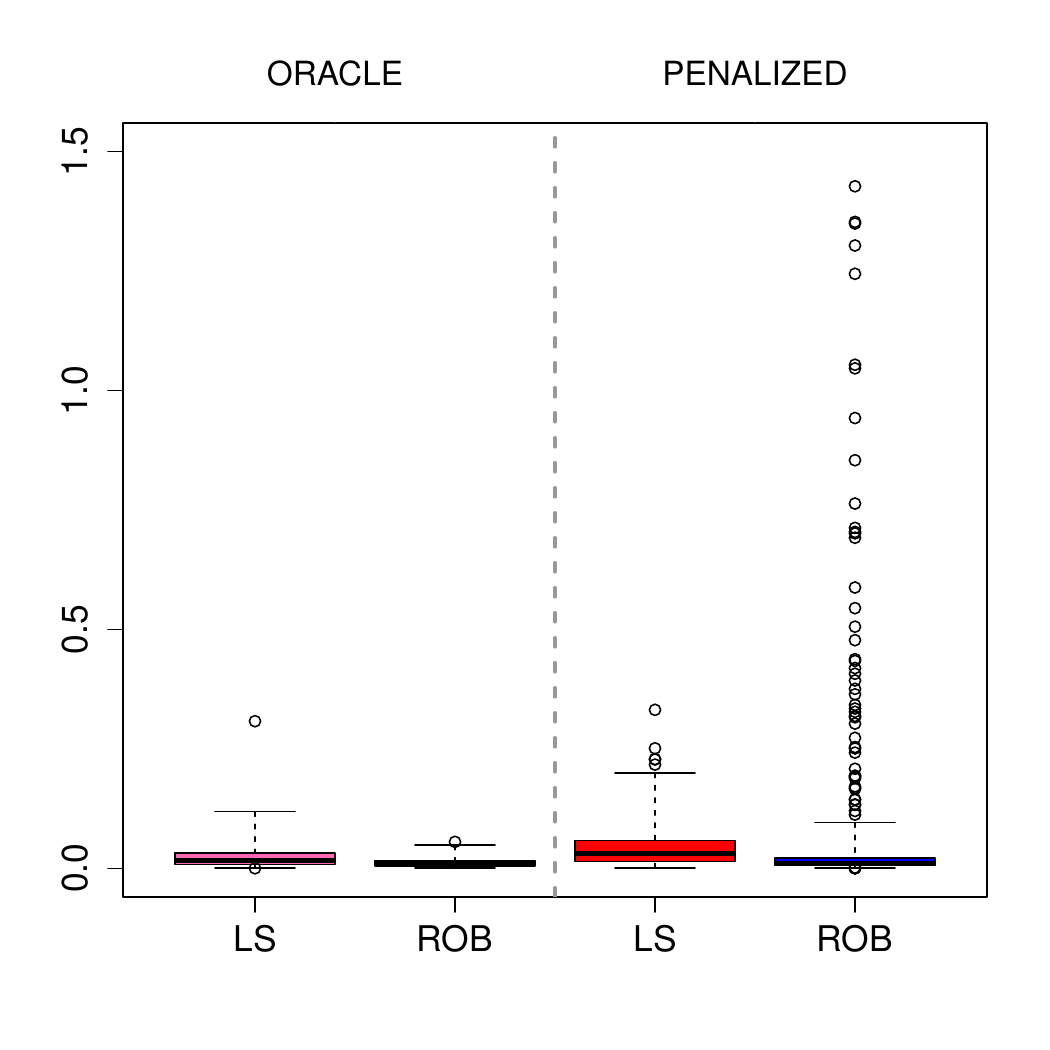}\\
$C_2$ & $C_3$\\[-2ex]
\includegraphics[scale=0.4]{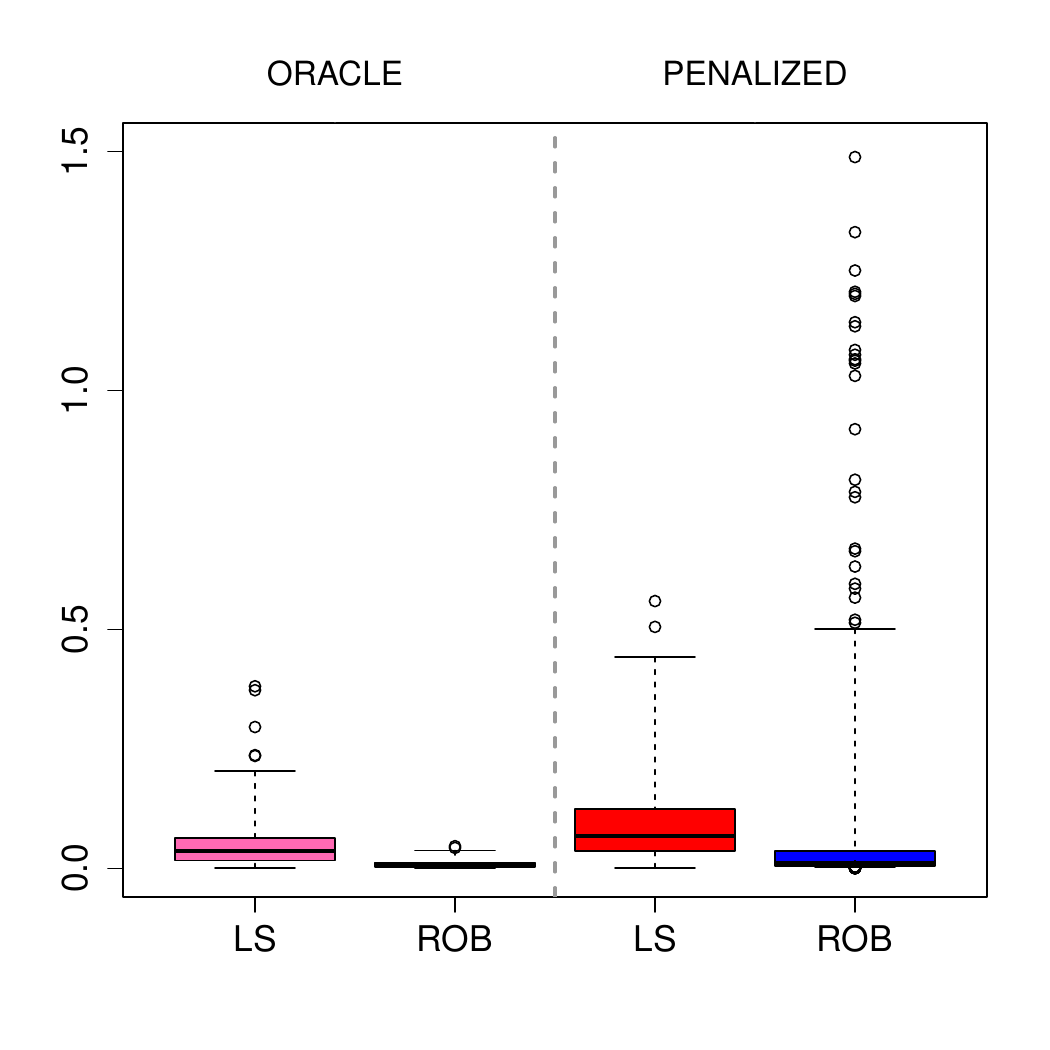} & \includegraphics[scale=0.4]{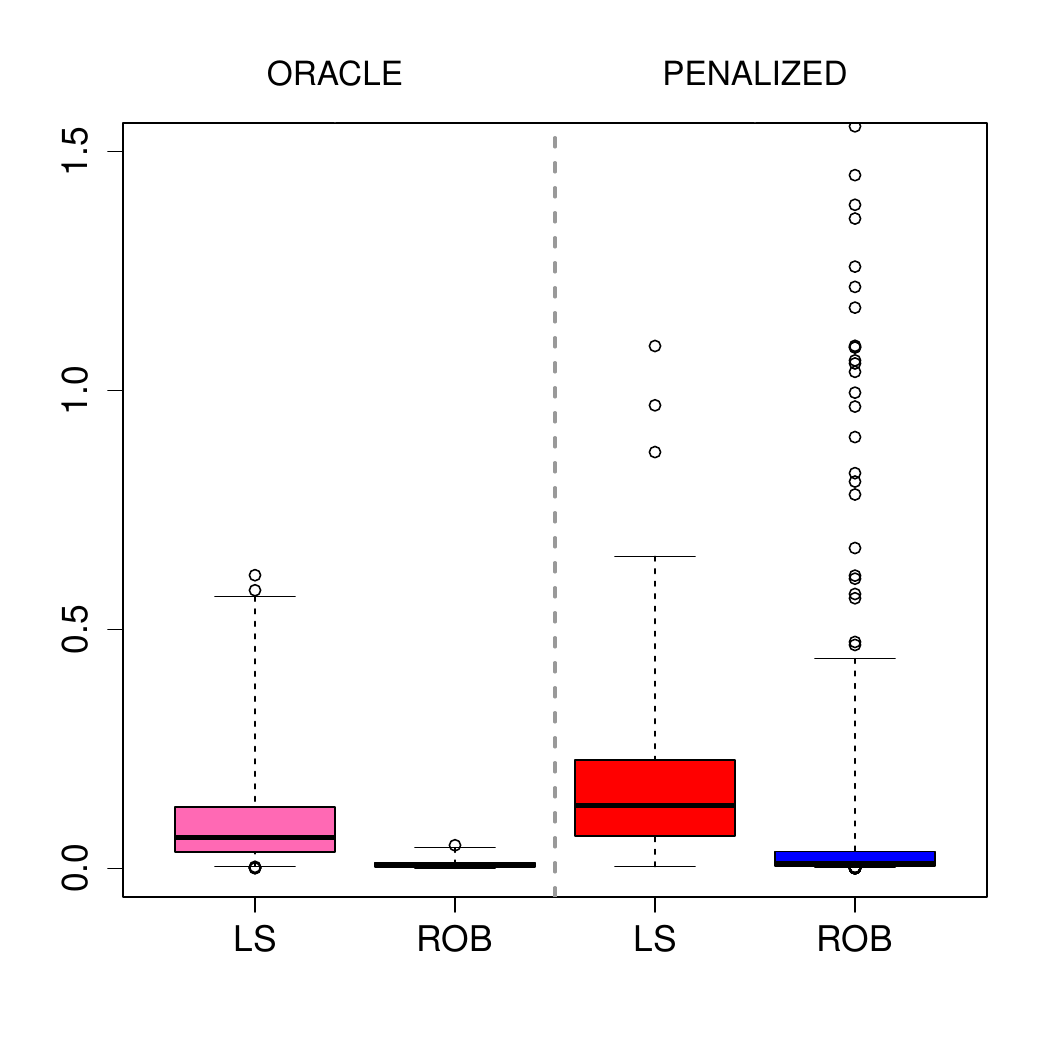}\\
$C_4$ \\[-2ex]
\includegraphics[scale=0.4]{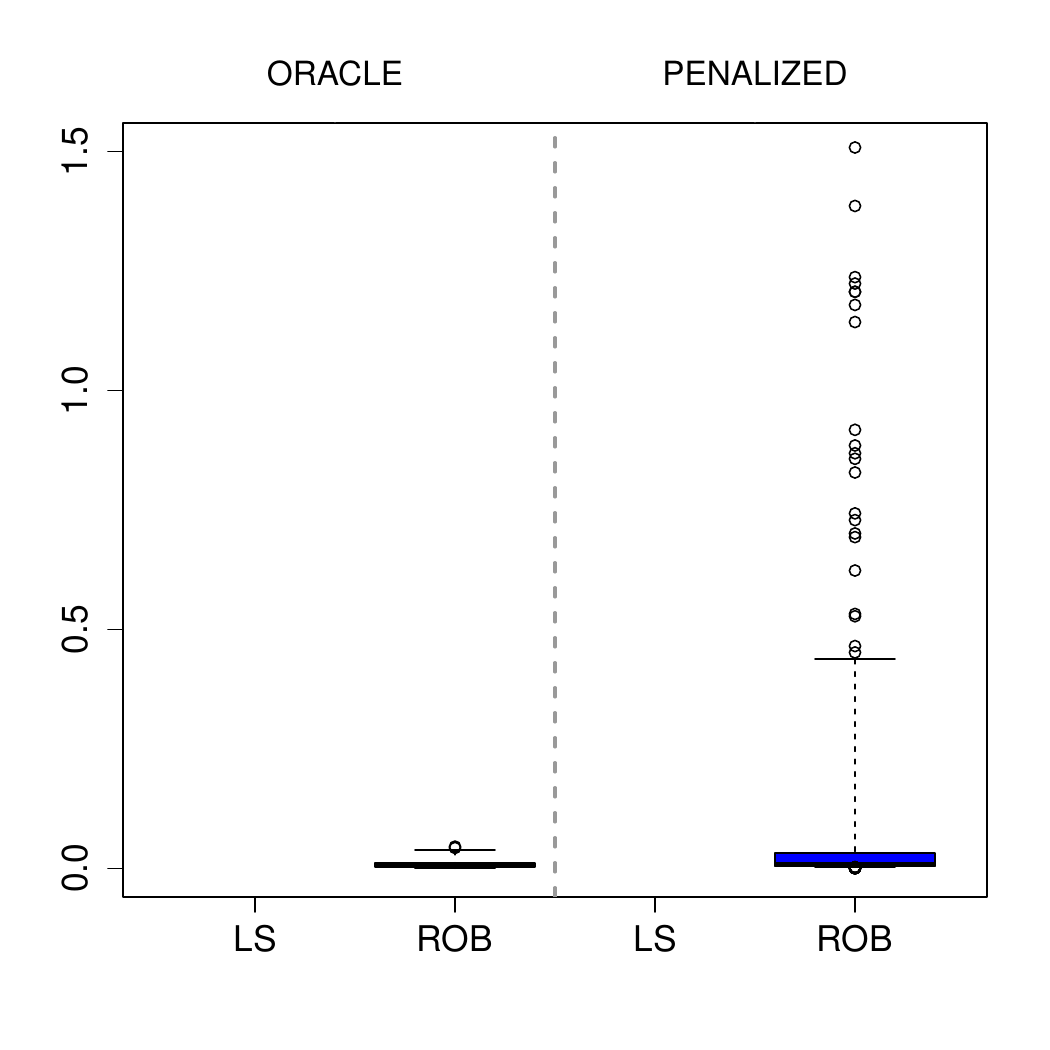} & 
\end{tabular}
\vskip-0.2in
\caption{\label{fig:GMSE-adjboxplots} Adjusted boxplots of the GMSE and OGMSE for the penalized and oracle estimators across all the contamination cases.}
\end{center}
\end{figure}

Figure \ref{fig:GMSE-adjboxplots}   presents the adjusted boxplots of the 500 GMSE values obtained for the penalized and oracle estimators under the contamination schemes. All plots have the same vertical axis to facilitate comparisons. Adjusted boxplots were introduced by \citet{Hubert:Vandervieren:2008} as a visualization tool similar to the boxplot but adapted to skewed data. As   seen in Table \ref{tab:estimation}, under for clean samples, both oracle estimators perform similarly, while the penalized estimators have larger dispersion than their oracle counterparts. Besides, the robust penalized estimator presents a few outliers and as expected, a    wider   box. Under $C_1$ to $C_4$, even though for some replications, large values of the GMSE are observed for the robust penalized estimator, the boxes still show lower values of the GMSE than for both $\textsc{ls}-$estimators. It is worth mentioning that, under $C_4$,   the GMSE obtained when considering   the $\textsc{ls}$   procedure explode leading to values so large that the boxplots cannot  be seen in the figure.

%
%

Finally, Table \ref{tab:betas} presents the proportion of times each component was detected as a zero component. Recall that the first three components are the non-zero ones. Both approaches behave similarly, never or almost never detecting them as zero under all the contamination schemes. In contrast, the last three components, which are the zero ones, should have proportions  close to $1$, but,  when high leverage points are present, the  results obtained for the least squares estimators  are poor, meaning that the classical procedure is not able to identify them as $0$. These results are in the same direction as those reported in Table \ref{tab:selection}. Only the robust estimator detects these variables as zero under all the contamination scenarios.  

\begin{table}[ht]
\centering
\begin{tabular}{cccccccc}
  \hline
& \textsc{Method} & $\beta_1$ & $\beta_2$ & $\beta_3$ & $\beta_4$ & $\beta_5$ & $\beta_6 $ \\ 
  \hline
 $C_0$ & \textsc{ls} & 0.00 & 0.00 & 0.00 & 1.00 & 0.99 & 0.99 \\ 
& \textsc{rob} & 0.00 & 0.00 & 0.00 & 0.99 & 1.00 & 1.00 \\ \hline
$C_1$ & \textsc{ls} & 0.00 & 0.00 & 0.00 & 0.97 & 0.94 & 0.94 \\ 
& \textsc{rob} & 0.00 & 0.01 & 0.00 & 1.00 & 1.00 & 1.00 \\ \hline
$C_2$ & \textsc{ls} & 0.00 & 0.00 & 0.00 & 0.90 & 0.85 & 0.84 \\ 
& \textsc{rob}  & 0.00 & 0.01 & 0.00 & 1.00 & 1.00 & 1.00 \\ \hline
$C_3$ & \textsc{ls} & 0.00 & 0.00 & 0.00 & 0.81 & 0.79 & 0.77 \\ 
& \textsc{rob} & 0.00 & 0.01 & 0.00 & 1.00 & 1.00 & 1.00 \\ \hline
$C_4$ & \textsc{ls} & 0.00 & 0.00 & 0.00 & 0.42 & 0.18 & 0.00 \\ 
& \textsc{rob} & 0.00 & 0.02 & 0.00 & 1.00 & 1.00 & 1.00 \\ 
   \hline
\end{tabular}
\caption{\label{tab:betas} Proportion of times each coefficient is estimated to be zero.} 
\end{table}

In order to have an insight of how the $RBIC$ and the $BIC$ select the regularization parameters over the grid $\Lambda$, for the two contaminations cases $C_0$ and $C_4$, Figure \ref{fig:lambdas-c0-c6} displays the pie charts of the proportion of times each value in the set $\Lambda$ was selected. The gray, purple, blue and pink  zones correspond to the values $0$, $0.2$, $0.4$ and   $0.6$, respectively. It can be appreciated that, under $C_0$, the least-square estimator shows larger areas for the value $0$ when considering the parameters $\lambda_1, \lambda_2$ and $\lambda_3$ which are those related to the non-zero components and larger areas in pink and blue corresponding to non-zero values of the penalty candidates,  for $\lambda_4,\lambda_5$ and $\lambda_6$. Note that these three parameters correspond to the null components.

\begin{figure}[ht!]
 \begin{center}
\small
	 \renewcommand{\arraystretch}{0.4}
 \newcolumntype{M}{>{\centering\arraybackslash}m{\dimexpr.1\linewidth-1\tabcolsep}}
   \newcolumntype{G}{>{\centering\arraybackslash}m{\dimexpr.2\linewidth-1\tabcolsep}}
\begin{tabular}{M  GG  GG}
\multicolumn{1}{c}{} & \multicolumn{2}{c}{$C_0$} & \multicolumn{2}{c}{$C_4$}\\ [2ex]
 & \multicolumn{1}{c}{\textsc{ls}} & \multicolumn{1}{c}{\textsc{rob}}  & \multicolumn{1}{c}{\textsc{ls}}  & \multicolumn{1}{c}{\textsc{rob}}\\ [-1ex]
$\lambda_1$ & 
\includegraphics[scale=0.2]{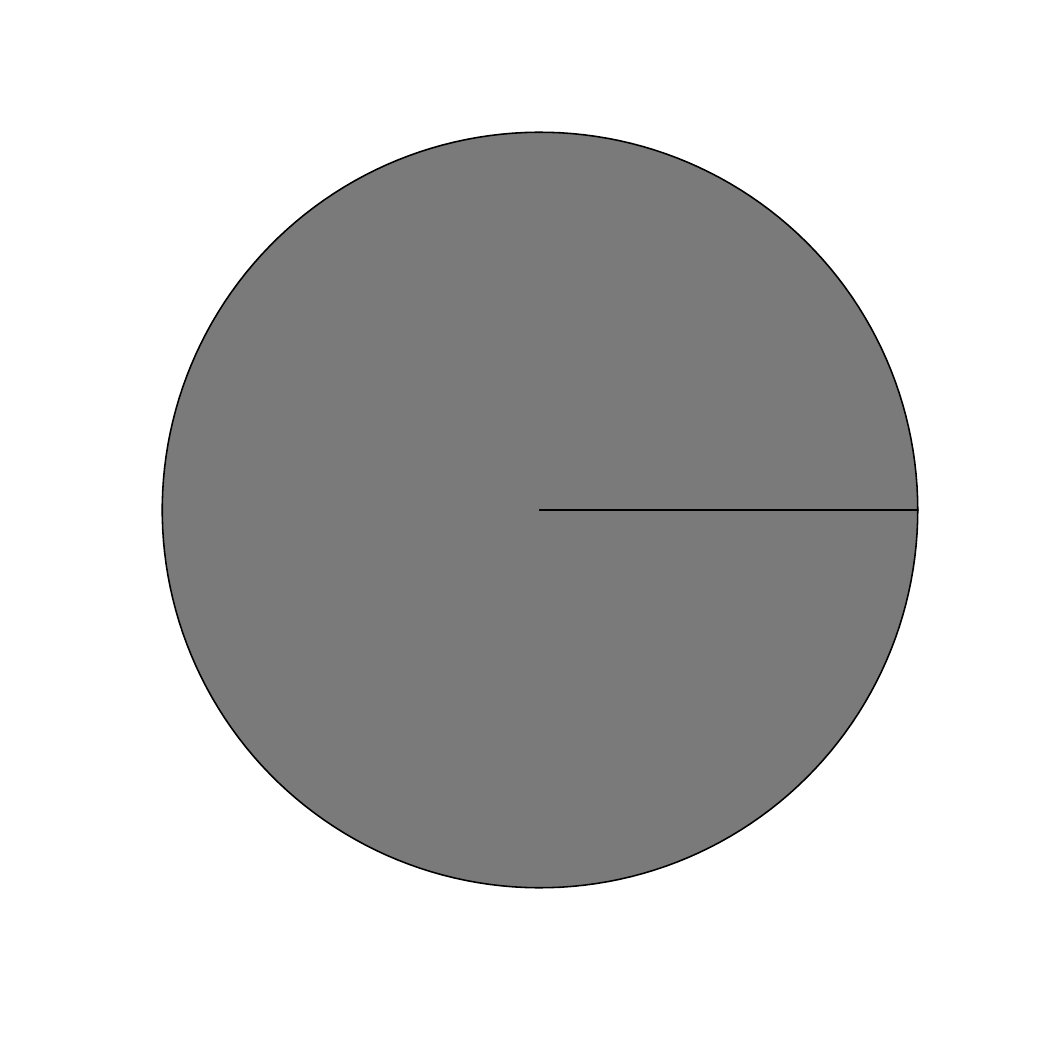} & 
\includegraphics[scale=0.2]{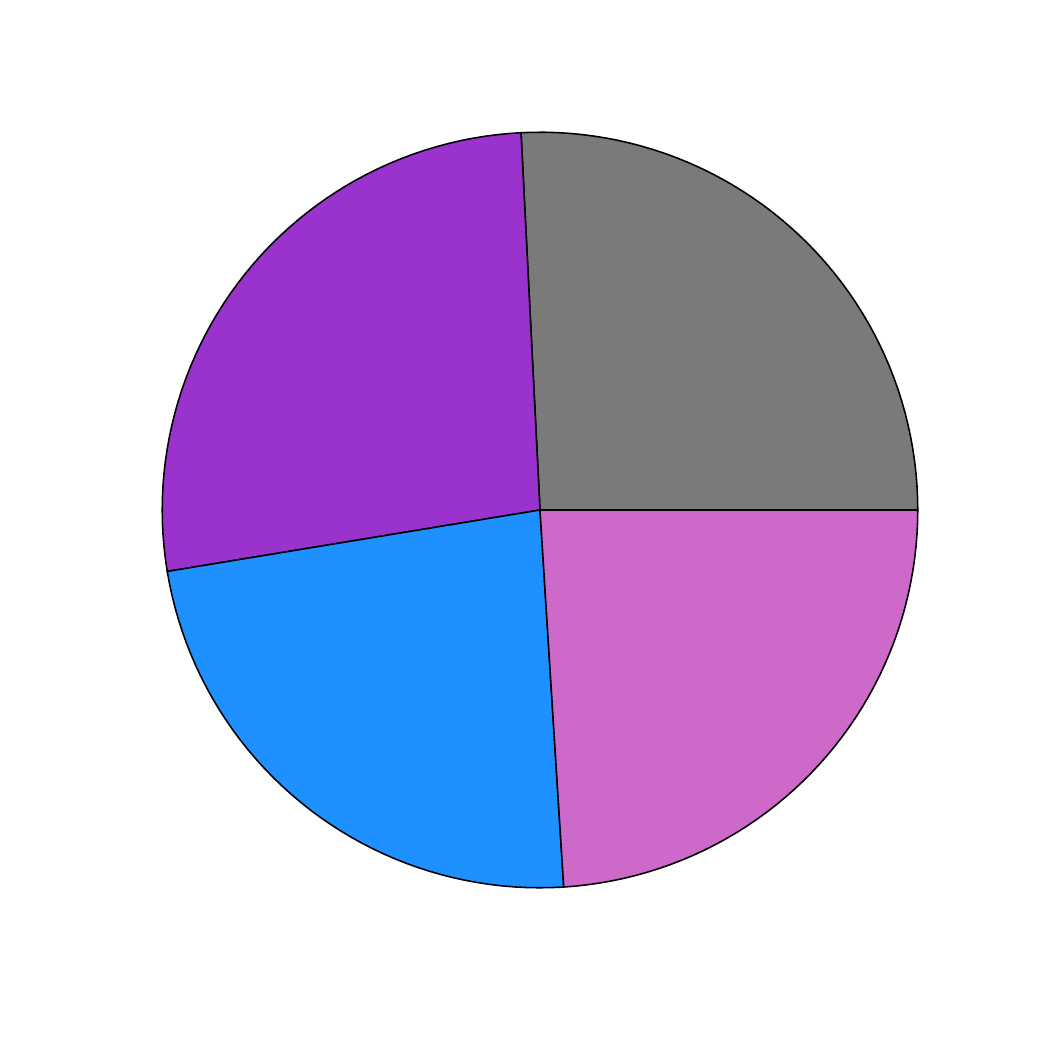} & 
\includegraphics[scale=0.2]{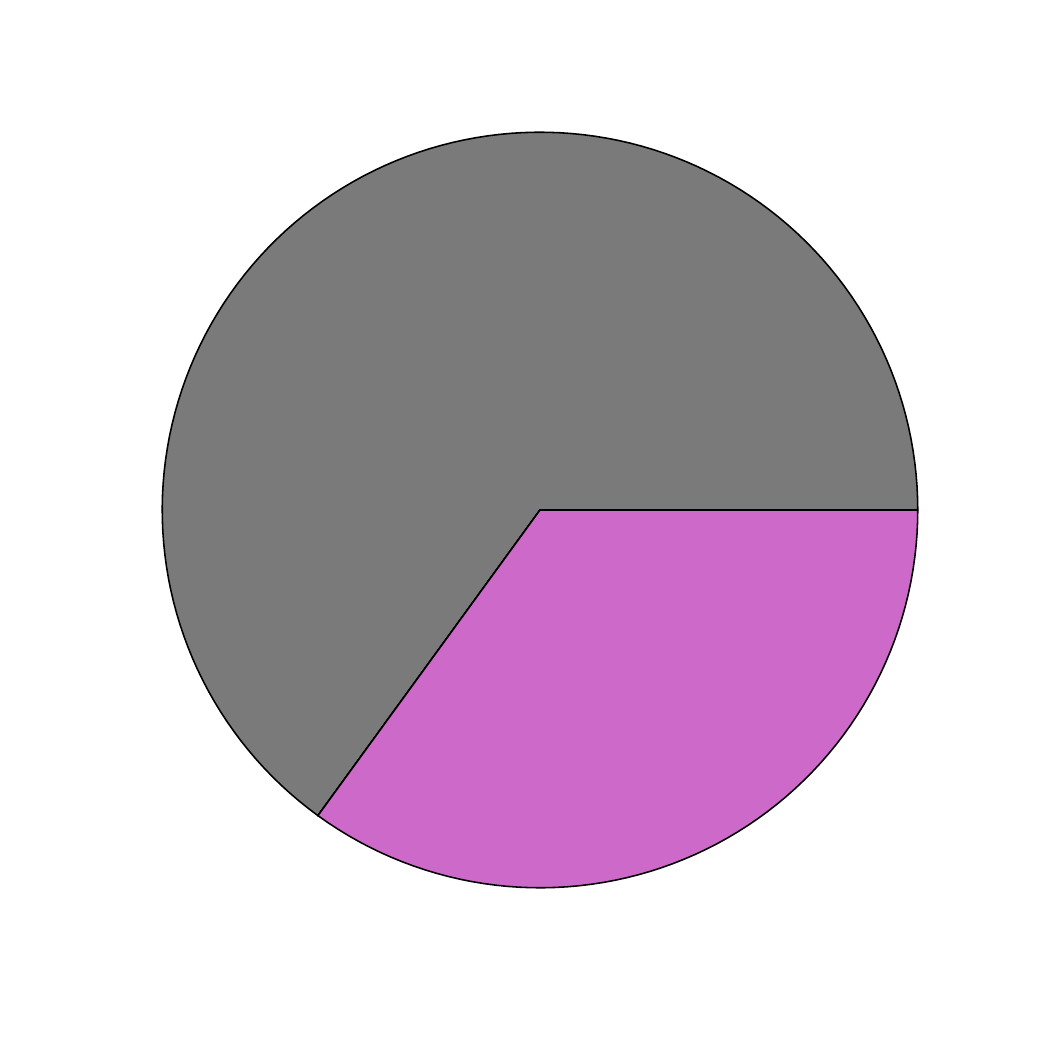} & 
\includegraphics[scale=0.2]{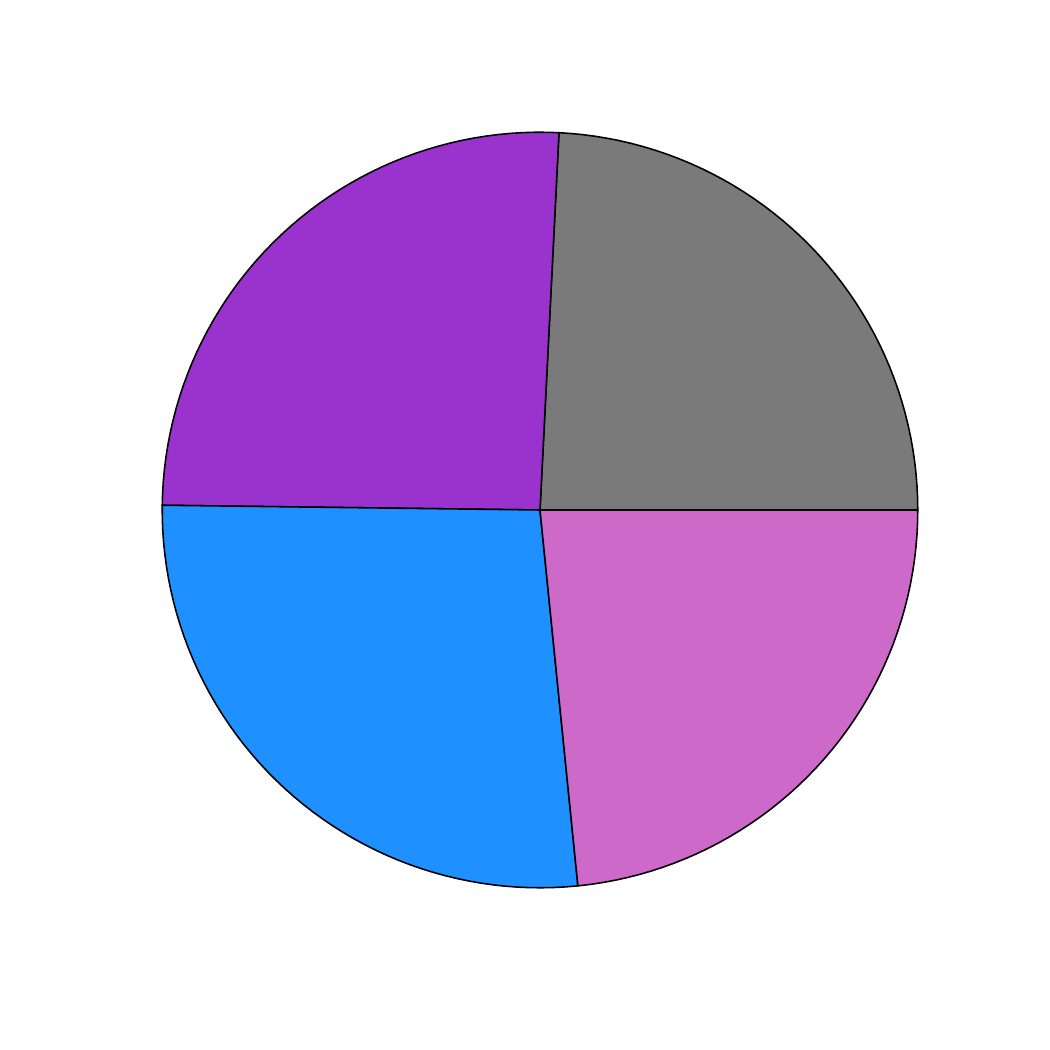} \\ [-4ex]
$\lambda_2$ & 
\includegraphics[scale=0.2]{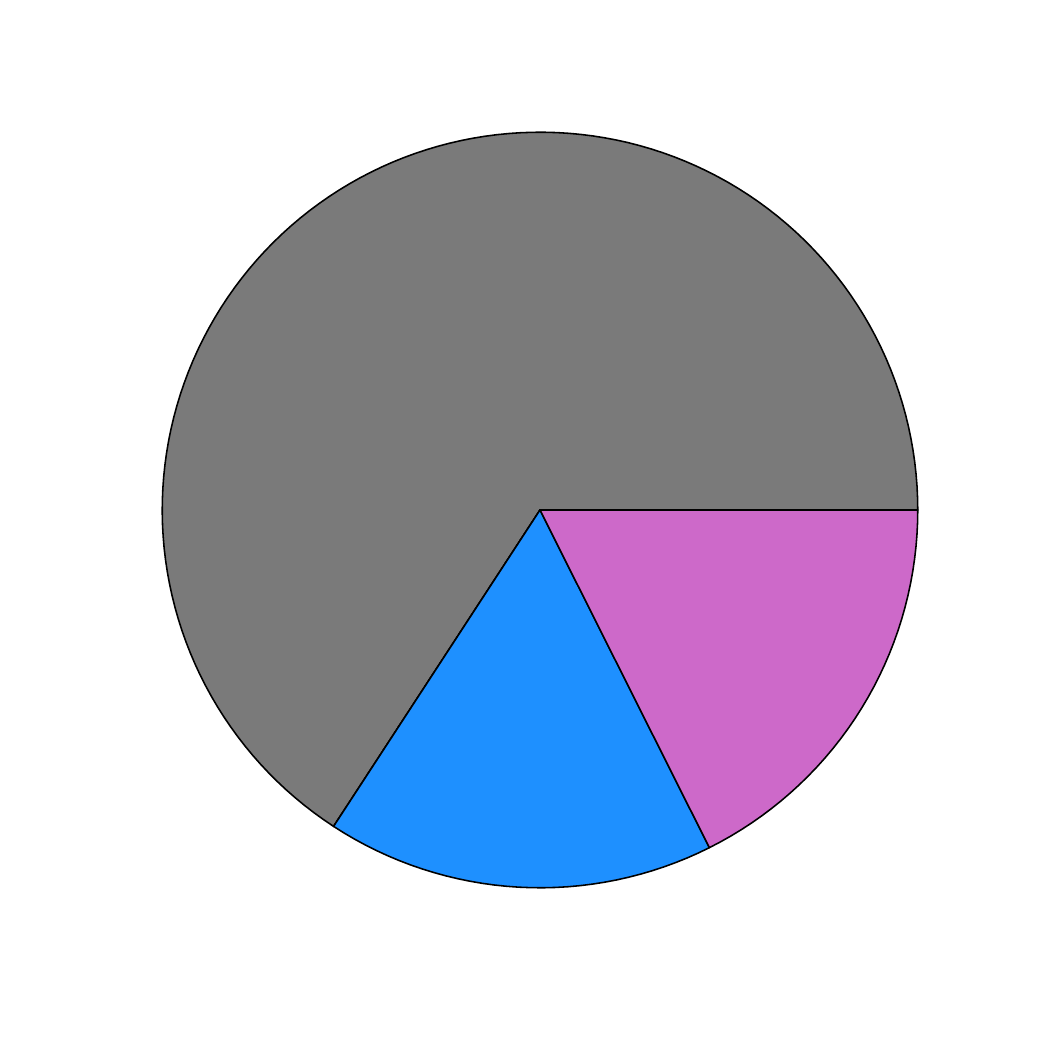} & 
\includegraphics[scale=0.2]{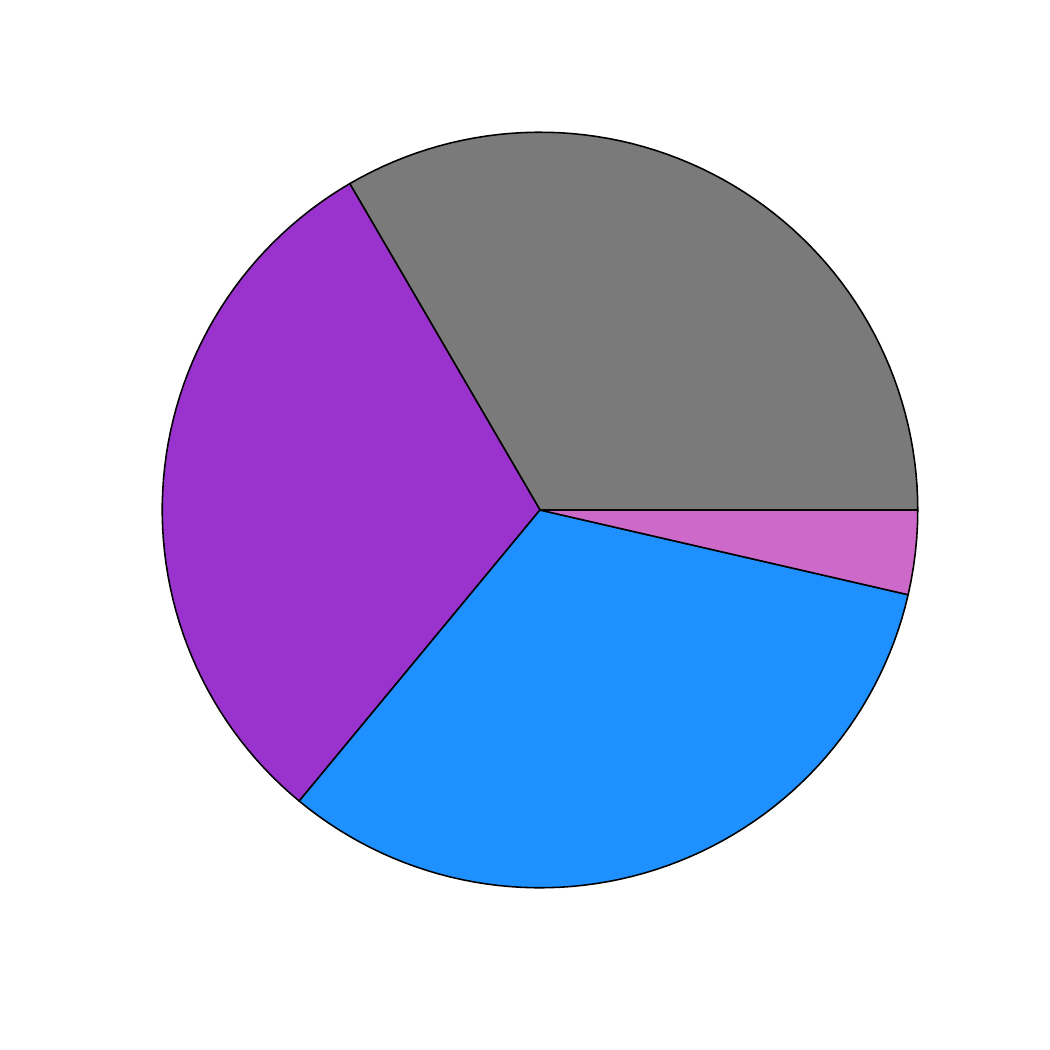} & 
\includegraphics[scale=0.2]{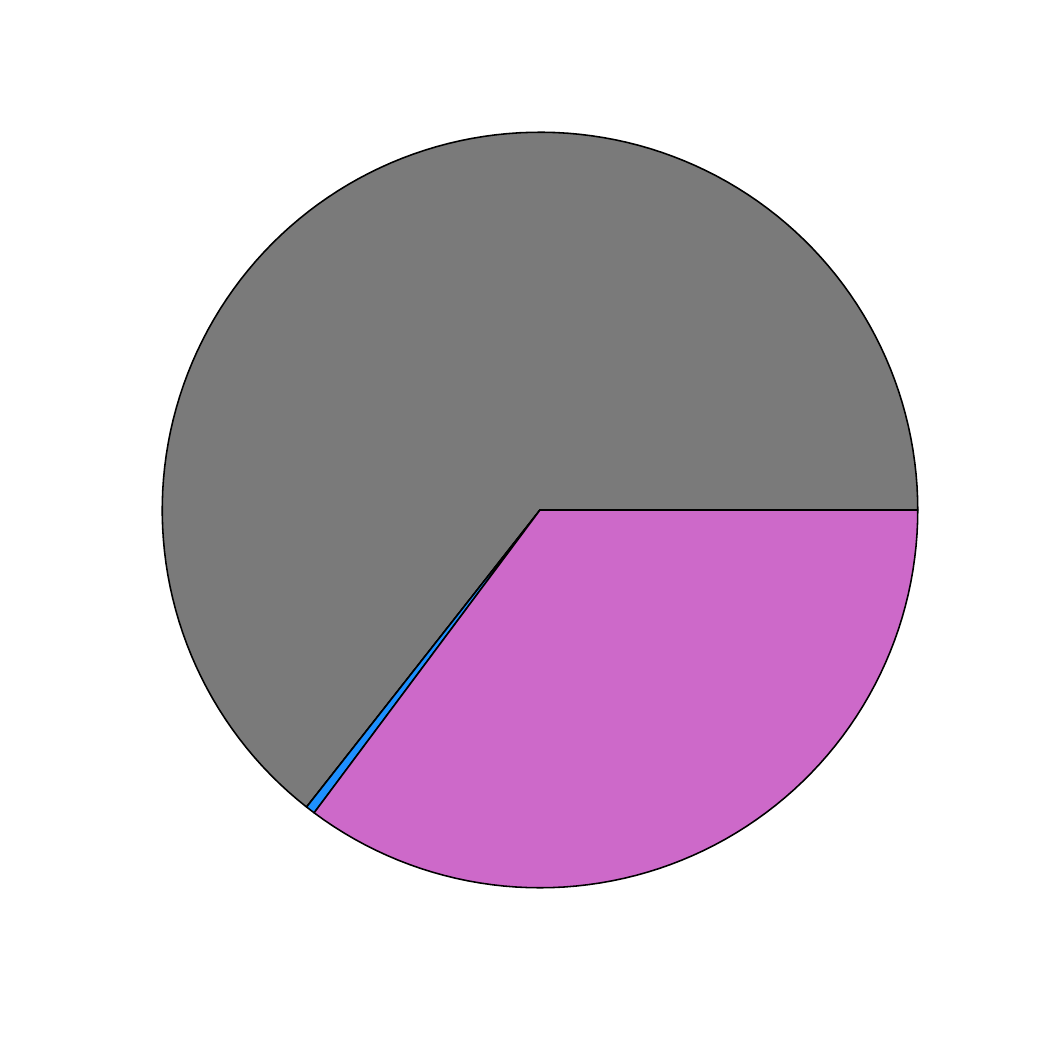} & 
\includegraphics[scale=0.2]{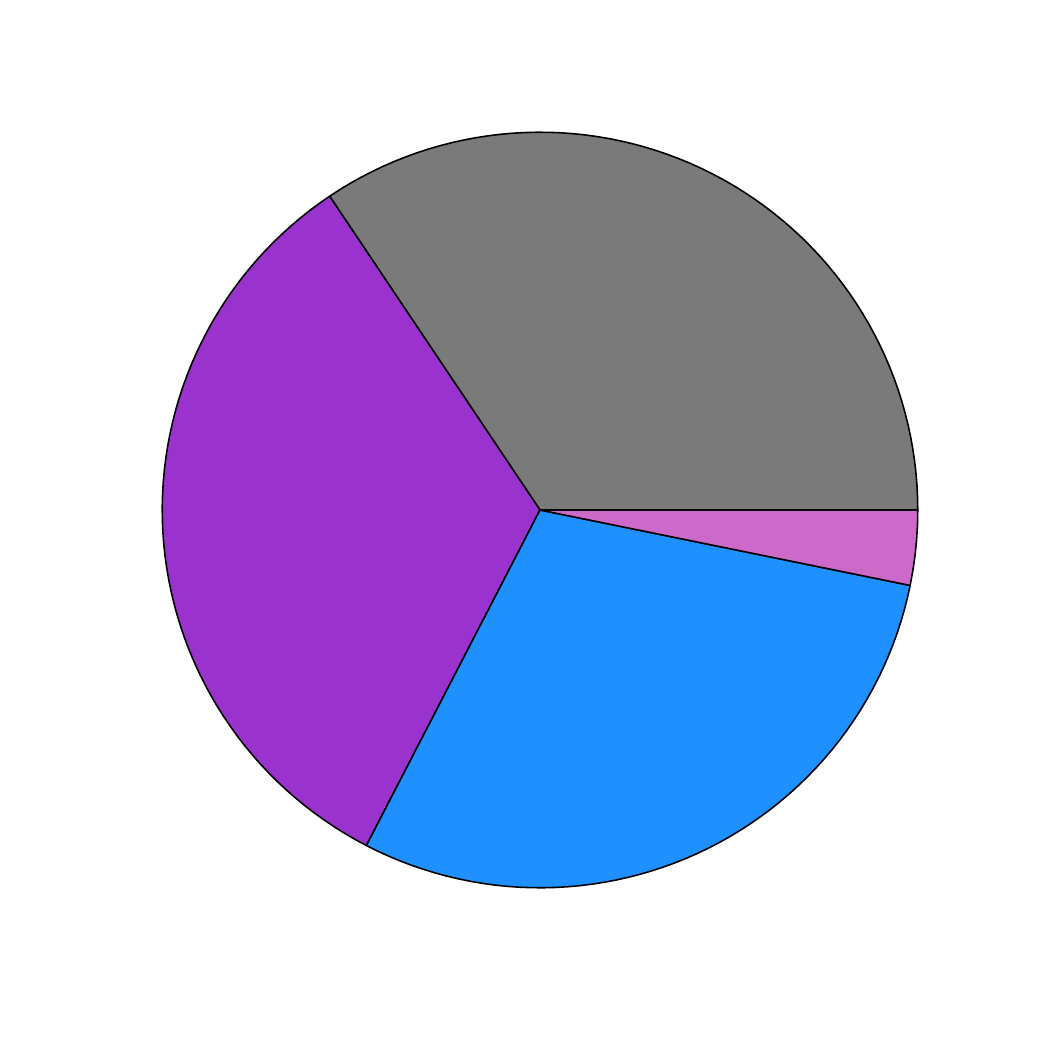}\\ [-4ex]
$\lambda_3$ & 
\includegraphics[scale=0.2]{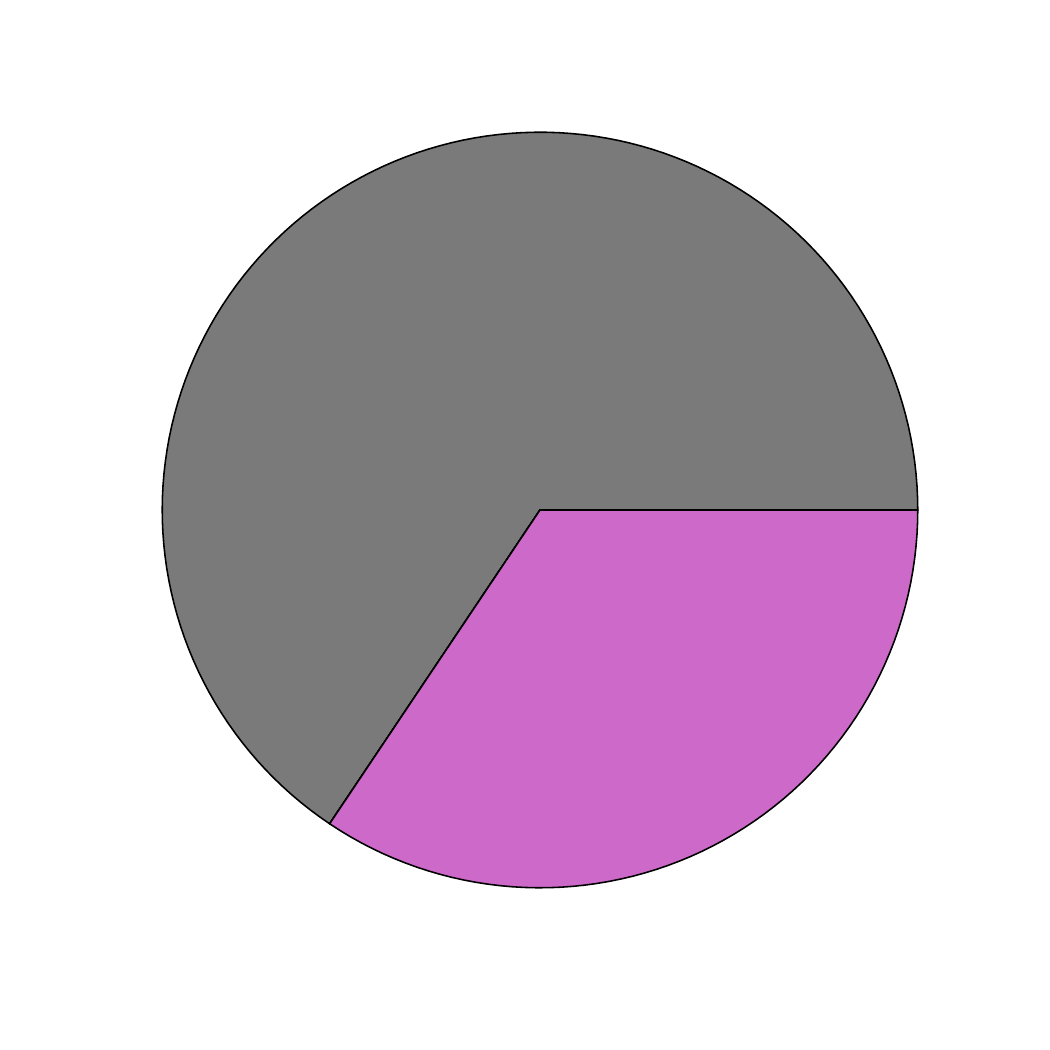} & 
\includegraphics[scale=0.2]{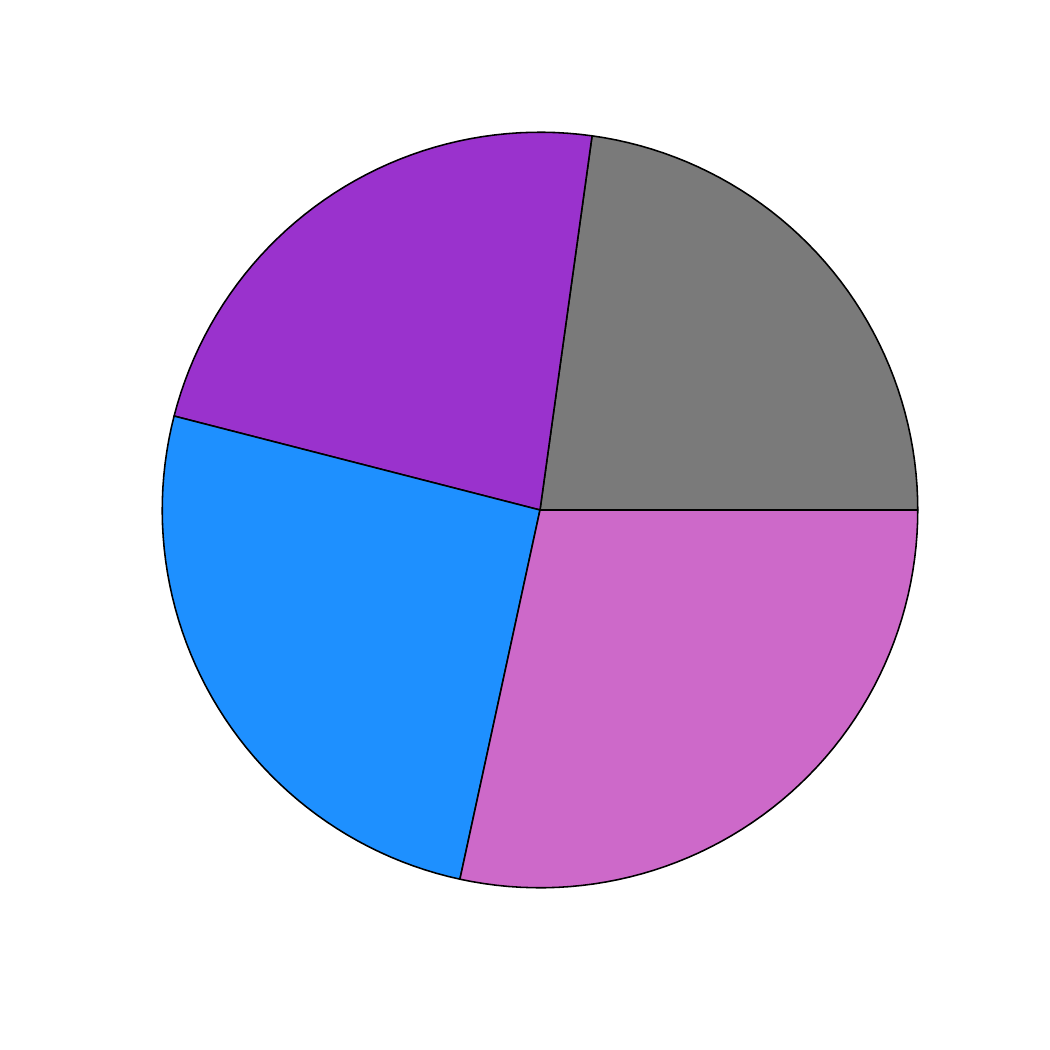} & 
\includegraphics[scale=0.2]{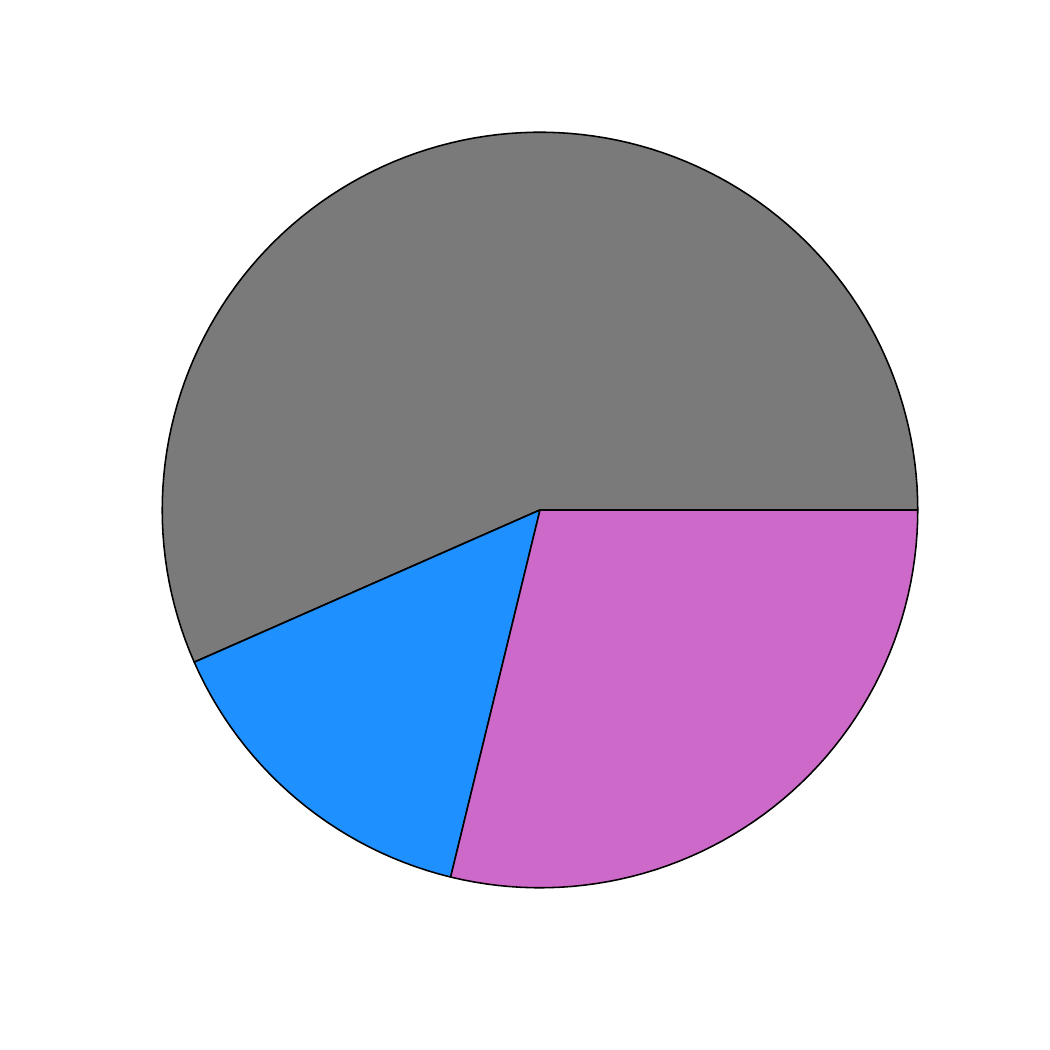} & 
\includegraphics[scale=0.2]{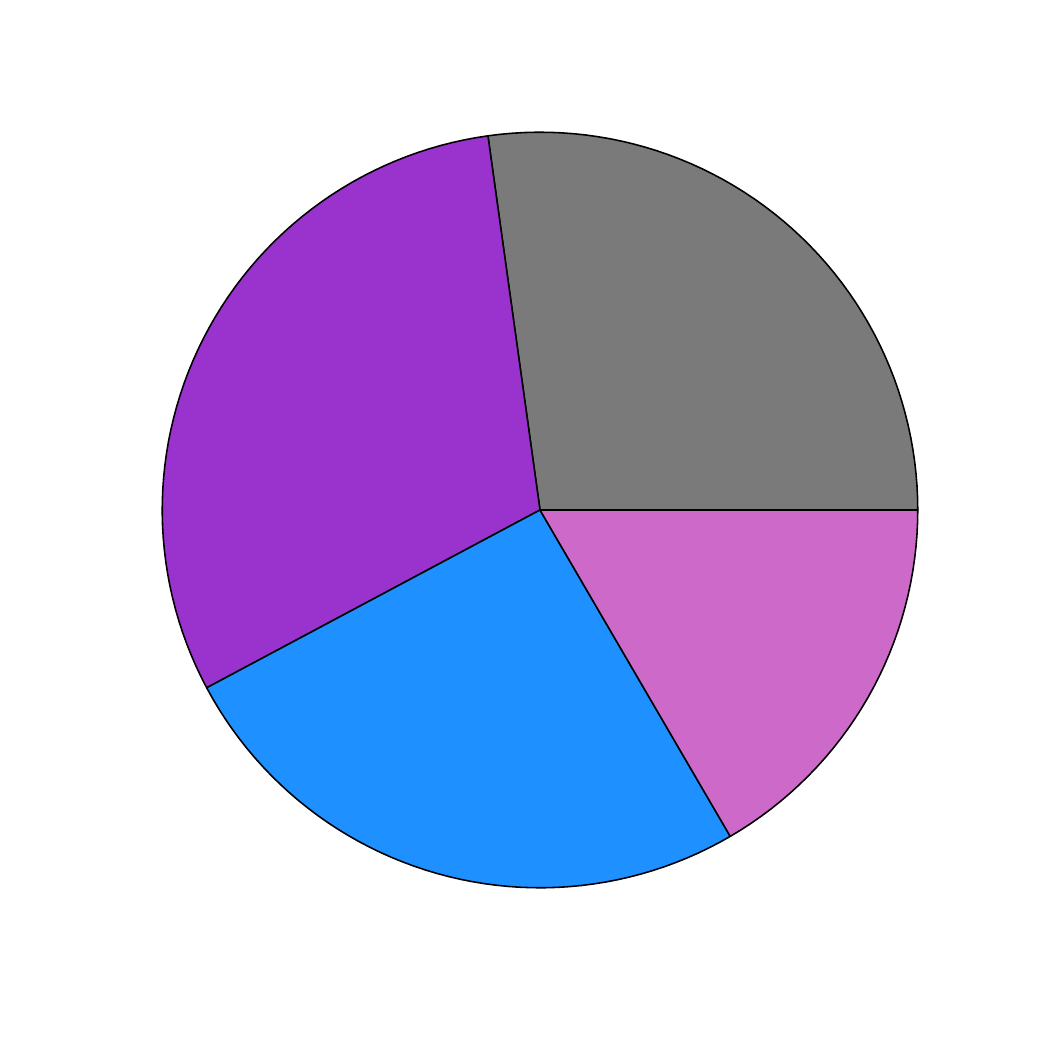}\\[-4ex]
$\lambda_4$ & 
\includegraphics[scale=0.2]{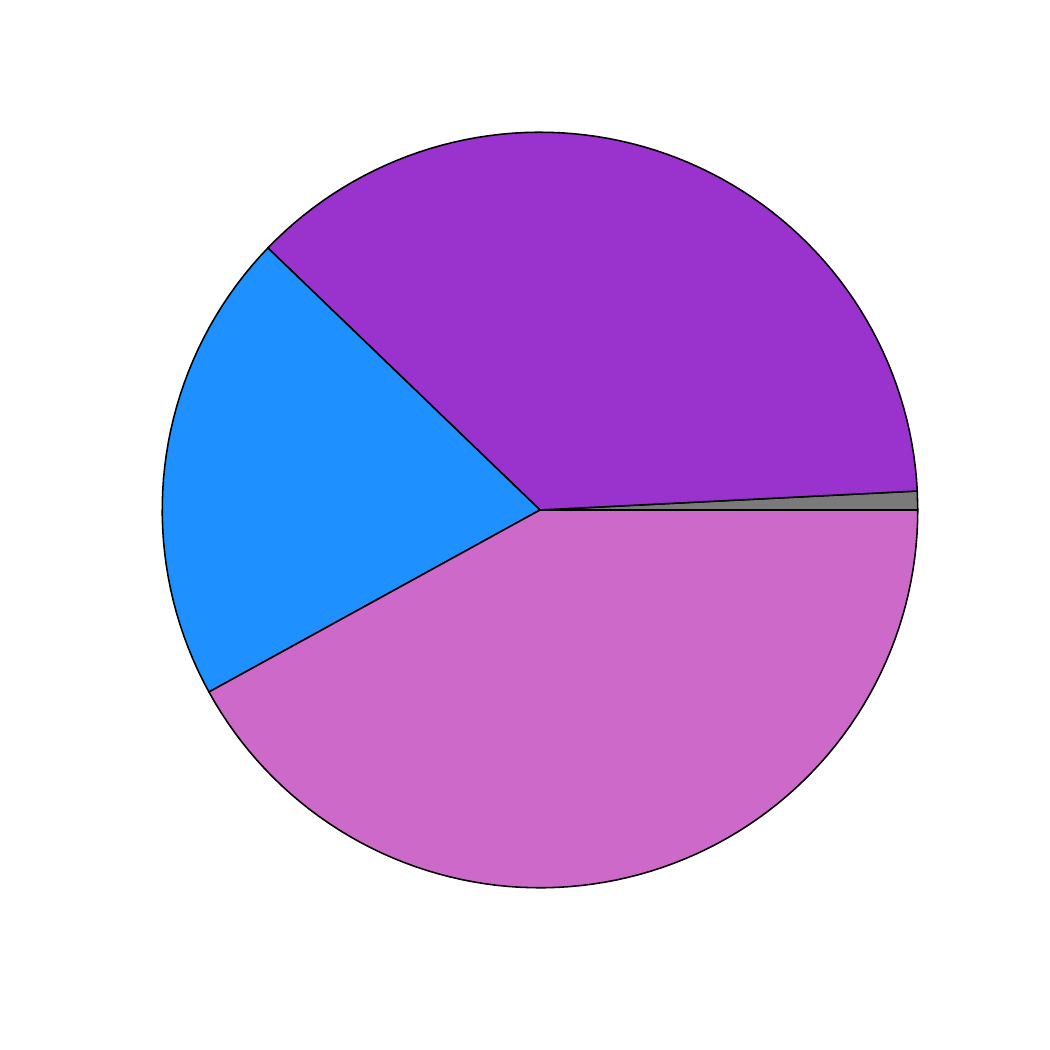} & 
\includegraphics[scale=0.2]{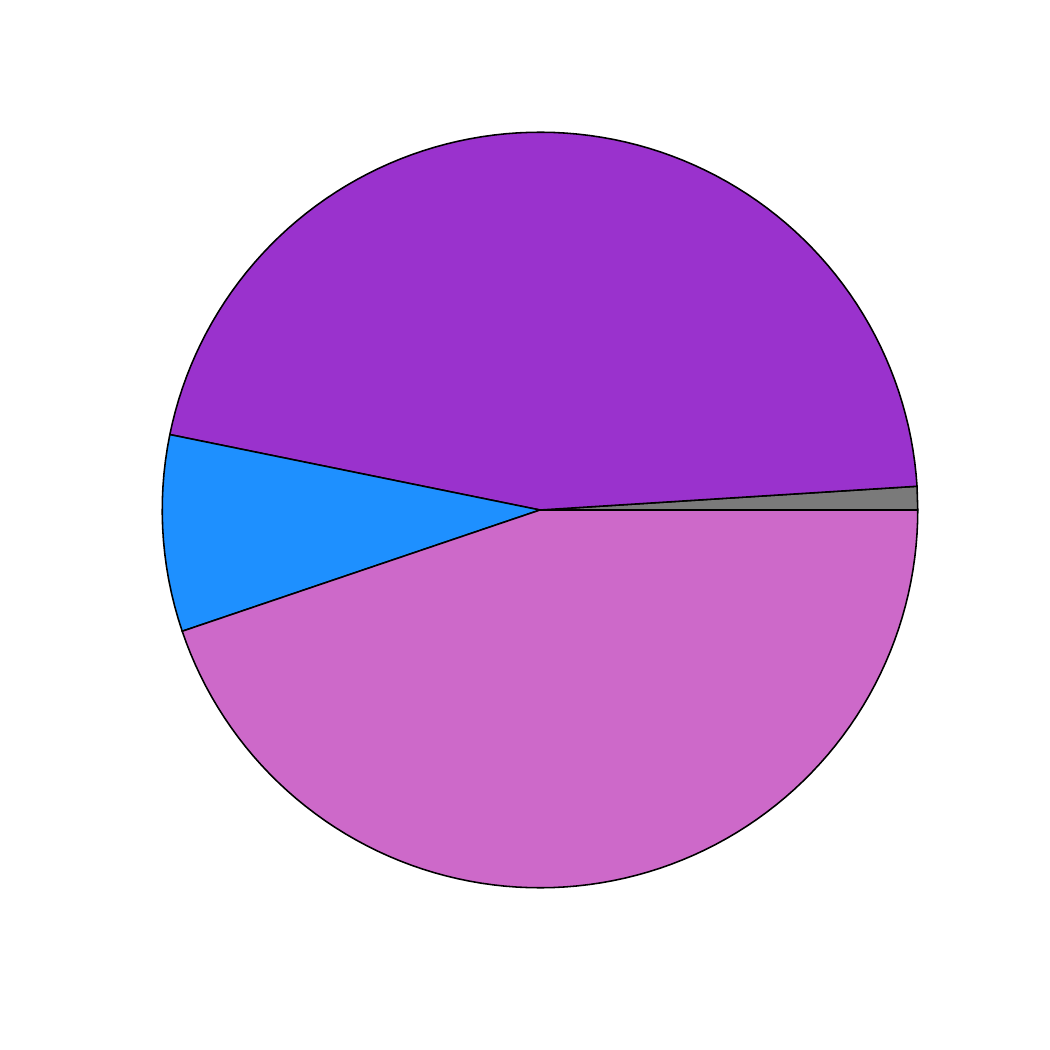} & 
\includegraphics[scale=0.2]{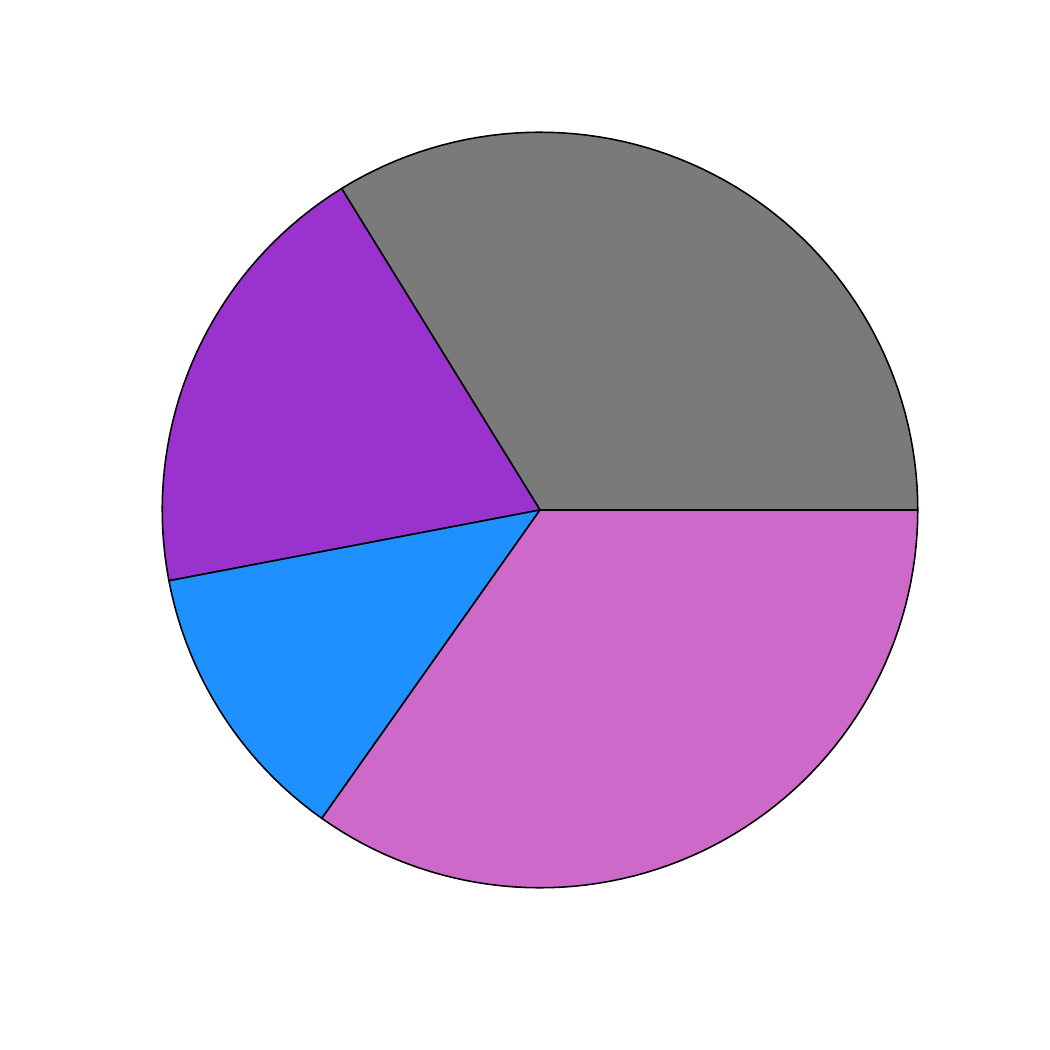} & 
\includegraphics[scale=0.2]{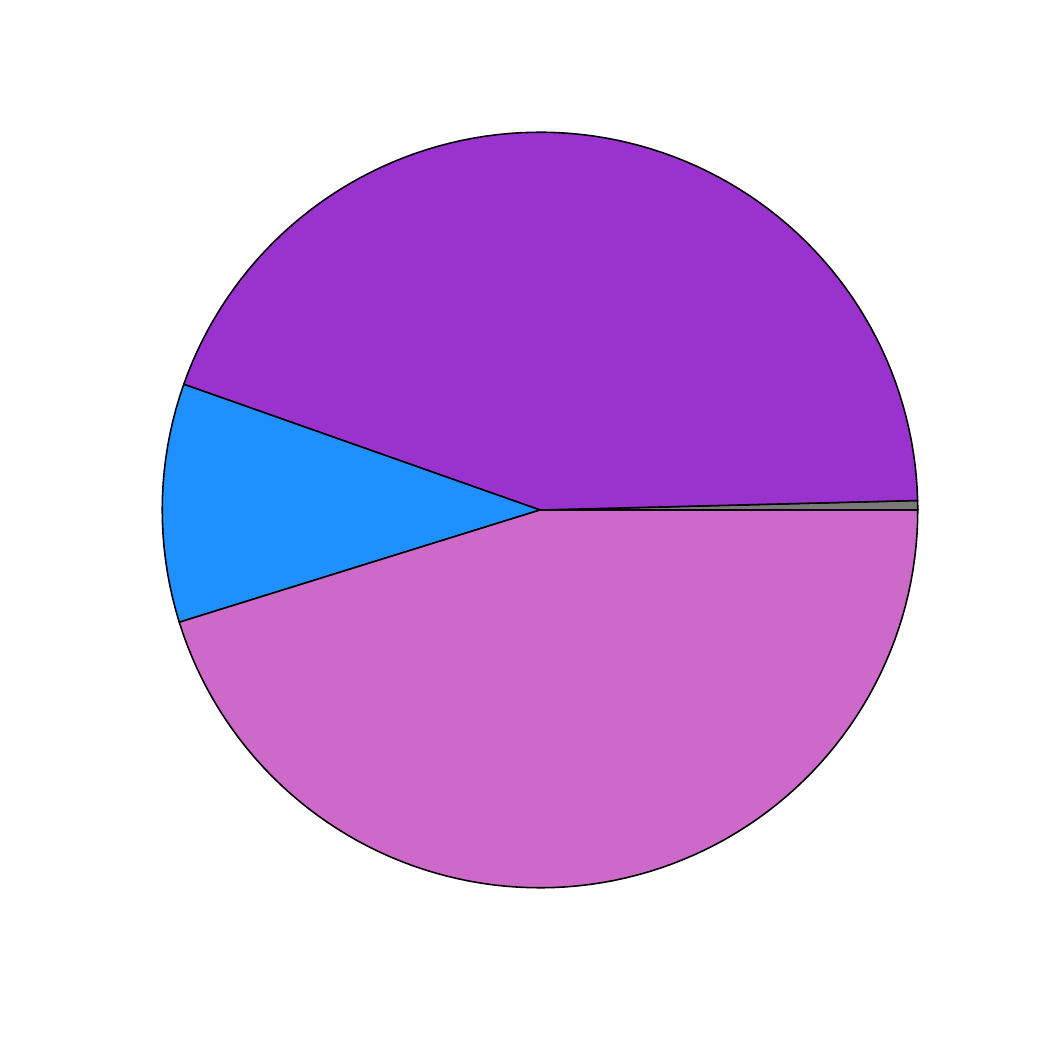}\\[-4ex]
$\lambda_5$ & 
\includegraphics[scale=0.2]{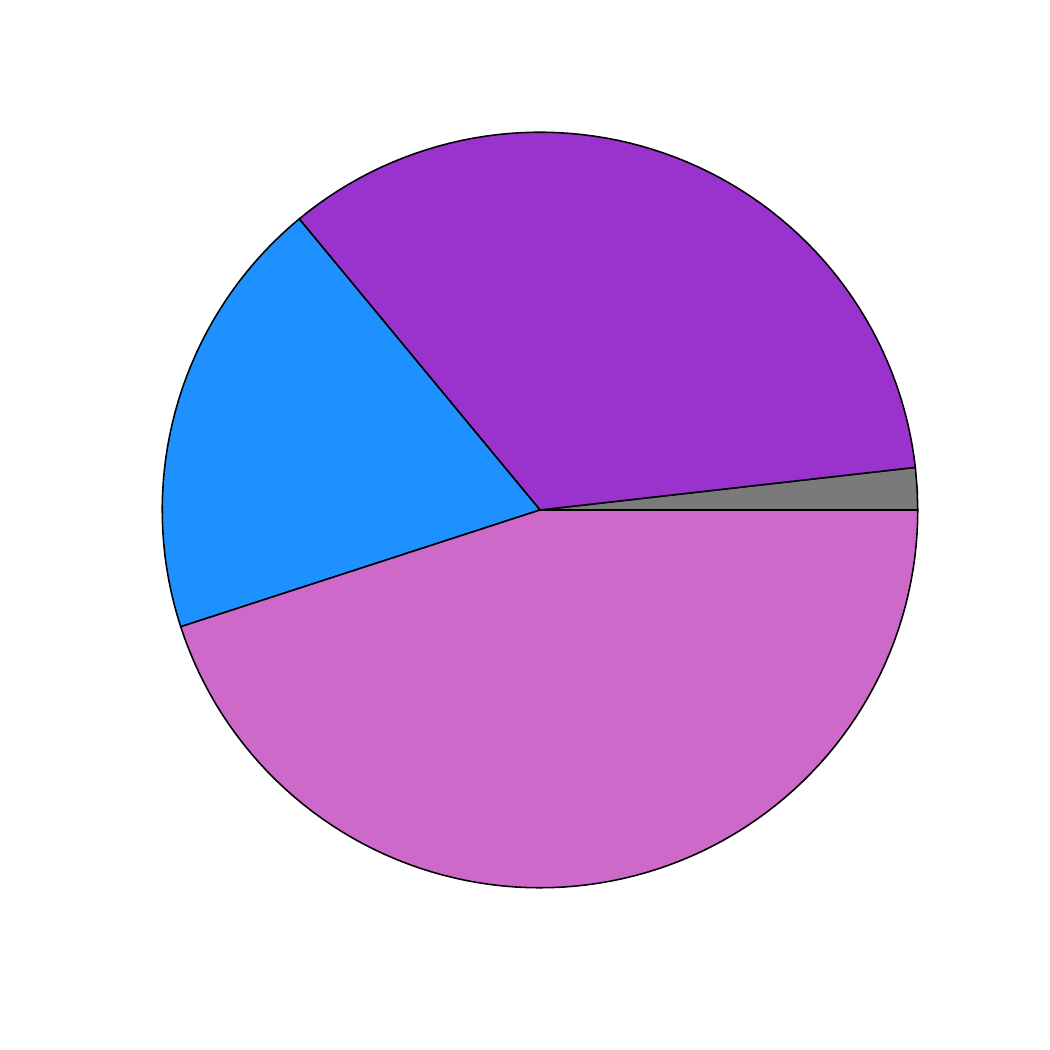} & 
\includegraphics[scale=0.2]{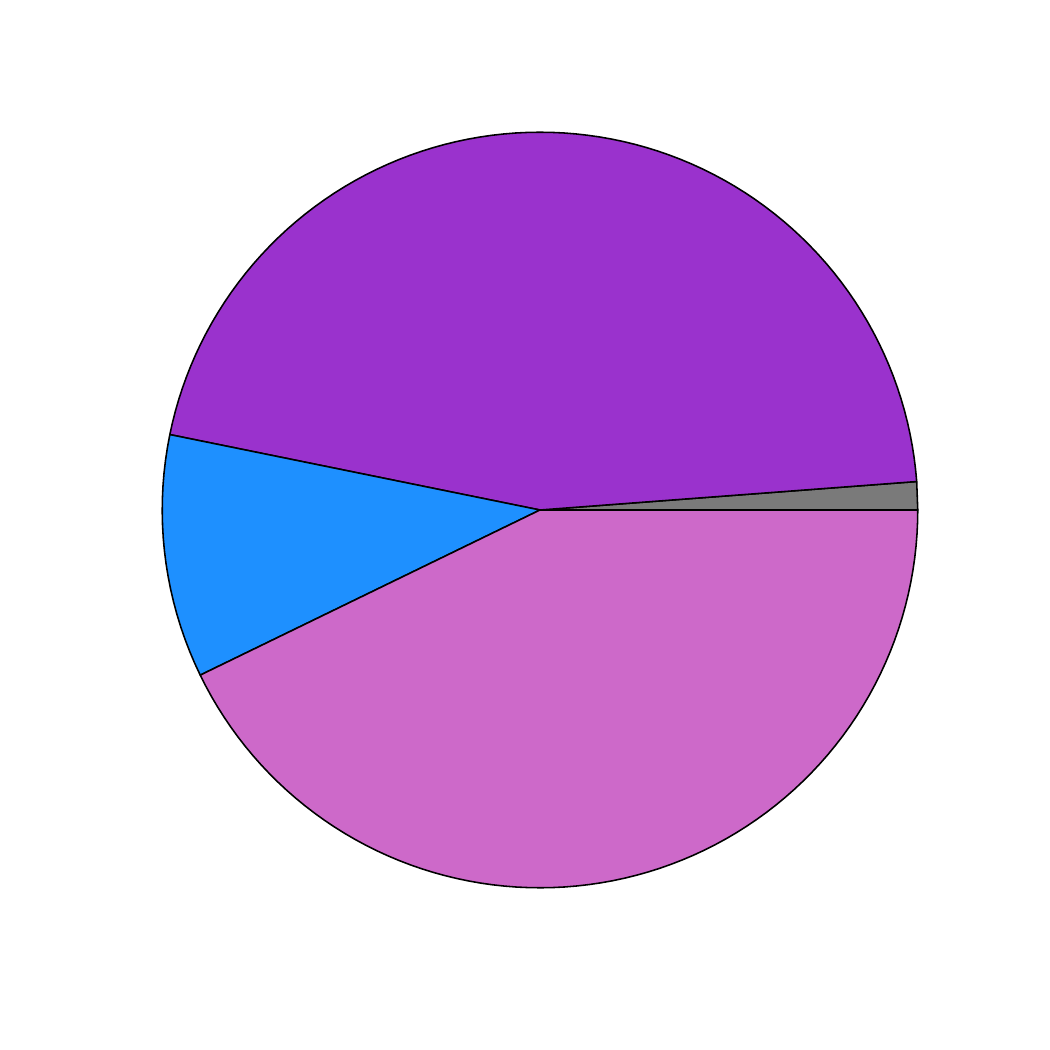} & 
\includegraphics[scale=0.2]{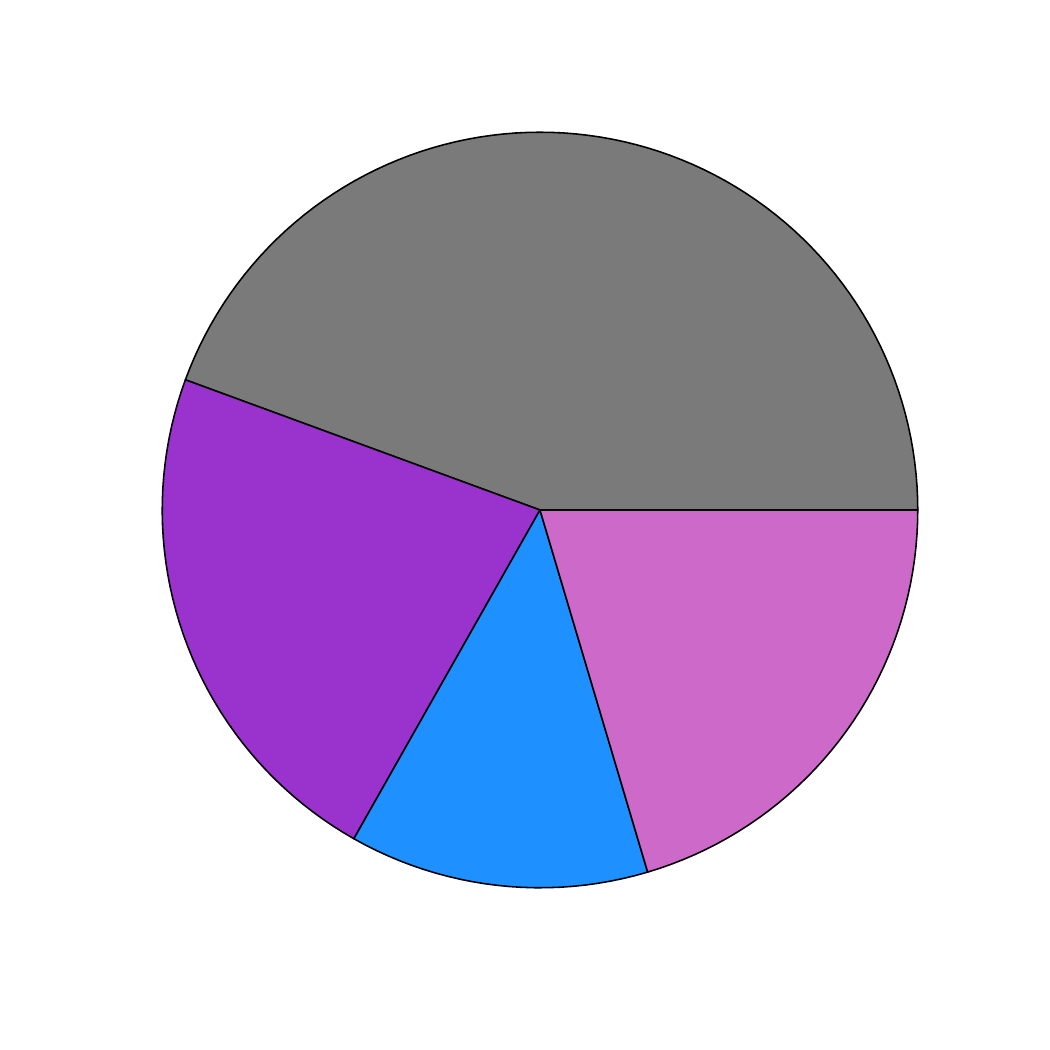} & 
\includegraphics[scale=0.2]{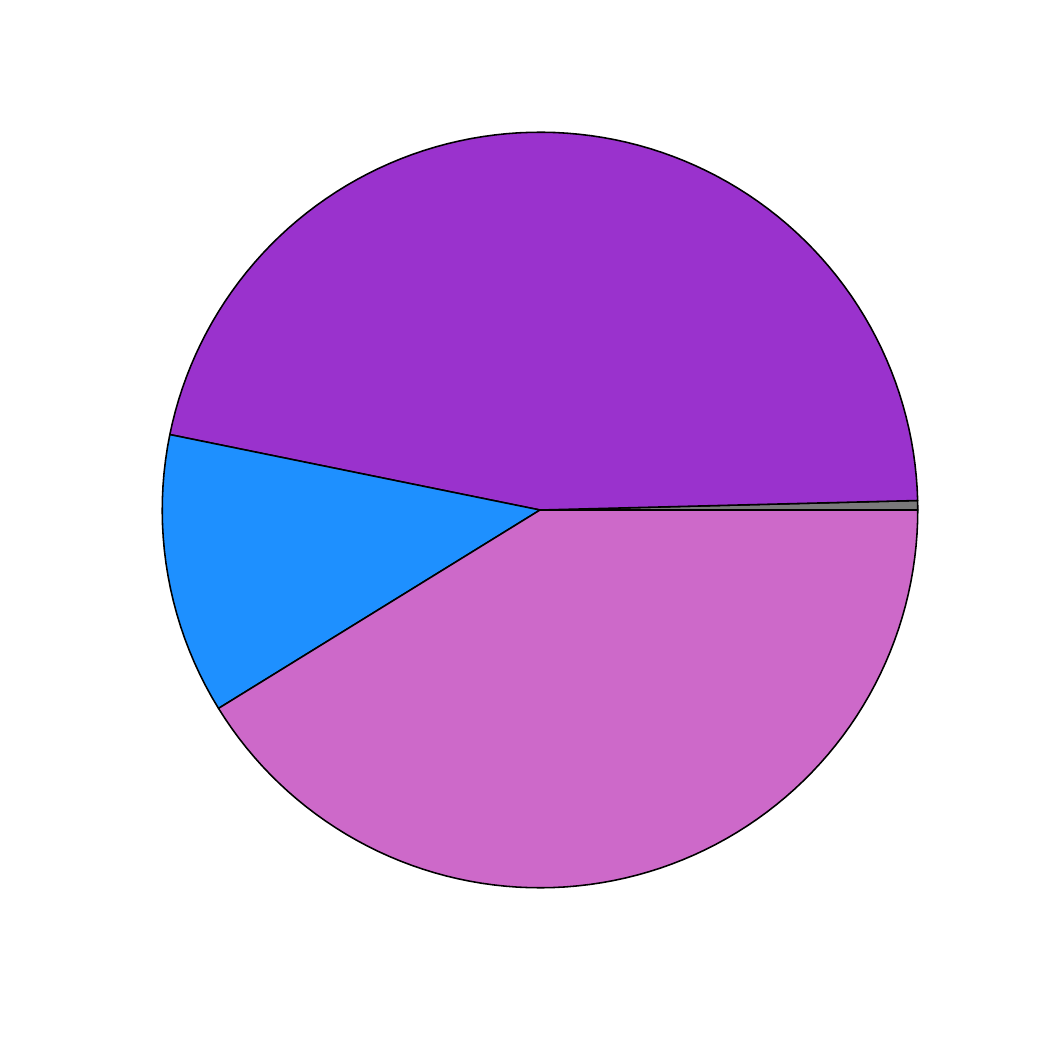}\\[-4ex]
$\lambda_6$ & 
\includegraphics[scale=0.2]{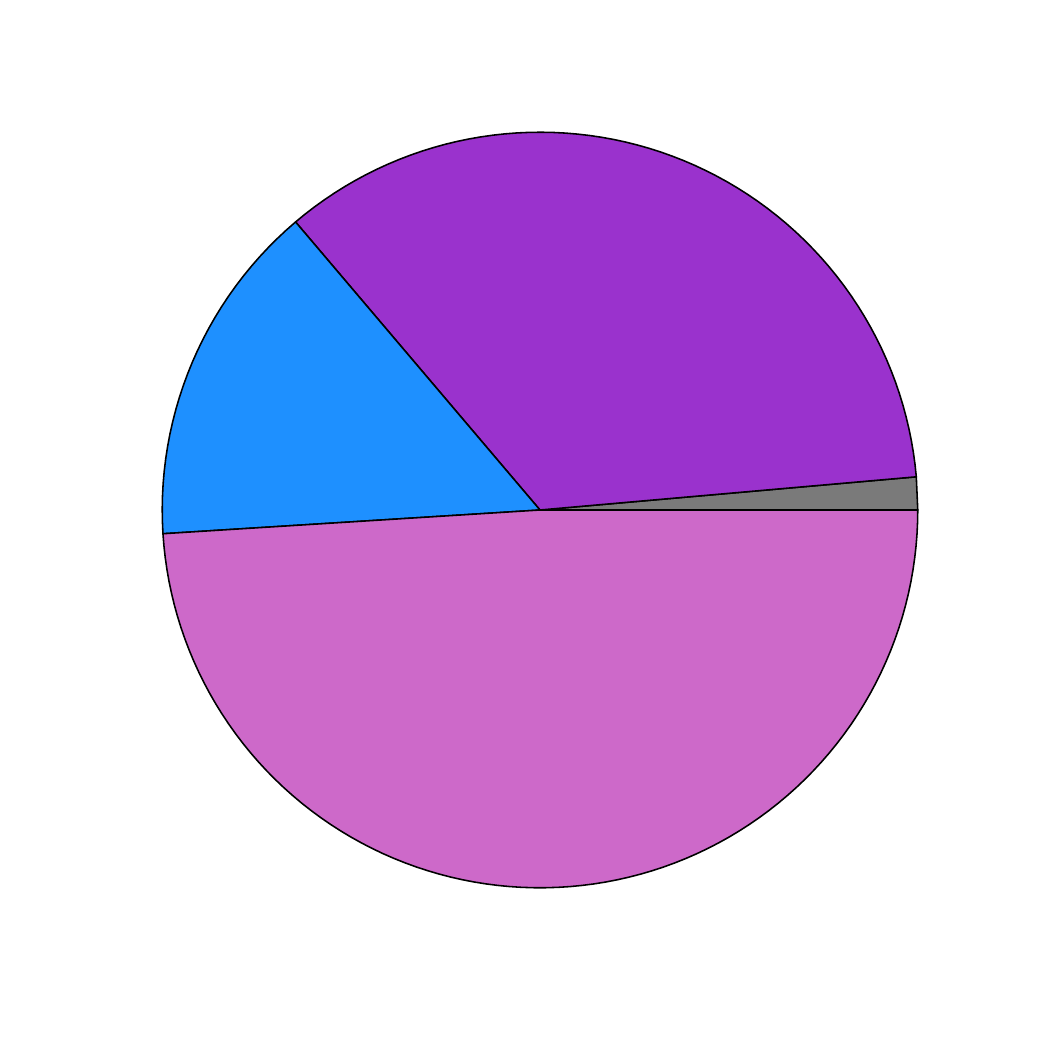} & 
\includegraphics[scale=0.2]{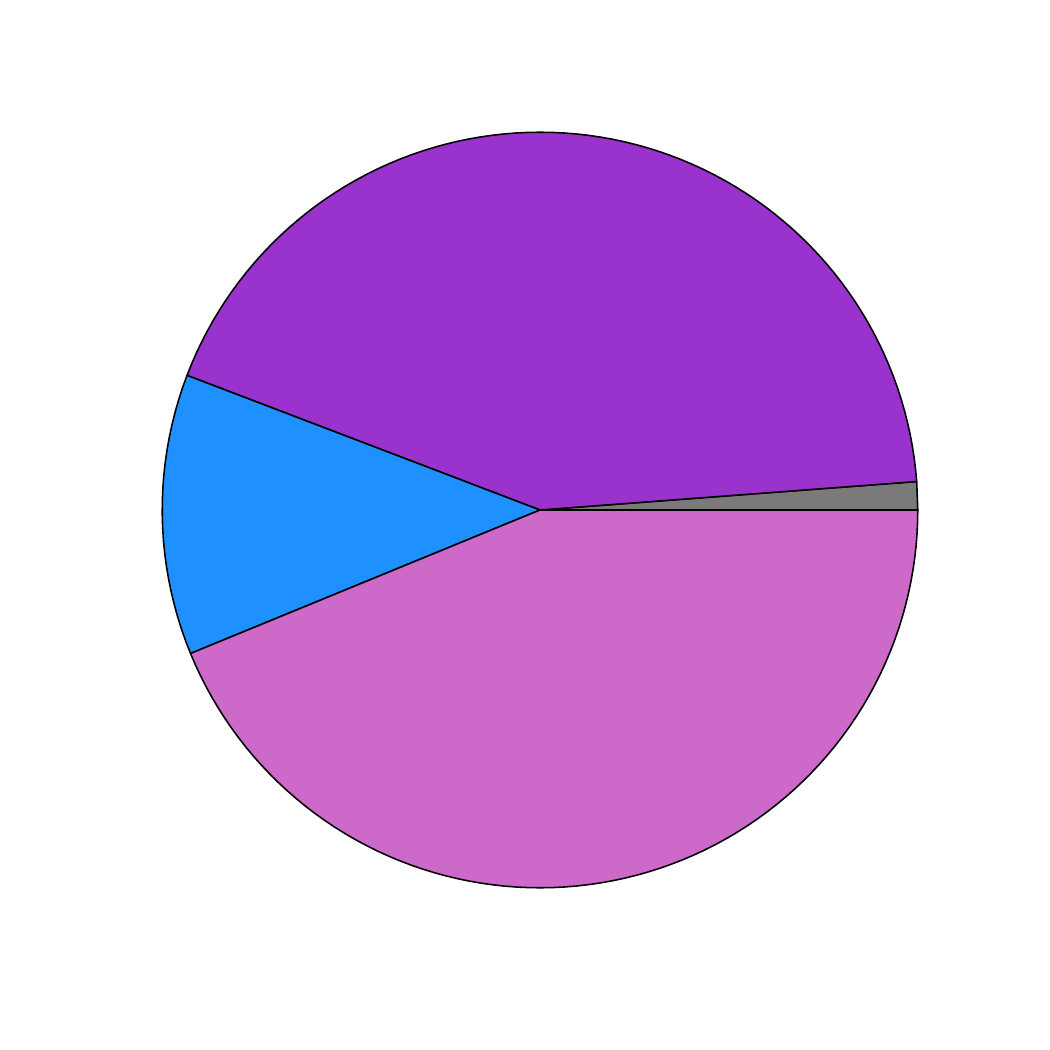} & 
\includegraphics[scale=0.2]{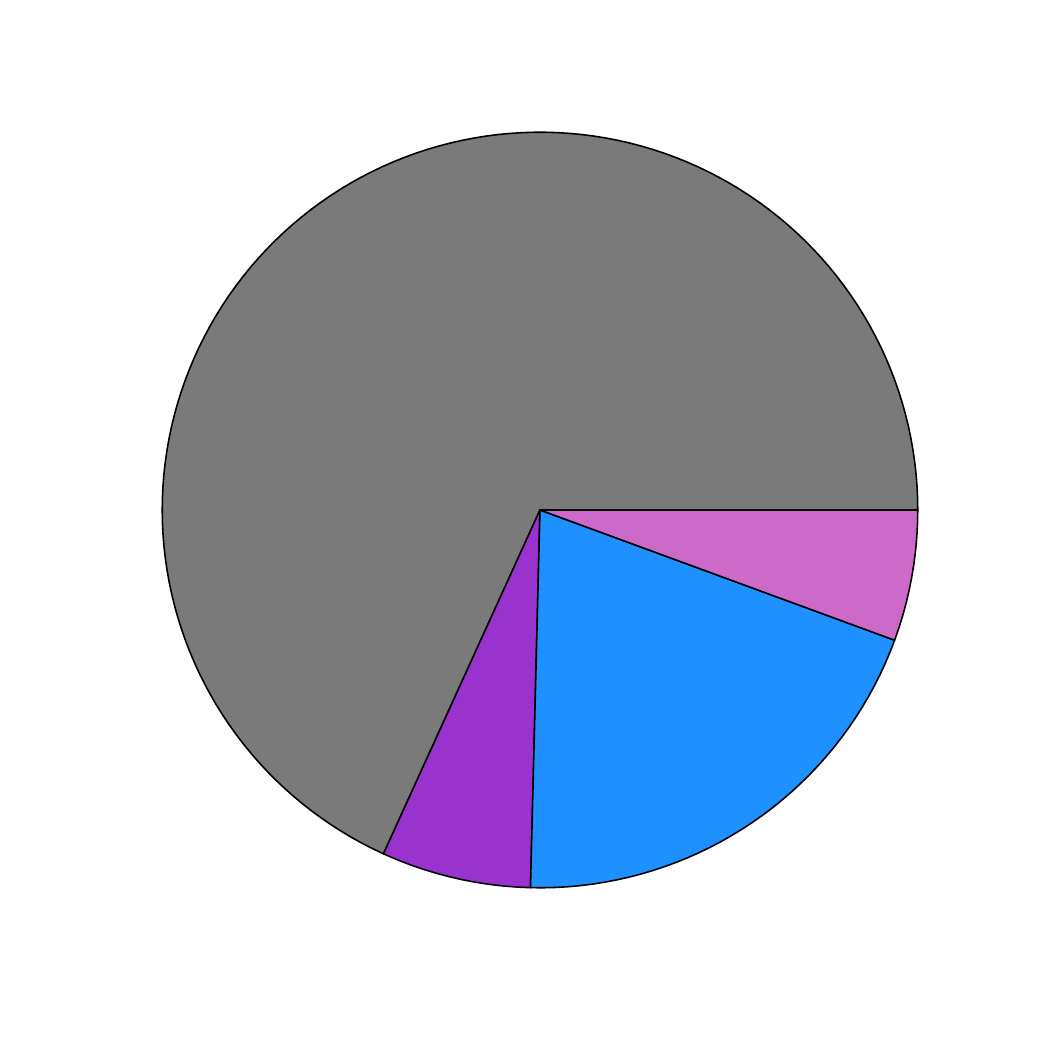} & 
\includegraphics[scale=0.2]{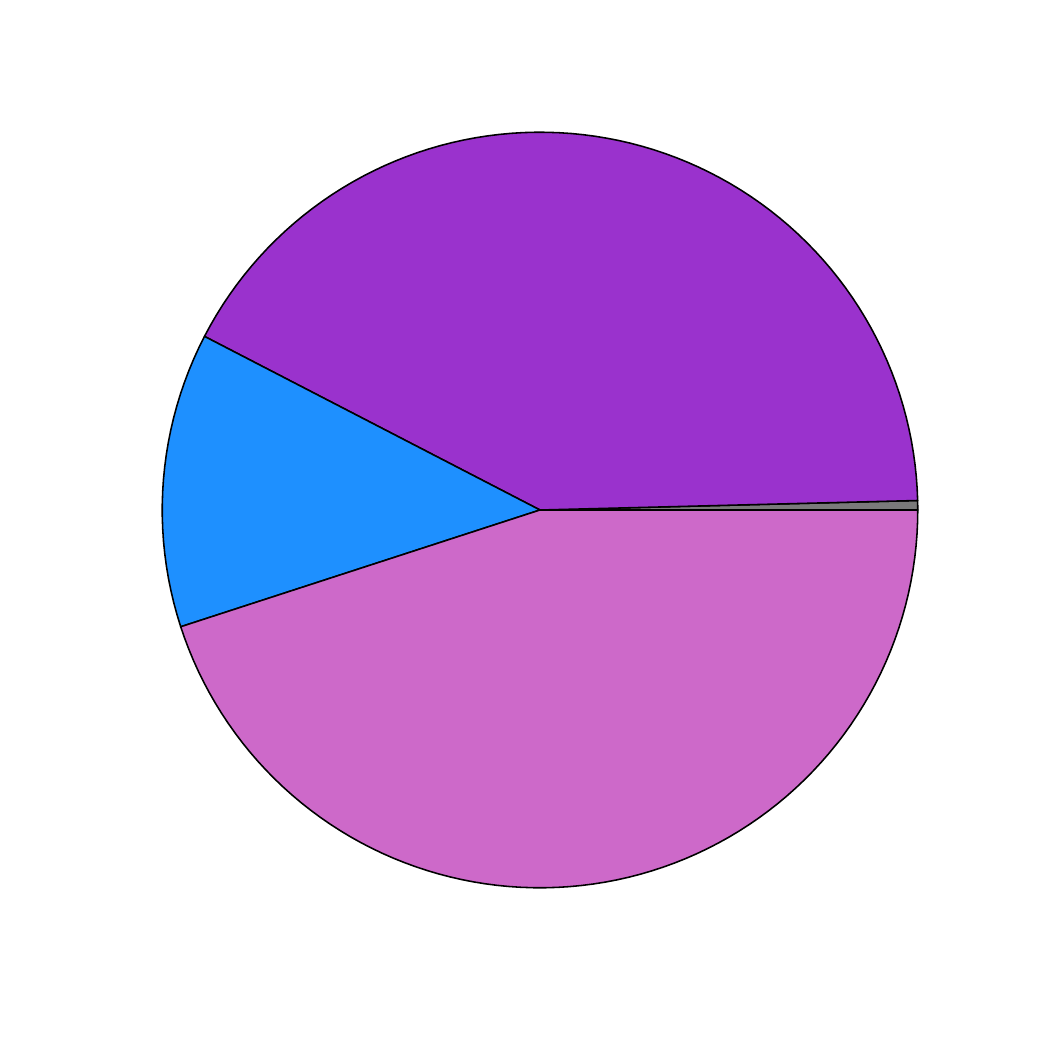}\\[-4ex]
   \end{tabular}
\vskip-0.1in 
\caption{\small \label{fig:lambdas-c0-c6} Pie charts with the proportion of times each value in the grid was selected for contaminations $C_0$ and $C_4$. The gray, purple, blue and pink  areas correspond to the values $0$, $0.2$, $0.4$ and   $0.6$, respectively.} 
\end{center}
\end{figure}

 The robust proposal performs similarly to the classical one for these last three parameters, but for the first three ones the gray zones are much smaller than those obtained for the least-square counterpart. Under $C_4$, the robust approach presents a stable and reliable behaviour   since the pie charts are almost similar to those obtained under $C_0$. In contrast, when looking at the behaviour of the classical selection procedure, one cannot avoid noticing that, even though for $\lambda_1$ to $\lambda_3$   the $0$ value was selected most of the times, the pie charts of the last three parameters are very different from those obtained for clean samples. More precisely, for the classical procedure the value $0$ is selected most of the times for the penalty parameters related to the last three components of $\bbe$, even when  these parameters correspond to  zero components. This fact explain the poor behaviour of the least squares estimator reported in Table \ref{tab:betas} under $C_4$, specially when considering the estimation of $\beta_6$, which is never estimated as $0$.

\section{Real data example}{\label{sec:realdata}}
In this section, we analyse  the plasma beta-carotene level data set, collected by \citet{Nierenberg:etal:1989}  which is also available  in \texttt{R} as the data set \texttt{plasma} of the library \texttt{gamlss.data}.  As mentioned in \citet{fairfield:fletcher:2002}, modelling the plasma concentrations of beta-carotene is of interest since low values might be associated with an increased risk of developing certain types of cancer such as lung, colon breast and prostate cancer.
This data set, that consists of 315 observations, was also considered in \citet{liu:etal:2011} who proposed a partially linear additive model  and estimated the  parameters using a least squares approach combined with  the SCAD penalty. More precisely, \citet{liu:etal:2011} considered as response the logarithm of BETAPLASMA which is modelled using a partially linear additive model taking as covariates associated to the linear component the gender labelled    SEX and the covariates BMI, CAL, FAT, FIBER, BETADIET,  ALCOHOL, SMOKE2 and SMOKE3, while  AGE and CHOL correspond to the predictors included in the model through additive nonparametric components, where the labelled covariates correspond to 
\begin{eqnarray*}
\mbox{SEX} &=& \mbox{1=male, 0=female}\\
\mbox{SMOK1}&=& \mbox{1=former smoker, 0=other}\\
\mbox{SMOK2} &=& \mbox{1=current smoker, 0=other}\\
\mbox{BMI} &=& \mbox{body mass index equal to (weight/(height)$^2$)}\\
\mbox{VIT1} &=& \mbox{1=fairly often, 0=other}\\
\mbox{VIT2} &=& \mbox{1=not often, 0=other}\\
\mbox{CAL} & =& \mbox{number of calories consumed per day}\\
\mbox{FAT} &=& \mbox{grams of fat consumed per day}\\
\mbox{ALCOHOL} &=& \mbox{number of alcoholic drinks consumed per week}\\
\mbox{BETADIET} &=& \mbox{dietary beta-carotene consumed (mcg/day)}\\
\mbox{CHOL} &=& \mbox{cholesterol consumed (mg/day)}\\
\mbox{FIBER} &=& \mbox{grams of fiber consumed per day}
\end{eqnarray*}
\citet{Guo:etal:2013}  proposed to model the BETAPLASMA using a \textsc{plam} with the same covariates as before but with  FIBER    entering the model in the additive component, instead of in the linear one. They considered an estimation procedure based on the composite quantile regression. The same model was  considered in \citet{lv:yang:guo:2017} who used a modal regression as estimation procedure. 

In this section, we consider  the model \textsc{plam} proposed by \citet{Guo:etal:2013}, that is, the response $Y$ is the plasma beta-carotene in ng/ml, named BETAPLASMA, which is modelled using the partially linear additive model
\begin{align}
Y = & \,\mu+\beta_1\mbox{SEX}+\beta_2\mbox{SMOK1}+\beta_3\mbox{SMOK2}+\beta_4\mbox{BMI}+\beta_5\mbox{VIT1} +\beta_6\mbox{VIT2} +\beta_7\mbox{CAL} 
\nonumber\\
& +  \beta_8\mbox{FAT}+\beta_9\mbox{ALCOHOL}+\beta_{10}\mbox{BETADIET}
 +\eta_1(\mbox{AGE}) +\eta_2(\mbox{CHOL})+\eta_3(\mbox{FIBER})+ \sigma \,\varepsilon
\nonumber\\
   = & \,\mu +\bbe \trasp \bZ+ \sum_{j=1}^3 \eta_j(X_j)+  \sigma \,\varepsilon\,.
\label{eq:betaplasma}
\end{align}
In \citet{liu:etal:2011} and \citet{Guo:etal:2013} one extremely high leverage point in alcohol consumption was observed and so the corresponding observation was deleted prior the analysis, then only 314 observations were used. Besides, in these papers and also in \citet{lv:yang:guo:2017} all variables except for the binary ones were standardized using the mean and the standard deviation. However, since the proposed method is resistant to outliers in the linear component, we considered the 315 observations. Furthermore,  taking into account that a robust procedure is used to estimate the unknown parameters, the non--binary variables are first standardized using as location the median instead of the mean and as dispersion the \textsc{mad} instead of the standard deviation.

In order to measure the performance of the estimators, we used a measure of the prediction capability of both methods. For that purpose, we split the sample in two groups.  A  sample of size $n_{\mbox{\footnotesize \sc test}}=100$  corresponding to the testing sample  was randomly selected, let $\itI$ the indices corresponding to this sample.  The remaining $n_{\mbox{\footnotesize \sc training}}=215$ observations were taken as the training sample to compute the estimators,  denoted $\wmu^{(-\itI)}$, $\wbbe^{(-\itI)}$ and $\weta_j^{(-\itI)}$.  Then, we calculated,   for $i\in \itI$, 
$$\wY_i= \wmu^{(-\itI)} +  \bZ_i\trasp \wbbe^{(-\itI)}+ \sum_{j=1}^3 \weta_j^{(-\itI)}(X_{ji})\,,$$
and  the prediction capability is measured through the median absolute prediction error, denoted \textsc{mape}, as
$\mbox{\textsc{mape}}=\median_{i \in \itI}\{|Y_i-\wY_i|\}\,.$
The selection of the testing sample   was repeated 50 times, leading to 50 values of the \textsc{mape}. Besides,  the   number of covariates selected at each  replication  was computed.  The average number of selected covariates  is  denoted as \textsc{av.size} in Table \ref{tab:mar-mape-av}  which also  reports the mean over the 50 replications of the \textsc{mape}.

\begin{table}[ht!]
\begin{center}
\begin{tabular}{|l|c|c|}
  \hline
\textsc{method} &   \textsc{mape} & \textsc{av.size} \\ 
  \hline
\textsc{penalized ls}   &   0.8850  &  7.70\\
\textsc{penalized rob}    &     0.6424   &   4.68 \\
\textsc{ls}  &  0.8920 &  10\\
\textsc{rob} &      0.6365  &   10 \\
   \hline 
\textsc{penalized ls$^{(-\textsc{out})}$} &  0.6326 & 5.18\\
   \hline 
\end{tabular}
\caption{\label{tab:mar-mape-av} Mean over replications of  the prediction errors (\textsc{mape}) and   average sizes of the resulting models (\textsc{av.size}) with  training samples of size $215$ and testing samples of size $100$, for the first four rows. The last row corresponds to the \textsc{mape} of the penalized least squares estimator computed without the detected vertical outliers and the  extremely high leverage point in alcohol consumption.}
\end{center}
\end{table}

All the measures were calculated for both the penalized robust proposal and its least squares counterpart (denoted \textsc{penalized} in the Table and Figure) and  also for the estimators with no penalization term, that is, for the robust approach of \citet{boente:martinez:2023} that do not select variables and  for the usual least squares approach which corresponds to $\rho(u)=u^2$. For the robust procedure, we use the the Tukey's loss function, as in the simulation study.
Figure  \ref{fig:MAPE-measure} displays the adjusted boxplots of the   \textsc{mape}   for the four estimators.

\begin{figure}[ht!]
\begin{center}
\includegraphics[scale=0.45]{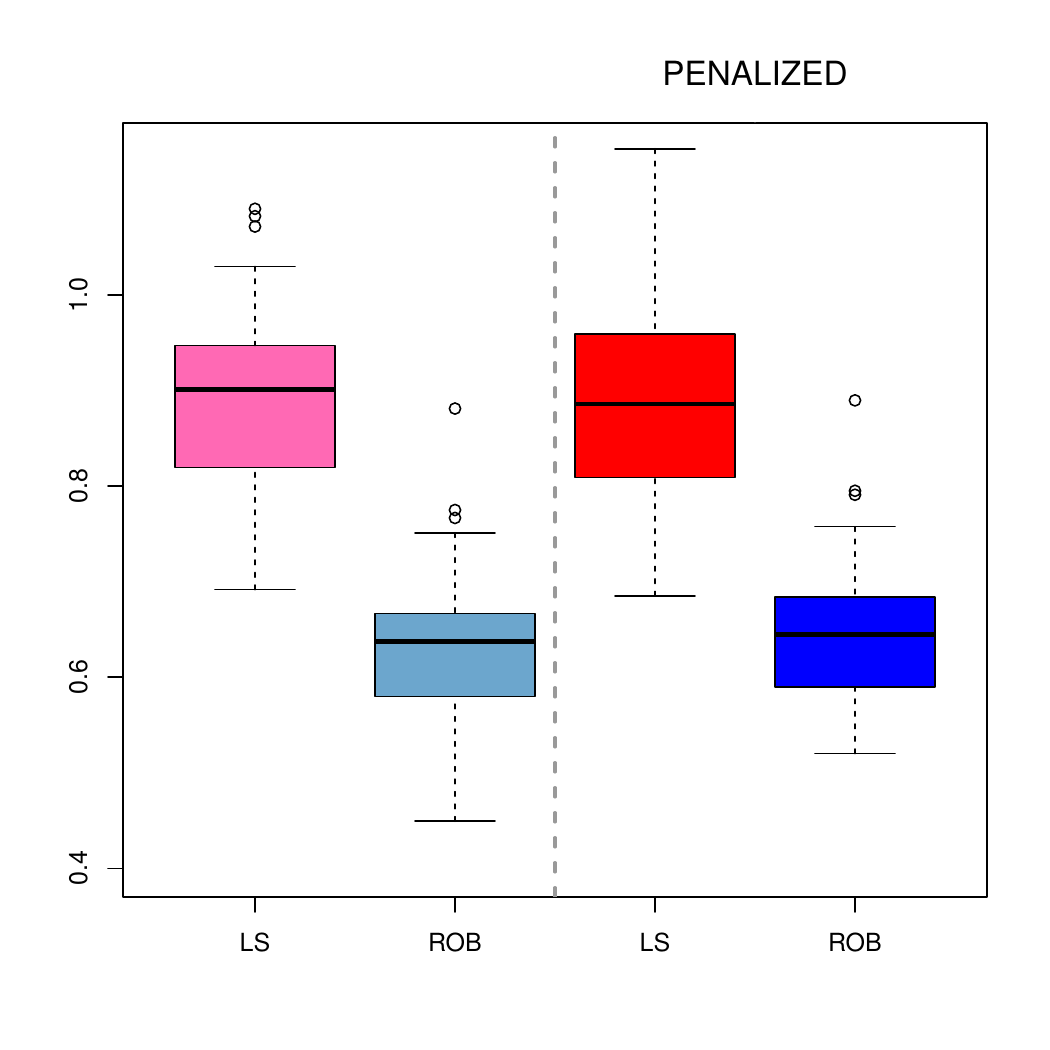}
\vskip-0.2in
\caption{\label{fig:MAPE-measure} Adjusted boxplots for the \textsc{mape} measures obtained for the estimators without penalization (on the left) and for the penalized (on the right) estimators.}
\end{center}
\end{figure}

As expected, the average size,  \textsc{av.size}, for the  estimators computed without a penalization term is equal to 10, i.e., to the number of covariates included in the linear regression component, since no variable selection is used. When considering the penalized estimators, the penalized robust proposal selects in average   4.68 covariates instead of the 10 original covariates $\bZ$ and the least squares approach leads to a larger number of components. The prediction measure considered is also increased for the least squares method both for the penalized and for the non--penalized estimators when compared with their robust counterparts. This behaviour may be explained by the effect that outliers have on the classical estimators. 

\begin{figure}[ht!]
\begin{center}
\includegraphics[scale=0.35]{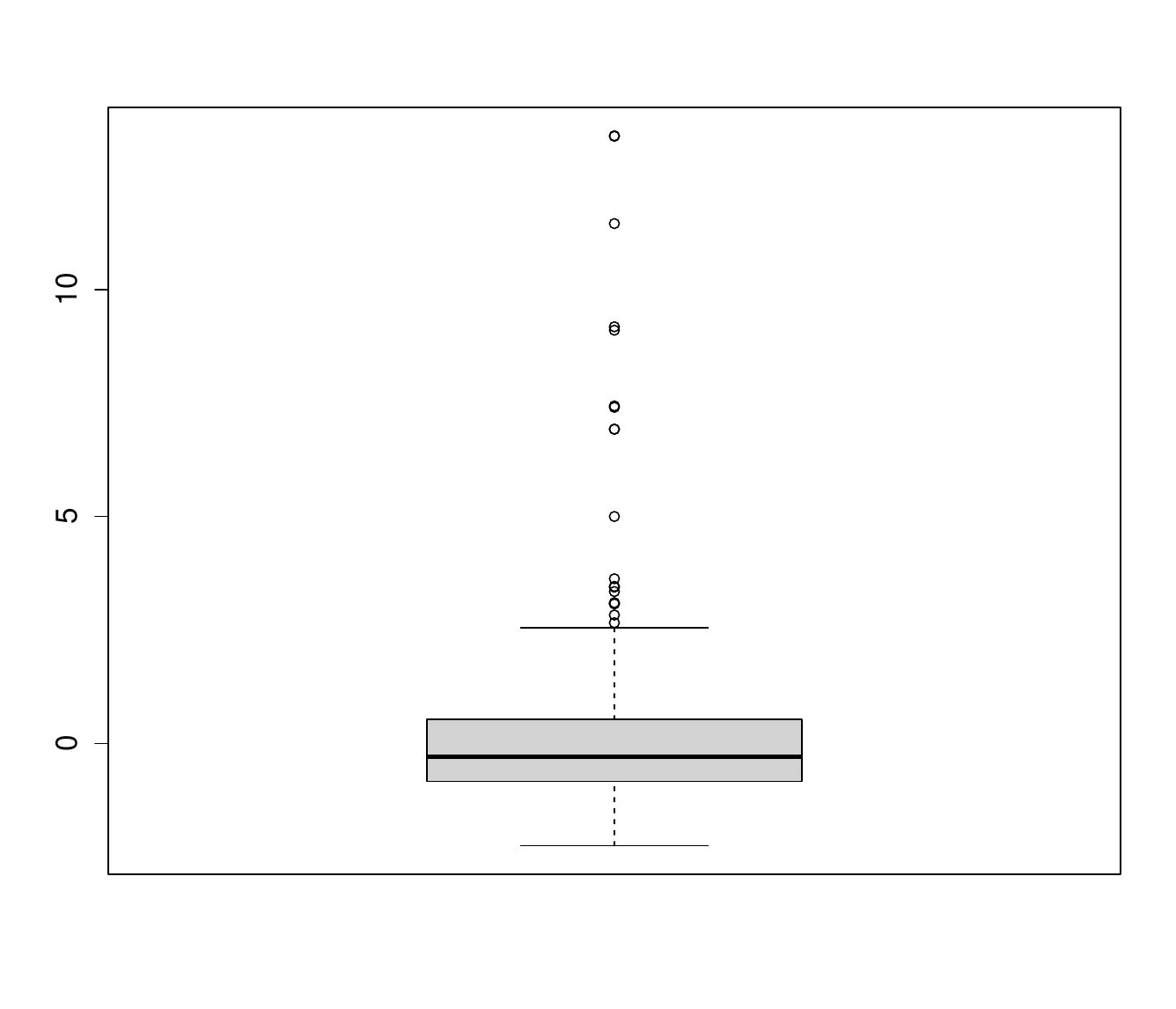}
\vskip-0.2in
\caption{\label{fig:boxplot-resid} Boxplot of the residuals obtained when fitting the complete data set using the penalized robust estimators.}
\end{center}
\end{figure}

To  identify the atypical observations that may have produced an increase on the the penalized least squares estimators  \textsc{mape},  we robustly estimate the parameters with the complete data set using the penalized approach introduced in Section \ref{sec:estimadores}.  The boxplot of the residuals  $r_i=Y_i-\wY_i$ displayed in Figure \ref{fig:boxplot-resid} shows the presence of 19 observations with large residuals, namely vertical outliers.

We repeat the prediction capability analysis of the penalized least squares approach after removing these vertical outliers and also the observation with high leverage  in alcohol consumption identified also  \citet{liu:etal:2011} and \citet{Guo:etal:2013}. The considered sample has then size 295 and, as above, we randomly chose a testing sample of size $n_{\mbox{\footnotesize \sc test}}=100$, while the remaining $n_{\mbox{\footnotesize \sc training}}=195$  correspond to the training sample to obtain the  \textsc{mape} for each of the replications.  The obtained results   are given in the last row of Table \ref{tab:mar-mape-av} and in Figure \ref{fig:MAPE-measure-lostres}, where the results for the  penalized least squares estimators computed without the atypical data are labelled \textsc{ls$^{(-\textsc{out})}$}. The obtained results and boxplot are quite similar to those corresponding to the penalized robust procedure using the whole data set, which confirms that the increase on the \textsc{mape} of the  least squares procedure,  previously described, is due by the effect of outliers.

\begin{figure}[ht!]
\begin{center}
\includegraphics[scale=0.45]{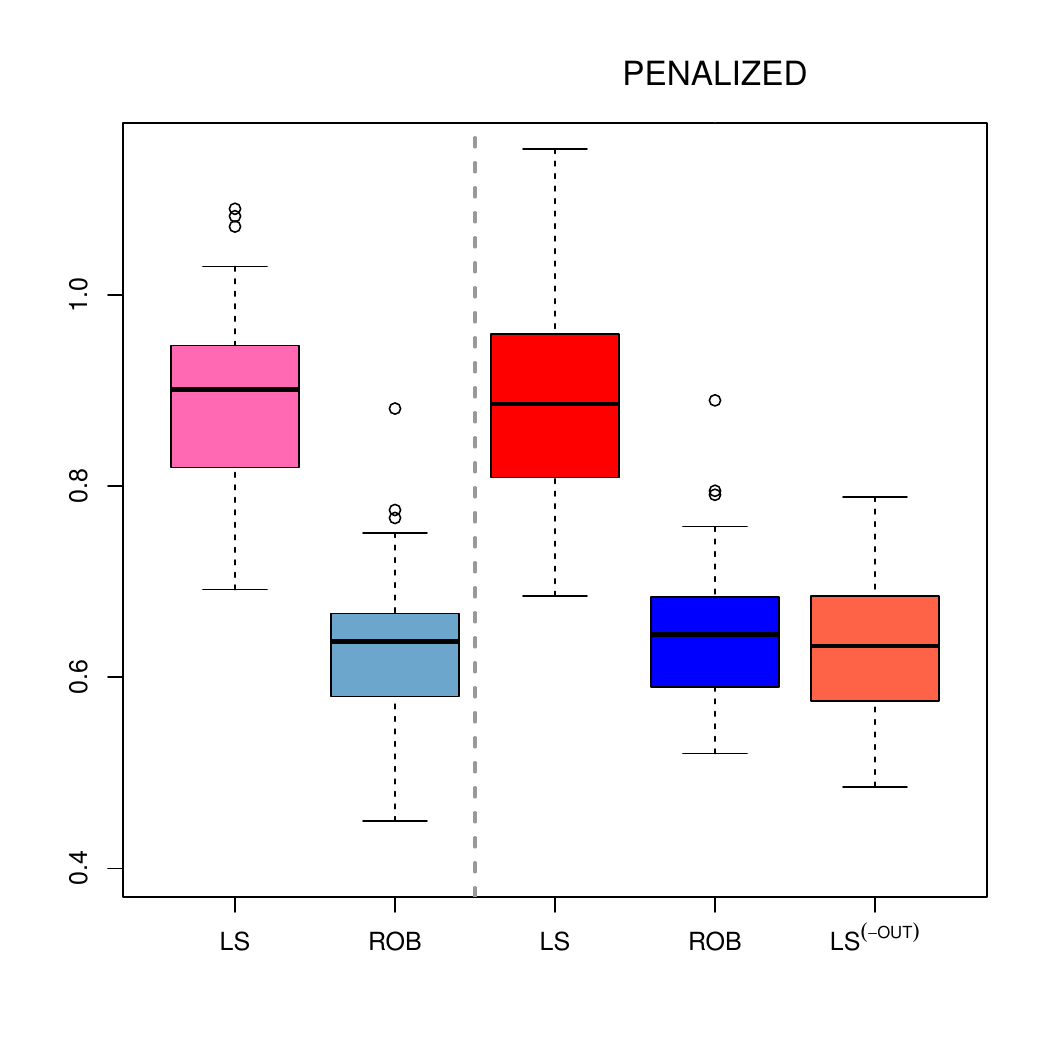}
\vskip-0.2in
\caption{\label{fig:MAPE-measure-lostres} Adjusted boxplots for the  \textsc{mape} measures obtained for the estimators without penalization (on the left) and for the penalized (on the right) estimators. The  \textsc{mape} corresponding to the penalized least squares estimator computed without the outliers is labelled \textsc{ls$^{(-\textsc{out})}$}.}
\end{center}
\end{figure}

\begin{table}[ht!]
\begin{center}
{\small{
\begin{tabular}{|l|c|c|c|c|c|c|c|c|c|c|}
  \hline
\textsc{penalized} &  \scriptsize{SEX}  &  \scriptsize{SMOK1}  &  \scriptsize{SMOK2}  &  \scriptsize{BMI}  &  \scriptsize{VIT1}  &  \scriptsize{VIT2}  &  \scriptsize{CAL}  & \scriptsize{FAT}  & \scriptsize{ALCOHOL}  &  \scriptsize{BETADIET}\\ 
  \hline
  \textsc{ls}      & \textbf{0.86} & 0.44 & \textbf{0.90} & {\textbf{1.00}} & {\textbf{1.00}} & {\textbf{0.86}} & 0.40 & \textbf{0.56} & \textbf{0.70} & {\textbf{0.98}}\\
\textsc{rob}     & {\textbf{0.72}} & 0.08 & 0.18 & {\textbf{1.00}} & {{0.40}} & 0.42 & 0.10 & {\textbf{0.64}} & 0.18 & {\textbf{0.96}} \\
   \hline
 \textsc{ls$^{(-\textsc{out})}$} & \textbf{0.70} & 0.08 & 0.40 & \textbf{1.00} & 0.44 & 0.32 & 0.06 & \textbf{0.74} & 0.44 & \textbf{1.00}
 \\
   \hline
\end{tabular}
}}
\caption{\label{tab:var-sel} Proportion of times each covariate is selected as active in the model for the penalized least squares and robust estimators. The last row corresponds to penalized least squares estimator computed after removing the detected vertical outliers and the   high leverage point in alcohol consumption.} 
\end{center}
\end{table}

Finally, in order to identify the covariates selected by each   penalized estimator    to be included in the model, Table \ref{tab:var-sel} presents the frequency of times that each variable was included in the model. We also include the results for the   penalized least squares estimator computed after removing the detected vertical outliers and the   high leverage point in alcohol consumption.  If one considers  a threshold of 0.5 to discard or include predictors, the \textsc{ls} approach with the whole sample selects   eight covariates   in the linear component, namely, SEX, SMOK2, BMI, VIT1, VIT2, FAT, ALCOHOL and BETADIET. In contrast, both the penalized robust procedure and penalized least squares approach after removing the atypical observations, \textsc{ls$^{(-\textsc{out})}$},   suggest  to include only the following four covariates: SEX, BMI, FAT and BETADIET. It is interesting to point out that SMOK2 and ALCOHOL variables are only chosen 18\%  of the times each of them by the robust proposal, while the least squares estimator with the complete sample selects them a 90\% and 70\% of the times, respectively. This may be due to the outlier present in alcohol consumption already discussed in \citet{liu:etal:2011} and \citet{Guo:etal:2013}.

Taking these observations into account,   the penalized robust proposal that selects variables in the linear regression component of the model seems an appropriate  choice when modelling this data set through the partially linear additive model \eqref{eq:betaplasma}.

\section{Concluding remarks}{\label{sec:coments}}
Partial additive linear regression models provide a useful tool to model a response when several covariables are present. The advantage over purely nonparametric ones is that they avoid the curse of dimensionality and make use of some preliminary information regarding the linear dependence on a subset of covariates.   When the linear regression coefficients are assumed to be sparse, i.e., when only a few explanatory variables included in the linear regression component are active,  the problem of joint estimation and automatic variable selection needs to be considered. In these circumstances, the statistical challenge of obtaining sparse and robust estimators  that are computationally feasible and provide variable selection should be complemented with the study of their asymptotic properties. 

In this paper, we have presented a family of estimators which are   reliable in the presence of atypical data and automatically selects variables. The regression coefficients are estimated through penalized $M-$regression estimators using preliminary estimators of the additive components and of the scale.
Consistency, rates of convergence and variable selection results are derived for   a broad family of penalty functions, which include  ADALASSO,  SCAD and MCP penalties.  The assumptions required to derived these results are very undemanding, which shows that these methods can be applied in very diverse contexts.    A robust    procedure to select the penalty parameter is also given.

The advantage of our proposal  over the classical one based on least squares is illustrated over a numerical study and the analysis of a real data set. In particular, the  results obtained in the simulation study illustrate that  robust methods   have a  performance similar to the classical ones   for clean samples and behave much better  in contaminated scenarios, showing greater reliability.   We exemplify our proposal on the   plasma beta-carotene level data set. The analysis shows that the robust  estimators automatically discard influential observations and select variables in a more reliable way.


\noi\textbf{\small Acknowledgements.} {\small  This research was partially supported by  the grant   \textsc{pict} 2021-I-A-00260 from \textsc{anpcyt} at  Argentina (Graciela Boente and Alejandra Mart\'{\i}nez). The research of Graciela Boente was also supported by grant 20020220200037\textsc{ba}    from the Universidad de Buenos Aires, Argentina  and the Spanish Project {MTM2016-76969P} from the Ministry of Economy, Industry and Competitiveness, Spain (MINECO/AEI/FEDER, UE). The research of Alejandra Mart\'{\i}nez was also supported by  the following Argentine grants  Proyecto  Interno CD-CBLUJ 204/19  from the Departamento de Ciencias B\'asicas, Universidad Nacional de Luj\'an and PICTO-2021-UNLU-00016 from \textsc{anpcyt} and Universidad Nacional de Luj\'an.}


\setcounter{equation}{0}
\def\theequation{A.\arabic{equation}}
  \setcounter{section}{0}
\renewcommand{\thesection}{\Alph{section}}
 
 \section{Appendix: Proofs}

\begin{proof}[Proof of \ref{teo:1}]
Note that by definition of $\wbbe=(\wbeta_1,\dots,\wbeta_q)\trasp$, we have that
\begin{align}
L_n\left(\wmu,\weta_1,\dots,\weta_p,\wsigma,\wbbe\right) &\leq L_n\left(\wmu,\weta_1,\dots,\weta_p,\wsigma,\wbbe\right) +\itJ_{\blach}(\wbbe)
\nonumber\\
&\leq L_n\left(\wmu,\weta_1,\dots,\weta_p,\wsigma,\bbe\right)+ \itJ_{\blach}(\bbe)\,.
\label{eq:ineq4.3}
\end{align}
 On the one hand, Lemma A.3 of \citet{boente:martinez:2023} implies that
\begin{equation}\label{eq:4.4}
\mathop{\sup_{ \varsigma >0, a\in \real, \bb \in \real^q}}_{g_1 \in \itS_{1}, \dots, g_p \in \itS_{p}}\left| L_n\left(a,g_1,\dots,g_p,\varsigma,\bb \right)-L\left(a,g_1,\dots,g_p,\varsigma,\bb \right)\right| \convpp 0\,,
\end{equation}
so $L_n\left(\wmu,\weta_1,\dots,\weta_p,\wsigma,\bbe\right)- L\left(\wmu,\weta_1,\dots,\weta_p,\wsigma,\bbe\right)\convpp 0$. On the other hand, assumptions \ref{ass:wsigma} and \ref{ass:weta} together with the Bounded Convergence Theorem  imply that $L\left(\wmu,\weta_1,\dots,\weta_p,\wsigma,\bbe\right)\convpp L\left(\mu,\eta_1,\dots,\eta_p,\sigma,\bbe\right)=\esp \rho_1\left(\varepsilon\right)$. Therefore,  the right hand of \eqref{eq:ineq4.3}  converges almost surely to $b_{\rho_1}=\esp \rho_1\left(\varepsilon\right) $, which implies that 
\begin{equation}
\limsup_{n\to\infty} L_n\left(\wmu,\weta_1,\dots,\weta_p,\wsigma,\wbbe\right) \leq b_{\rho_1} \qquad \mbox{ a.s. }
\label{eq:convergeLn}
\end{equation}
 
Fix $\delta>0$. It will be enough to show that  with probability one
\begin{equation}\label{eq:4.5}
\liminf_{n\to\infty} \inf_{\delta \leq \|\bb-\bbech\|} L_n\left(\wmu,\weta_1,\dots,\weta_p,\wsigma,\bb\right) > b_{\rho_1}\,,
\end{equation}
Indeed if \eqref{eq:4.5} holds, from \eqref{eq:convergeLn}, we get that for any $\delta>0$, $\prob(\exists n_0: \forall n\ge n_0 \; \|\wbbe-\bbe\|<\delta)=1 $, so $\wbbe\convpp\bbe$, as desired. 

The proof of \eqref{eq:4.5} follows the same arguments considered in the proof of Theorem 2.2.2 of \citet{smucler:tesis:2016}.
Note that $ L_n\left(\wmu,\weta_1,\dots,\weta_p,\wsigma,\bb\right)= \wA_n(\bb) + \wB_n(\bb)+\wC_n(\bb) + L\left(\mu,\eta_1,\dots,\eta_p,\sigma,\bb\right)$, where
\begin{align*}
\wA_n(\bb)& = L_n\left(\wmu,\weta_1,\dots,\weta_p,\wsigma,\bb\right)- L\left(\wmu,\weta_1,\dots,\weta_p,\wsigma,\bb\right)\,,\\
\wB_{n}(\bb)& = L\left(\wmu,\weta_1,\dots,\weta_p,\wsigma,\bb\right)- L\left(\mu,\eta_1,\dots,\eta_p,\wsigma,\bb\right)\,,\\
\wC_{n}(\bb)&= L\left(\mu,\eta_1,\dots,\eta_p,\wsigma,\bb\right)- L\left(\mu,\eta_1,\dots,\eta_p,\sigma,\bb\right)\,.
\end{align*}
Then, 
\begin{align}
\inf_{\delta \leq \|\bb-\bbech\|} L_n\left(\wmu,\weta_1,\dots,\weta_p,\wsigma,\bb\right) & \geq \inf_{\delta\leq\|\bb-\bbech\|}\wA_n(\bb) +\inf_{\delta\leq\|\bb-\bbech\|} \wB_n(\bb)+  \inf_{\delta\leq\|\bb-\bbech\|} \wC_n(\bb)
\nonumber\\
 & +\inf_{\delta\leq\|\bb-\bbech\|}L\left(\mu,\eta_1,\dots,\eta_p,\sigma,\bb\right)\,.
 \label{eq:cotainfLn}
\end{align}
Using \eqref{eq:4.4}, we have that $\sup_{\delta \leq \|\bb-\bbech\|} \left|\wA_n(\bb)\right| \convpp 0$, therefore 
\begin{equation}
\label{eq:wAnb}
\inf_{\delta\leq\|\bb-\bbech\|}\wA_n(\bb) \convpp 0\,.
\end{equation}
First note that for any $\varsigma$, 
$$L\left(\mu,\eta_1,\dots,\eta_p,\varsigma,\bb\right)= \esp \rho_1\left(\frac{\sigma \varepsilon + \bZ\trasp(\bbe-\bb)}{\varsigma}\right)\,.$$
Then, using a Taylor's expansion of order one and the fact that from assumption \ref{ass:rho_bounded_derivable}(b), $\zeta_1(t)=t\psi_1(t)$ is bounded, we get that
\begin{align*}
|\wC_{n}(\bb)| &\le  \esp  \left| \rho_1\left(\frac{\sigma \varepsilon + \bZ\trasp(\bbe-\bb)}{\wsigma}\right)-   \rho_1\left(\frac{\sigma \varepsilon + \bZ\trasp(\bbe-\bb)}{\sigma}\right)\right| \leq \|\zeta_1\|_{\infty}\left|\frac{\wsigma-\sigma}{\min(\sigma, \wsigma)}\right|\,.
\end{align*}
 Using that $\wsigma$ is a strong consistent estimator of $\sigma>0$, we obtain that  $\sup_{\delta \leq \|\bb-\bbech\|} |\wC_n(\bb)|\convpp 0$,  so  
\begin{equation}
\label{eq:wCnb} 
\inf_{\delta\leq\|\bb-\bbech\|} \wC_n(\bb) \convpp 0\,.
\end{equation}
Using again a Taylor's expansion of order one  and the fact that from assumption  \ref{ass:rho_bounded_derivable}(b)   $\psi_1$ is bounded, we obtain that
\begin{align*}
|\wB_{n}(\bb)| &\le  \esp \left|\rho_1\left(\frac{Y-\wmu-\sum_{j=1}^p \weta_j(X_{j})-\bZ\trasp\bb}{\wsigma}\right)-\rho_1\left(\frac{Y-\mu-\sum_{j=1}^p \eta_j(X_j)-\bZ\trasp\bb}{\wsigma}\right)\right|\\
&   \leq \|\psi_1\|_{\infty}\left|\frac{1}{\wsigma}\right|\esp \left|\wmu+\sum_{j=1}^p\weta_j(X_j)-\left(\mu+\sum_{j=1}^p \eta_j(X_j)\right)\right|  \leq \|\psi_1\|_{\infty} \frac{|\wmu-\mu| + \sum_{j=1}^p  \|\weta_j-\eta_j\|_{\infty}}{\wsigma}\,.
\end{align*}
Using again the consistency of $\wsigma$ given in assumption \ref{ass:wsigma} and the consistency of $\wmu$ and $\weta_j$ stated in \ref{ass:weta},   we get that  $\sup_{\delta \leq \|\bb-\bbech\|} |\wB_n(\bb)|$ converges almost surely to $0$,  that is, 
\begin{equation}
\label{eq:wBnb} 
\inf_{\delta\leq\|\bb-\bbech\|} \wB_n(\bb) \convpp 0\,.
\end{equation}
From \eqref{eq:cotainfLn}, \eqref{eq:wAnb}, \eqref{eq:wCnb} and \eqref{eq:wBnb}, we obtain that with probability 1
\begin{equation}\label{eq:cotainfLnlimite}
\liminf_{n\to\infty} \inf_{\delta \leq \|\bb-\bbech\|} L_n\left(\wmu,\weta_1,\dots,\weta_p,\wsigma,\bb\right) \ge \inf_{\delta\leq\|\bb-\bbech\|}L\left(\mu,\eta_1,\dots,\eta_p,\sigma,\bb\right)=D \,.
\end{equation}
The proof will be completed if we show that $D> b_{\rho_1}$. Suppose that  $D\leq  b_{\rho_1}$. Let $\{\bb_m\}_{m\ge 1}$ be a sequence  such that $\|\bb_m-\bbe\|\geq \delta$ for all $m$ and
$$\lim_{m\to \infty}L\left(\mu,\eta_1,\dots,\eta_p,\sigma,\bb_m\right)=D\,.$$
Assume that for some subsequence $m_k$, $\bb_{m_k}$ converges to a point  $\bbe^*$ such that $\|\bbe^*-\bbe\|\geq \delta$. The Bounded Convergence Theorem entails that 
$$D=\lim_{k\to \infty}L\left(\mu,\eta_1,\dots,\eta_p,\sigma,\bb_{m_k}\right)=L\left(\mu,\eta_1,\dots,\eta_p,\sigma,\bbe^*\right)\,.$$
 Therefore, using that $D\leq  b_{\rho_1}$, we conclude that
$$L\left(\mu,\eta_1,\dots,\eta_p,\sigma,\bbe^*\right) \leq  b_{\rho_1}=L\left(\mu,\eta_1,\dots,\eta_p,\sigma,\bbe\right)\,,$$
which contradicts the fact that $L\left(\mu,\eta_1,\dots,\eta_p,\sigma,\bb \right)  $ has a unique minimum at $\bb=\bbe$ as mentioned in Remark \ref{remark:comentarios}. Thus,   $\|\bb_m\|\to +\infty$. 

Denote $\bb_m^*= {\bb_m}/{\|\bb_m\|} $, then there exists a subsequence $\{m_k\}_{k\ge 1}$ such that $\bb_{m_k}^*\to \bbe^*$ with $\|\bbe^*\|=1$. It follows that
\begin{align}
 b_{\rho_1} &\geq D= \lim_{k\to \infty} L\left(\mu,\eta_1,\dots,\eta_p,\sigma,\bb_{m_k}\right)=\lim_{k\to \infty}  \esp\rho_1\left(\frac{\sigma\,\epsilon+  \bZ\trasp\left(\bbe-\bb_{m_k}^* \|\bb_{m_k}\|\right)}{\sigma}\right)\nonumber \\
&\geq  \liminf_{k\to \infty} \esp\rho_1\left(\epsilon +\frac{ \bZ\trasp\left(\bbe-\bb_{m_k}^* \|\bb_{m_k}\|\right)}{\sigma}\right)\,\buno_{\{\bZ\trasp\bb_{m_k}^*\neq 0\}} 
\nonumber
\\
& \geq  \esp  \liminf_{k\to \infty} \rho_1\left(\epsilon +\frac{ \bZ\trasp\left(\bbe-\bb_{m_k}^* \|\bb_{m_k}\|\right)}{\sigma}\right)\,\buno_{\{\bZ\trasp\bb_{m_k}^*\neq 0\}} \,, 
\label{eq:fatou}
\end{align}
where, in the last inequality, we have used Fatou's lemma. Recall that \ref{ass:rho_bounded_derivable}(a)  states that $\rho_1$ is an even function and $\lim_{t\to \infty} \rho_1(t)=\|\rho_1\|_\infty=1$, then the right hand side in \eqref{eq:fatou} equals $\prob(\bZ\trasp\bbe^*\neq 0)$, so we have
$$ b_{\rho_1} \geq D\ge \prob(\bZ\trasp\bbe^*\neq 0)\,,$$ 
which contradicts assumption \ref{ass:probaZ} which states that $\prob(\bZ\trasp\bbe^*\neq 0)> 1-c\ge  b_{\rho_1}$.  Then, we have that $D> b_{\rho_1}$, so from  \eqref{eq:cotainfLnlimite},  for any $\delta>0$,   \eqref{eq:4.5} holds, concluding the proof.
\end{proof}

Henceforth, we denote as $N(\epsilon, \itF, L_s(\qu))$ and $N_{[\;]}(\epsilon, \itF, L_s(\qu))$   the covering  and bracketing  numbers  of the class $\itF $ with respect to the distance in $ L_s(\qu)$ and as  $\|f\|_{\qu,2}=\left(\esp_{\qu}(f^2)\right)^{\frac 12}$.  For a class of functions $\itF$ with envelope $F$, define  the bracketing integral as 
$$J_{[\;]}(\delta, \itF, L_2(P)) =\int_0^\delta \sqrt{1+ \log N_{[\;]}(\delta, \itF, L_2(P)) } d\delta\,.$$
Recall that $\itL_1=\itC^{1}[0,1]$ corresponds to the space of continuously differentiable functions on $[0,1]$ with  norm $\|\eta\|_{\itL_1}=\max(\|\eta\|_{\infty},\|\eta^{\prime}\|_{\infty})$. 
From now on, we denote $\itV_{\itL_{1}, M}=\{\eta \in\itL_1: \|\eta\|_{\itL_1}\le M \}$   the ball of radius $M$. Theorem 2.7.1 in \citet{vanderVaart:wellner:1996} entails that  $\log N\left(\delta, \itV_{\itL_{1}, 1}, \|\cdot\|_{\infty} \right) \le   K\, ( {1}/{\delta})\,,$ 
where the constant $K$  is independent of $\delta$, so 
\begin{equation}
\label{eq:cubroVM}
\log N\left(\delta, \itV_{\itL_{1}, M}, \|\cdot\|_{\infty}  \right) \le   K\, \frac{M}{\delta}\,.
\end{equation}

In order to prove Theorem \ref{teo:rate} which derives consistency rates, we will need the following Lemmas.

\begin{lemma}\label{lema:5.1} Let $(Y_i,\bZ_i\trasp,\bX_i\trasp)\trasp$ be i.i.d. observations satisfying \eqref{eq:plam} with the errors $\varepsilon_i$ independent from the vector of covariates $(\bZ_i\trasp,\bX_i\trasp)\trasp$. Let $\rho_1$ be a function satisfying \ref{ass:rho_bounded_derivable} and \ref{ass:psi-dos-der}.   Assume \ref{ass:densidadepsilon} to \ref{ass:weta},  \ref{ass:zmom2} and  \ref{ass:wetatasa} hold.
Then, we have that
\begin{equation}{\label{eq:aprobar2}}
\wbW=\frac 1{\sqrt{n}} \sum_{i=1}^n \left\{\psi_1\left(\frac{Y_i-\wmu-\sum_{j=1}^p \weta_j(X_{ij})-\bZ_i\trasp \bbe}{\wsigma}\right)\,  \bZ_{i}- \bV(\wmu, \weta_1, \dots, \weta_p,\wsigma)\right\} =O_{\prob}(1)\,,
\end{equation}
where 
$$\bV(a, g_1, \dots, g_p, \varsigma)=\esp \psi_1\left(\frac{Y-a-\sum_{j=1}^p g_j(X_{j})-\bZ\trasp \bbe}{\varsigma}\right)\,   \bZ\,.$$
If, in addition \ref{ass:zmean0} holds, we have that 
\begin{equation}{\label{eq:aprobar}}
\frac 1{\sqrt{n}} \sum_{i=1}^n \psi_1\left(\frac{Y_i-\wmu-\sum_{j=1}^p \weta_j(X_{ij})-\bZ_i\trasp \bbe}{\wsigma}\right)\,  \bZ_i =O_{\prob}(1)\,.
\end{equation}
\end{lemma}

  \begin{proof} 
To prove \eqref{eq:aprobar2}, it is enough to show that, for each $1\le \ell\le q$, $\wW_\ell =O_{\prob}(1)$, where  $\wbW=(\wW_1, \dots, \wW_q)\trasp$.
For that purpose,  define the class  of functions
  \begin{align*}
  \itF_\ell &=\left\{f_{a,\bg, \varsigma}(y, \bx, \bz)= \psi_1\left(\frac{y-a-\sum_{j=1}^p g_j(x_{j})-\bz\trasp \bbe}{\varsigma}\right) \,  z_{\ell} :\right. \\
  &  \hskip.3in |\varsigma-\sigma|< \sigma/2 , |\mu-a|\le 1/2, g_j\in \itL_1, \|g_j-\eta_j\|_{\itL_1}\le M \Bigg	\}\,,
   \end{align*}
  where $ z_{\ell}$ is the $\ell-$th component of $\bz$, $\bx=(x_1,\dots,x_p)$, $\bg=(g_1,\dots, g_p)$ and   $M$ is the constant given in assumption \ref{ass:wetatasa}, that is, the constant such that 
\begin{equation}
\label{eq:wetaacotada}
\lim_{n\to \infty} \prob\left(\max_{1\le j\le p} \|\weta_j-\eta_j\|_{\itL_1}\le M\right)=1\,.
\end{equation}  
Note that $P f_{a,\bg, \varsigma} = V_{\ell}(a, g_1, \dots, g_p, \varsigma)$. The class $\itF_\ell$ has envelope $F(y, \bx, \bz)=  \|\psi_1\|_{\infty} \|\bz\|$ and $\|F\|_{L_2(P)}<\infty$ by  \ref{ass:zmom2}. Write for simplicity $f_{a,\bg, \varsigma}$ instead of $f_{a,\bg, \varsigma}(y, \bx, \bz)$. Then, given $f_{a,\bg, \varsigma}, f_{a_0,\bg_0, \varsigma_0}  \in \itF_\ell $, we have
\begin{align*}
\left|f_{a,\bg, \varsigma} -f_{a_0,\bg_0, \varsigma_0} \right| &  \le \left|f_{a,\bg, \varsigma}-f_{a,\bg, \varsigma_0}\right| + \left|f_{a,\bg, \varsigma_0}-f_{a_0,\bg_0, \varsigma_0}\right|\\
& \le | z_{\ell}|\;  \left|\psi_1\left(\frac{y-a-\sum_{j=1}^p g_j(x_{j})-\bz\trasp \bbe}{\varsigma}\right)- \psi_1\left(\frac{y-a-\sum_{j=1}^p g_{j}(x_{j})-\bz\trasp \bbe}{\varsigma_0}\right)\right|\\
& + | z_{\ell}|\; \left|\psi_1\left(\frac{y-a-\sum_{j=1}^p g_j(x_{j})-\bz\trasp \bbe}{\varsigma_0}\right)-\psi_1\left(\frac{y-a_0-\sum_{j=1}^p g_{0,j}(x_{j})-\bz\trasp \bbe}{\varsigma_0}\right)\right| \\
& \le | z_{\ell}|\; \left\{2\frac{\|\varphi_1\|_{\infty}}{\sigma} |\varsigma-\varsigma_0|+ 2\frac{\|\psi_1^{\prime}\|_{\infty}}{\sigma}\left(|a-a_0|+\sum_{j=1}^p \|g_j-g_{0,j}\|_{\infty}\right)\right\} \\
& \le   2\, \frac{\|\varphi_1\|_{\infty}+\|\psi_1^{\prime}\|_{\infty} }{\sigma} \; | z_{\ell}|\left(|\varsigma-\varsigma_0|+|a-a_0|+\sum_{j=1}^p \|g_j-g_{0,j}\|_{\infty}\right)
\end{align*} 
Then, if we denote $B_1= 2\, (\|\varphi_1\|_{\infty}+\|\psi_1^{\prime}\|_{\infty} )/{\sigma}$, we have that 
\begin{equation}
\label{eq:cotaf}
\left|f_{a,\bg, \varsigma} -f_{a_0,\bg_0, \varsigma_0} \right| \le B_1  \|\bz\| \left(|\varsigma-\varsigma_0|+|a-a_0|+\sum_{j=1}^p \|g_j-g_{0,j}\|_{\infty}\right)\,.
\end{equation}
Given $\delta>0$, take $\nu= \delta /B_2$  with $B_2=2 B_1  (p+2) \left(\esp\|\bZ\|^2\right)^{1/2}$ and denote $N_1=N\left(2\,\nu, \itV_{\itL_{1}, M}, \|\;\|_{\infty}  \right) $, $ N_2, N_3\in \natu$    be the integer part of $5/\nu$ and $5\sigma/(2\nu)$. Then, the sets  $\itI_2=\{ a \in \real:   |a-\mu|\le 1\}$ and $\itI_3=\{ \varsigma \in \real: |\varsigma-\sigma|\le \sigma/2\}$ can be     covered by $N_2$ and $N_3$ intervals of length at most $\nu$.

 Denote as $\itG_j=\{ g_j\in \itL_1, \|g_j-\eta_j\|_{\itL_1}\le M\}$. Clearly, $g\in \itG_j$ if and only if $g_j-\eta_j\in \itV_{\itL_{1}, M}$. Hence, we can   take $ g_{j,1}, \dots, g_{j,N_1}$ in  $\itG_j$  such that $ \itG_j\subset \cup_{s=1}^{N_1} \{g\in \itL_1: \|g-g_{j,s}\|_{\infty}\le \nu\}$. Similarly, choose $a_1,\dots, a_{N_2}\in \itI_2$ and $\varsigma_1,\dots, \varsigma_{N_2}\in \itI_3$ such that $\itI_2\subset  \cup_{s=1}^{N_2} \{ a \in \real:   |a-a_s|\le \nu\}$ and  $\itI_3\subset  \cup_{s=1}^{N_3} \{ \varsigma>0:   |\varsigma-\varsigma_s|\le \nu\}$. Then, for any  $f_{a,\bg, \varsigma}\in \itF_{\ell}$, there exist $1\le s_a\le N_2$, $1\le s_j\le N_1$, for $1\le j\le p$ and $1\le s_{\varsigma}\le N_3$ such that 
$|a-a_{s_a}|\le \nu$, $\|g_j-g_{j, s_j}\|_{\infty}\le \nu$ and $|\varsigma- \varsigma_{s_{\varsigma}}|\le \nu$, so if we denote as
$\wtf= f_{a_{s_a},\bg_{\bese}, \varsigma_{s_{\varsigma}}}$ with $\bg_{\bese}=(g_{1, s_1}, \dots, g_{p, s_p})$, from \eqref{eq:cotaf} we have that
$$\left|f_{a,\bg, \varsigma} -\wtf \right| \le B_1  (p+2) \nu \;  \|\bz\| \,. $$
Therefore, if we define $\wtf_{\mbox{\sc \scriptsize u}}= \wtf +  B_1  (p+2) \nu \;  \|\bz\| $ and $\wtf_{\mbox{\sc \scriptsize l}}= \wtf -  B_1  (p+2) \nu \;  \|\bz\| $, we have that $\wtf_{\mbox{\sc \scriptsize l}}\le f_{a,\bg, \varsigma} \le \wtf_{\mbox{\sc \scriptsize u}}$ and 
$$\| \wtf_{\mbox{\sc \scriptsize u}}- \wtf_{\mbox{\sc \scriptsize l}}\|_{L_2(P)}= 2 B_1  (p+2) \nu \;  \left(\esp\|\bZ\|^2\right)^{1/2}=\delta\,.$$
Therefore,
$$N_{[\;\;]}\left(\delta,\itF_\ell,L^2(P)\right)\le N_1^p \, N_2 \, N_3\le  N\left(\frac{2\,\delta}{B_2}, \itV_{\itL_{1}, M}, \|\cdot\|_{\infty}  \right)^p \frac{25\sigma\, B_2^2}{2\,\delta^2}\,,$$
which together with \eqref{eq:cubroVM} imply that 
$$\log N_{[\;\;]}\left(\delta,\itF_\ell,L^2(P)\right)\le  K\,p\, \frac{M\, B_2}{2\, \delta}+ \log\left(\frac{25\sigma\, B_2^2}{2}\right)+2\log\left(\frac{1}{\delta}\right)\,.$$
Using that $\log(p)\le p$ for $p\ge 1$, we get that for $\delta<1$,
$$\log N_{[\;\;]}\left(\delta,\itF_\ell,L^2(P)\right)\le B_3 \frac{1}{\delta}  $$
with $B_3= K\,p\, {M\, B_2}/2+2  + \log\left( {25\sigma\, B_2^2}/{2}\right)$ which implies that $ J_{[\;\;]}\left(1,\itF_\ell,L^2(P)\right)<\infty$.
Theorem 2.14.2 in \citet{vanderVaart:wellner:1996} implies that for some universal constant $A$,
$$\esp\left(\sqrt{n}\sup_{f\in\itF_{\ell}}|(P_n-P)f|\right)\le A \; J_{[\;\;]}\left(1,\itF_\ell,L^2(P)\right)\|F\|_{L_2(P)}= A_0<\infty\,.$$
Denote $\itA_n=\{\max_{1\le j\le p} \|\weta_j-\eta_j\|_{\itL_1}< M\;,\; |\wsigma-\sigma|\le \sigma/2\;,\; |\wmu-\mu|\le 1\}$. Given $\delta>0$ from \eqref{eq:wetaacotada} the consistency of $\wsigma$ and of $\wmu$, we get that there exists $n_0\in \natu$ such that,   for any $n\ge n_0$,
$\prob\left(\itA_n\right)\ge 1-\delta/2$. Hence, taking into account that in $\itA_n$, $|\wW_\ell|\le \sqrt{n} \sup_{f\in\itF_{\ell}}|(P_n-P)f|$, we obtain that if $C>  2\, A_0/\delta$
\begin{align*}
\prob\left(|\wW_\ell| > C \right) & \le \prob\left(\sqrt{n} \sup_{f\in\itF_{\ell}}|(P_n-P)f|>C \cap \itA_n\right)+ \frac{\delta}{2}\\
& \le \prob\left(\sqrt{n} \sup_{f\in\itF_{\ell}}|(P_n-P)f|>C \right)+ \frac{\delta}{2}\\
& \le \frac{1}{C}\esp\left(\sqrt{n} \sup_{f\in\itF_{\ell}}|(P_n-P)f| \right)+ \frac{\delta}{2}\\
& \le A_0 \frac{1}{C} + \frac{\delta}{2}\le \delta\,,
\end{align*}
which concludes the proof of \eqref{eq:aprobar2}.

Hence, to prove \eqref{eq:aprobar} it is enough to show that
\begin{equation}{\label{eq:aprobar3}}
  {\sqrt{n}}  V(\wmu, \weta_1, \dots, \weta_p,\wsigma) =O_{\prob}(1)\,.
\end{equation}
Denote $\bh(\bX)=\esp(  \bZ| \bX )$, the independence between the errors and the covariates imply that $\esp(  \bZ| (\bX, \varepsilon) )=\esp(  \bZ| \bX )$. Then,  we have that 
\begin{align*}
V_{\ell}(a, g_1, \dots, g_p, \varsigma)& = \esp  \psi_1\left(\frac{\sigma \varepsilon +\mu-a+ \sum_{j=1}^p \left(\eta_j(X_j)- g_j(X_{j})\right)}{\varsigma}\right)\,   Z_{\ell} \\
& =\esp \left\{ \psi_1\left(\frac{\sigma \varepsilon +\mu-a+ \sum_{j=1}^p \left(\eta_j(X_j)- g_j(X_{j})\right)}{\varsigma}\right)\, \esp(  Z_{\ell}|(\bX,\varepsilon))\right\}\\
& =\esp \left\{ \psi_1\left(\frac{\sigma \varepsilon +\mu-a+ \sum_{j=1}^p \left(\eta_j(X_j)- g_j(X_{j})\right)}{\varsigma}\right)\, h_{\ell}(\bX)\right\} \,.
\end{align*}
From assumption \ref{ass:zmean0},  $\bh(\bX)= \bcero_q$, so \eqref{eq:aprobar3} holds  concluding the proof.
\end{proof}

 \begin{lemma}\label{lema:An} Let $(Y_i,\bZ_i\trasp,\bX_i\trasp)\trasp$ be i.i.d. observations satisfying \eqref{eq:plam} where the errors $\varepsilon_i$ are independent of the covariates $(\bZ_i\trasp,\bX_i\trasp)\trasp$. Let $\rho_1$ be a function satisfying \ref{ass:rho_bounded_derivable} and \ref{ass:psi-dos-der} and assume that \ref{ass:wsigma},   \ref{ass:weta} and \ref{ass:zmom2} hold. Then, for any random sequence  $\wtbbe_n\convprob \bbe$, we have that $\bA_n(\wtbbe_n)\convprob \bA$, where
 \begin{align}
\label{eq:matrizA}
 \bA  & =   \frac{1}{\sigma^2}\esp \psi_1^\prime\left(\varepsilon\right)  \esp\left(\bZ\bZ\trasp\right)=\frac{1}{\sigma^2}\esp \psi_1^\prime\left(\varepsilon\right)  \bV_{\bz} \\
  \label{eq:matrizAn}
 \bA_n(\bb) & =   \frac{1}{\wsigma^2}\frac{1}{n} \sum_{i=1}^n\psi_1^\prime\left(\frac{Y_i-\wmu-\sum_{j=1}^p \weta_j(X_{ij})-\bZ_i\trasp \bb}{\wsigma}\right)  \,\bZ_i \bZ_i \trasp\,.
 \end{align}
\end{lemma}

\begin{proof}
By the consistency of $\wsigma$, it will be enough to show that $\bB_n(\wtbbe)\convprob \bB$, where $\bB=\esp \psi_1^\prime\left(\varepsilon\right)  \bV_{\bz} $ and
$$\bB_n(\bb)=\frac{1}{n} \sum_{i=1}^n\psi_1^\prime\left(\frac{Y_i-\wmu-\sum_{j=1}^p \weta_j(X_{ij})-\bZ_i\trasp \bb}{\wsigma}\right)  \,\bZ_i \bZ_i \trasp\;.$$
We will show the convergence component--wise, for that reason, given $1\le \ell, s\le q$, denote $\bB_{n,\ell,s}(\bb)$ and $\bB_{\ell,s}$ the $(\ell,s)-$components of $\bB_n(\bb)$ and $\bB$, respectively.

Note that $\bB_{n,\ell,s}(\wtbbe)= \bB_{n,\ell,s}^{(1)}+ \bB_{n,\ell,s}^{(2)} $
\begin{align*}
 \bB_{n,\ell,s}^{(1)}& =  \frac{1}{n} \sum_{i=1}^n\psi_1^\prime\left(\frac{Y_i-\mu-\sum_{j=1}^p \eta_j(X_{ij})-\bZ_i\trasp \wtbbe}{\wsigma}\right)  \, Z_{i,\ell}  Z_{i,s}=\frac{1}{n} \sum_{i=1}^n\psi_1^\prime\left(\frac{\sigma \varepsilon_i+ \bZ_i\trasp(\bbe- \wtbbe)}{\wsigma}\right)  \, Z_{i,\ell}  Z_{i,s}\\
\bB_{n,\ell,s}^{(2)} & = \frac{1}{n} \sum_{i=1}^n\left\{\psi_1^\prime\left(\frac{Y_i-\wmu-\sum_{j=1}^p \weta_j(X_{ij})-\bZ_i\trasp \wtbbe}{\wsigma}\right)-\psi_1^\prime\left(\frac{Y_i-\mu-\sum_{j=1}^p \eta_j(X_{ij})-\bZ_i\trasp \wtbbe}{\wsigma}\right) \right\}  \, Z_{i,\ell}  Z_{i,s} 
\end{align*}
Assumption \ref{ass:psi-dos-der} entails that
$$\left|\bB_{n,\ell,s}^{(2)}\right|\le K_{\psi^{\prime}} \frac{1}{\wsigma}\left\{|\wmu-\mu|+\sum_{j=1}^p \|\weta_j-\eta_j\|_{\infty}\right\} \frac{1}{n} \sum_{i=1}^n \left| Z_{i,\ell}  Z_{i,s} \right |\,,$$
where $ K_{\psi^{\prime}}$ stands for the Lipschitz constant of $\psi^{\prime}$. Then, using assumptions \ref{ass:wsigma} and \ref{ass:weta} and the fact that $\esp\|\bZ\|^2<\infty$, we get that $\bB_{n,\ell,s}^{(2)}\convpp 0$.

To show that $ \bB_{n,\ell,s}^{(1)} \convprob \bB_{\ell,s}$   consider the class of functions
\begin{equation}{\label{eq:claseFnpsi}}
\itF=\left\{ f(\varepsilon, \bz)=\psi_1^\prime\left( \frac{\sigma \varepsilon+ \bz\trasp \bb}{\varsigma}\right) \, z_{\ell}  z_{s}\;,\;  \bb\in\real^q\;,\;  \|\bb\|\le 1\;,\; \varsigma\in  \left[\frac{\sigma}2, 2\sigma\right]\right\} \;.
\end{equation}
Using that from \ref{ass:zmom2} $\esp \|\bZ\|^2 <\infty$, the continuity of $\psi_1^\prime$ and the fact that $\Theta=\{(\varsigma,  \bb):  \bb\in\real^q, \|\bb\|\le 1, \varsigma\in [ {\sigma}/2, 2\sigma]\}$ is compact, we immediately obtain from Lemma 3.10 in \citet{sara2000} that 
$$ \sup_{f\in \itF } \left|P_n f - P f\right|\convpp 0 \,,$$
where     $P_n$ the empirical distribution of $(\varepsilon_i,\bZ_i\trasp)\trasp$. Therefore, taking into account that $\wsigma \convpp \sigma$ and $\bbe- \wtbbe\convprob 0$, we obtain that
\begin{equation}\label{eq:convMn}
 \frac{1}{n} \sum_{i=1}^n\psi_1^\prime\left(\frac{\sigma \varepsilon_i+ \bZ_i\trasp(\bbe- \wtbbe)}{\wsigma}\right)  \, Z_{i,\ell}  Z_{i,s}- M(\bbe- \wtbbe,\wsigma)\convprob 0\,,
 \end{equation}
where $M(\bb, \varsigma)=\esp\left\{\psi_1^\prime\left(\left({\sigma \varepsilon+ \bZ\trasp\bb)}\right)/{\varsigma}\right)  \, Z_{ \ell}  Z_{s}\right\}$.

The Dominated Convergence Theorem together with the independence between the errors and the covariates imply that $\lim_{\bb \to \bcero, \varsigma\to \sigma}M(\bb, \varsigma)= \bB_{\ell,s}$, thus  using the fact that $\bbe- \wtbbe\convprob 0$ and $\wsigma\convpp \sigma$ and \eqref{eq:convMn}, we conclude the proof.
\end{proof}

\vskip0.1in
\begin{proof}[Proof of Theorem \ref{teo:rate}]
We will begin the proof considering both the situation of twice differentiable and  ADALASSO and then, when needed, we will indicate the different approaches to be taken into account for both type of penalties.

 For the sake of simplicity, we denote
\begin{equation}
\label{eq:wELEn}
\wELE_n(\bb)=L_n\left(\wmu,\weta_1,\dots,\weta_p,\wsigma, \bb\right)\,,
\end{equation}
where with $L_n$ is defined in \eqref{eq:funcionLn}. Then,  
$$PL_{n,\blach_n}(\bb)=L_n\left(\wmu,\weta_1,\dots,\weta_p,\wsigma, \bb\right)+\itJ_{\blach_n}(\bb)=\wELE_n(\bb)+\itJ_{\blach_n}(\bb)\,,$$ 
and we have strength the dependence of $\bla$ on $n$. Using a Taylor's expansion of order 2 of  $\wELE_n(\bb)$ around $\bbe$, we get
\begin{equation*}
\wELE_n(\wbbe) = \wELE_n(\bbe) + (\wbbe - \bbe)\trasp \nabla \wELE_n(\bbe) + \frac{1}{2} (\wbbe - \bbe)\trasp \bA_n(\wtbbe_n)(\wbbe - \bbe) \,,
\end{equation*}
 where $\wtbbe_n=\bbe + \tau_n(\wbbe - \bbe)$ is an intermediate point between $\bbe$ and $\wbbe$,  $\tau_n \in [0,1]$, $\nabla \wELE_n(\bb)$ is the gradient of the function $\wELE_n(\bb)$  given by 
$$ \nabla \wELE_n(\bb)= - \frac{1}{\wsigma}\frac 1n \sum_{i=1}^n \psi_1\left(\frac{Y_i-\wmu-\sum_{j=1}^p \weta_j(X_{ij})-\bZ_i\trasp \bb}{\wsigma}\right)\,  \bZ_i $$
 and   $\bA_n(\bb)$  defined in \eqref{eq:matrizAn} corresponds to  the  Hessian of $\wELE_n(\bb)$. 
 
Let   $\delta$ be a fixed positive constant and note that $\wtbbe_n\convpp \bbe$, since $\wbbe$ is a consistent estimator of $\bbe$.
Then, 
 Lemma \ref{lema:An} entails that $\bA_n(\wtbbe_n)\convprob \bA$, while assumption \ref{ass:zmom2} implies that the smallest eigenvalue $\xi_1$ of $\bA$ is strictly positive. Therefore, if we denote as  $\itA_n = \left \{\|\bA_n(\wtbbe) - \bA\| < \xi_1 / 2 \right \}$, we get that  there exists $n_1\in \natu$ such that for every $n \geq n_1$, $\prob (\itA_n) > 1 - \delta/4$. Hence, in $\itA_n$, we have the lower bound
 \begin{align*}
 (\wbbe - \bbe)\trasp \bA_n(\wtbbe_n)(\wbbe - \bbe) & =  (\wbbe - \bbe)\trasp \bA (\wbbe - \bbe) +  (\wbbe - \bbe)\trasp\left( \bA_n(\wtbbe_n)-\bA\right)(\wbbe - \bbe) 
\\
 & \ge  \xi_1\, \|\wbbe - \bbe\|^2 -  \|\wbbe - \bbe\|^2 \|\bA_n(\wtbbe)   - \bA\| \ge \frac{1}{2} \xi_1\, \|\wbbe - \bbe\|^2\,,
 \end{align*}
which, together with the fact that $PL_{n,\blach_n}(\wbbe)\le PL_{n,\blach_n}(\bbe)$, leads to
\begin{align}
0 &\geq PL_{n,\blach_n}( \wbbe) -PL_{n,\blach_n}(\bbe)=\wELE_n(\wbbe)+\itJ_{\blach_n}(\wbbe)-\wELE_n(\bbe) -\itJ_{\blach_n}(\bbe) \nonumber \\
&\ge  (\wbbe  - \bbe )\trasp \nabla \wELE_n(\bbe) +   \frac{\xi_1}{2} \|\wbbe - \bbe\|_2^2 +\itJ_{\blach_n}(\wbbe) -  \itJ_{\blach_n}(\bbe)\,. \label{eq:rate_ineq}
\end{align}
Note that $|(\wbbe  - \bbe )\trasp \nabla \wELE_n(\bbe)|\le   \|\wbbe_n - \bbe_0\| \|  \nabla \wELE_n(\bbe) \|$, so $ (\wbbe  - \bbe )\trasp \nabla \wELE_n(\bbe) \ge - \,\|\wbbe - \bbe\| \|  \nabla \wELE_n(\bbe) \|$ and from \eqref{eq:rate_ineq},  we get
\begin{align}
0 &\geq PL_{n,\blach_n}( \wbbe) -PL_{n,\blach_n}(\bbe) \ge - \,\|\wbbe - \bbe\|\; \|  \nabla \wELE_n(\bbe) \| +   \frac{\xi_1}{2} \|\wbbe - \bbe\|_2^2 +\itJ_{\blach_n}(\wbbe) -  \itJ_{\blach_n}(\bbe) 
\label{eq:rate_ineq2}
\end{align}

Lemma \ref{lema:5.1} entails that $\sqrt{n} \;\nabla \wELE_n(\bbe)  = O_\prob(1)$, so there exists a constant $M_1$ such that, for all $n$, 
$\prob(\itB_n) > 1- \delta/4$, where $\itB_n = \{\|\sqrt{n}\; \nabla \wELE_n(\bbe) \| < M_1\}$. 
Therefore, using \eqref{eq:rate_ineq2}, we get that in $\itA_n\cap \itB_n$, 
\begin{align}
0 &\geq PL_{n,\blach_n}( \wbbe) -PL_{n,\blach_n}(\bbe)\ge  -  \|\wbbe  - \bbe \| \frac{1}{\sqrt{n}} \|\sqrt{n}\; \nabla \wELE_n(\bbe) \| +   \frac{\xi_1}{2} \|\wbbe - \bbe\|_2^2 +\itJ_{\blach_n}(\wbbe) -  \itJ_{\blach_n}(\bbe)
\nonumber\\
& \ge  -  \|\wbbe - \bbe\| \frac{M_1}{\sqrt{n}} +   \frac{\xi_1}{2} \|\wbbe - \bbe\|_2^2 +\itJ_{\blach_n}(\wbbe) -  \itJ_{\blach_n}(\bbe)\,.
\label{eq:rate_ineq3}
\end{align}
Without loss of generality we assume that the first $k$ components of $\bbe$ are non--null and the remaining ones are $0$, that is, $\bbe = (\bbe_{\act}\trasp, \bcero_{q-k}\trasp)\trasp$ and $\bbe_{\act} \in \real^k$ corresponds to the  vector with active coordinates of $\bbe$. Taking into account that $\itJ_{\blach_n}(\bb)=\sum_{s=1}^q p_{\lambda_{n,s}}(|b_s|)$. and $p_{\lambda}(0)=0$  and $p_{\lambda}(t)\ge 0$, for $t\ge 0$,  we get that 
$$\itJ_{\blach_n}(\wbbe) - \itJ_{\blach_n}(\bbe)= \sum_{s = 1}^k p_{\lambda_{n,s}}(|\wbeta_{s}|) - p_{\lambda_{n,s}}(|\beta_{s}|) + \sum_{s = k+1}^q p_{\lambda_{n,s}}(|\wbeta_{s}|) \ge \sum_{s = 1}^k p_{\lambda_{n,s}}(|\wbeta_{s}|) - p_{\lambda_{n,s}}(|\beta_{s}|)\,,$$
which together with \eqref{eq:rate_ineq3} leads to
\begin{equation}
0 \geq - \|\wbbe  - \bbe \|  \frac{1}{\sqrt{n}} M_1 + \frac{\xi_1}{2} \|\wbbe_n - \bbe\|^2 + \sum_{s = 1}^k p_{\lambda_{n,s}}(|\wbeta_{s}|) - p_{\lambda_{n,s}}(|\beta_{s}|).
\label{eq:rate_ineq4} 
\end{equation} 
\underline{Let us proceed to derive (a).}
 Let $\nu$ be such that $b_n(\nu)\convprob 0$ and define the sets  $\itC_{n,1}$ and $\itC_{n,2}$  as $\itC_{n,2} =  \{\|\wbbe  - \bbe \|  \leq \nu \}$    and
$$\itC_{n,1}=\left\{b_n(\nu)=\sup\{|p_{\lambda_n,s}^{\prime\,\prime}(|\beta_{s}| + \tau \nu)| : \tau \in [-1,1] \;, \;  1 \leq s \leq q \;\; \text{and} \;\; \beta_{s} \neq 0 \} \leq \frac{\xi_1}{2}\right\}.$$

Using that $\wbbe \convprob \bbe$ and $b_n(\nu)\convprob 0$,  we can choose $n_2\in \natu$ such that for every $n \geq n_2$, $\prob(\itC_{n,1})>1-\delta/4$ and $\prob(\itC_{n,2}) \geq 1 - \delta/4$. Let $\itC_n=\itC_{n,1} \cap\itC_{n,2}$, then $\prob(\itC_n)\ge 1-\delta/2$.

Using a second order Taylor's  expansion, we have
\begin{equation*}
 p_{\lambda_{n,s}}(|\wbeta_{s}|) - p_{\lambda_{n,s}}(|\beta_{s}|) = p_{\lambda_{n,s}}^{\prime} (|\beta_{s}|)  ( |\wbeta_{s}|  - | \beta_{s}|) + \frac{1}{2}  p_{\lambda_{n,s}}^{\prime\,\prime}(\theta_{n,s}) ( |\wbeta_{s} | -  |\beta_{s}| )^2\,,
\end{equation*} 
where $\theta_{n,s}$ lies between $|\wbeta_{s}| $ and $|\beta_{s}|$.

 Using that $|\, |a|-|b|\,|\le |a-b|$, $ p_{\lambda_{n,s}}^{\prime} (|\beta_{s}|) \ge 0$ and that in the event $\itA_n \cap \itB_n \cap \itC_n$, $| p_{\lambda_{n,s}}^{\prime\,\prime} (\theta_{n,s})|\le \xi_1/2$, since $\max(0, |\beta_{s}|-\nu) <\theta_{n,s} \le |\beta_{s}|+\nu$, we get that 
\begin{eqnarray}
\itJ_{\blach_n}(\wbbe) - \itJ_{\blach_n}(\bbe) &\ge & \sum_{s=1}^k p_{\lambda_{n,s}}(|\wbeta_{s}|) - p_{\lambda_{n,s}}(|\beta_{s}|)  
\nonumber
\\
 & \ge &  - \sum_{s=1}^k p_{\lambda_{n,s}}^{\prime} (|\beta_{s}|)  |\wbeta_{s}   -   \beta_{s}| -
 \frac{1}{2} \sum_{s=1}^k |p_{\lambda_{n,s}}^{\prime\,\prime}(\theta_{n,s})| (  \wbeta_{s}  -   \beta_{s}  )^2 
 \nonumber
\\
 & \ge & - a_n \sum_{s=1}^k |\wbeta_{s}   -   \beta_{s}|- \frac{\xi_1}{4}  \sum_{s=1}^k (  \wbeta_{s}  -   \beta_{s}  )^2
\nonumber
\\
  & \ge & - a_n \sqrt{k}\; \|\wbbe - \bbe\| - \frac{\xi_1}{4} \|\wbbe - \bbe\|^2\,. 
  \label{eq:cotaJinf} 
\end{eqnarray}
Hence, \eqref{eq:rate_ineq4}  and \eqref{eq:cotaJinf} imply that
\begin{equation*}
0 \geq - \|\wbbe  - \bbe \|  \frac{1}{\sqrt{n}} M_1 + \frac{\xi_1}{2} \|\wbbe  - \bbe \| ^2 - a_n \sqrt{k}\; \|\wbbe - \bbe\| - \frac{\xi_1}{4} \|\wbbe - \bbe\|^2= - \|\wbbe  - \bbe \|  \left(\frac{1}{\sqrt{n}} M_1 +a_n \sqrt{k}\right) + \frac{\xi_1}{4} \|\wbbe  - \bbe \| ^2  \;, 
\end{equation*} 
that is,
\begin{equation*}
0 \geq  \frac{\xi_1}{4} \|\wbbe  - \bbe \|  -  \left(  \frac{1}{\sqrt{n}} M_1 +  a_n \sqrt{k}\right)\ge \frac{\xi_1}{4} \|\wbbe  - \bbe \|  -  \left(  \frac{1}{\sqrt{n}}+  a_n \right)\left( M_1 +\sqrt{k}\right)\;   \;, 
\end{equation*} 
which implies that in $\itA_n \cap \itB_n \cap \itC_n$, we have $   \|\wbbe_n - \bbe\|\le 4\alpha_n (M_1 + \sqrt{k})/\xi_1 $, where $\alpha_n= a_n + n^{-1/2}$.  The result follows now from the fact that, for $n \geq \max( n_1,n_2)$, 
$\prob(\itA_n\cap \itB_n \cap \itC_n ) \geq 1 - \delta$.

\underline{Let us proceed to derive (b).} Taking into account that $\beta_s\ne 0$, for $1\le s\le k$, we have that $ A_{\bbech}=\min_{1\le s\le k} |\beta_s|/2$ is positive. Let $0<\nu<  A_{\bbech}/2$,  using that $\wbbe \convprob \bbe$ and $\wbbe_{\ini} \convprob \bbe$,  we can choose $n_3\in \natu$ such that for every $n \geq n_3$, $\prob(\itC_{n})>1-\delta/4$ where $\itC_n=  \{\|\wbbe  - \bbe \| + \|\wbbe_{\ini}  - \bbe \|\leq \nu \}$. Then, in $\itC_n$,  for any  $1\le s\le k$, we have that $ |\wbeta_{\ini,s}|\ge A_{\bbech}/2$ and 
\begin{equation*}
 p_{\lambda_{n,s}}(|\wbeta_{s}|) - p_{\lambda_{n,s}}(|\beta_{s}|) = \iota_n  \frac{|\wbeta_{s}|-|\beta_{s}|}{ |\wbeta_{\ini,s}|} \ge - \iota_n  \frac{|\wbeta_{s} -\beta_{s}|}{ |\wbeta_{\ini,s}|}\ge - 2\,\frac{\iota_n}{ A_{\bbech}} \|\wbbe- \bbe\|
\end{equation*}  
where we have used again that  $|\, |a|-|b|\,|\le |a-b|$. Hence, we obtain that,  in $\itC_n$, 
\begin{eqnarray}
\sum_{s=1}^k p_{\lambda_{n,s}}(|\wbeta_{s}|) - p_{\lambda_{n,s}}(|\beta_{s}|)  \ge - 2\, \frac{\iota_n}{ A_{\bbech}} \,k\, \|\wbbe- \bbe\|
  \label{eq:cotaadaptiveinf} 
\end{eqnarray}
Hence, using that in $\itA_n \cap \itB_n \cap \itC_n$, \eqref{eq:rate_ineq4}  and \eqref{eq:cotaadaptiveinf} hold, we get   that, in $\itA_n \cap \itB_n \cap \itC_n$,
\begin{equation*}
0 \geq - \|\wbbe  - \bbe \|  \frac{1}{\sqrt{n}} M_1 + \frac{\xi_1}{2} \|\wbbe  - \bbe \| ^2 - 2\,  \frac{\iota_n}{ A_{\bbech}} \,k\, \|\wbbe- \bbe\|= - \|\wbbe  - \bbe \|  \left(\frac{1}{\sqrt{n}} M_1 + 2\, \frac{\iota_n}{ A_{\bbech}} \,k\,\right) + \frac{\xi_1}{4} \|\wbbe  - \bbe \| ^2  \;, 
\end{equation*} 
that is,
\begin{equation*}
0 \geq  \frac{\xi_1}{4} \|\wbbe  - \bbe \|  -  \left(  \frac{1}{\sqrt{n}} M_1 + 2\,  \frac{\iota_n}{ A_{\bbech}} \,k\,\right)=   \frac{\xi_1}{4} \|\wbbe  - \bbe \|  -    \frac{1}{\sqrt{n}}\left( M_1  + 2\,  \frac{\sqrt{n}\, \iota_n}{ A_{\bbech}} \,k\right)\;   \;, 
\end{equation*}
which implies that
$$ \sqrt{n} \|\wbbe  - \bbe \|  \le \frac{4}{\xi_1}   \left(M_1 + 2\,  \frac{\sqrt{n}\, \iota_n}{ A_{\bbech}} \,k\right) $$
Taking into account that $\sqrt{n}\, \iota_n=O_{\prob}(1)$, we conclude that there exists $M_{\iota}>0$ such that $\prob(\itD_n)>1-\delta/4$, for all $n$, where $\itD_n=\{\sqrt{n}\, \iota_n\le\; M_{\iota} \} $. Hence, the set $\itE_n=\itA_n \cap \itB_n \cap \itC_n\cap \itD_n$ has probability larger than $1-\delta$, for $n\ge  \max( n_1,n_3)$ and in $\itE_n$,   $   \sqrt{n} \|\wbbe_n - \bbe_0\|\le ({4}/{\xi_1})  (M_1 + 2\, k M_{\iota}/ A_{\bbech})  $ which concludes the proof.
\end{proof}

\begin{proof}[Proof of Theorem \ref{teo:rate2}]
The proof follows similar arguments to those considered in the proof of Theorem 3 in \citet{Bianco:Boente:Chebi:2022}, but adapted to the model we are considering.  
Given $\tau > 0$, we will show that $\prob\left(\wbbe_{\noact} = \bcero_{q-k}\right) > 1 - \tau$ for $n$ large enough. As in the proof of Theorem \ref{teo:rate}, we give the common steps and then differentiate according to the penalty used.

Define $V_n: \real^k \times \real^{q-k} \to \real$ as
$$
V_n(\bu_1, \bu_2) = \wELE_n\left (\bbe_{\act} + \frac{\bu_1}{\sqrt{n}}  ,  \frac{\bu_2}{\sqrt{n}}\right ) + \itJ_{\blach}\left (\bbe_{\act} + \frac{\bu_1}{\sqrt{n}}, \frac{\bu_2}{\sqrt{n}}\right ),
$$ 
where $\wELE_n(\bb)$ is defined in \eqref{eq:wELEn}. Taking into account that $\sqrt{n}\|\wbbe  - \bbe\|=O_{\prob}(1)$, we have that there exists  $C > 0$   such that $\prob(\itC_n) \geq 1 - \tau/4$, for all $n\in \natu$, where $\itC_n = \{\sqrt{n}\|\wbbe  - \bbe\| \leq C\} $.
Then taking into account that  $\bbe=(\bbe_{\act}\trasp,\cero_{q-k}\trasp)\trasp$, for each $\omega \in \itC_n$, we have that $\wbbe$ can be written as 
\begin{equation}
\label{eq:wbbeenCn} 
\wbbe  = \left(\bbe_{\act}\trasp + \frac{\wbu_{1}\trasp}{\sqrt{n}},  \frac{\wbu_{2}\trasp}{\sqrt{n}}\right)\trasp \,,
\end{equation}
where $\|\wbu\|\le C$ and $\wbu = (\wbu_{1}\trasp, \wbu_{2}\trasp)\trasp$, $\wbu_1=\wbu_{1,n} \in \real^{k}$, $\wbu_2=\wbu_{2,n} \in \real^{q-k}$. Using that   \eqref{eq:opt-plam} implies that $\wbbe=\argmin_{\bb\in\real^q}PL_{n,\blach}(\bb)=\argmin_{\bb\in\real^q} \{\wELE_n(\bb)+\itJ_{\blach}(\bb)\}$ and that for $\omega \in \itC_n$, we have  the representation \eqref{eq:wbbeenCn}, we get that
\begin{equation}
\label{eq:un}
(\wbu_{1}\trasp, \wbu_{2}\trasp)\trasp=\argmin_{\|\bu_1\|_2^2 + \|\bu_2\|_2^2\le C^2}V_n(\bu_1, \bu_2) \,.
\end{equation} 
Our goal is to prove that, with high probability, ${V_n(\bu_1, \bu_2) - V_n(\bu_1,\bcero_{q-k})> 0}$  for all $\|\bu_1\|_2^2+\|\bu_2\|_2^2\le C^2$ with $\bu_2\ne \bcero_{q-k}$. 

Take $\bu_1 \in \real^{k}$ and $\bu_2 \neq \bcero_{q-k}$ such that $\|\bu_1\|_2^2+\|\bu_2\|_2^2\le C^2$. Note that $V_n(\bu_1, \bu_2) - V_n(\bu_1,\bcero_{q-k}) = S_{1,n}(\bu)+ S_{2,n}(\bu)$, where $ \bu=(\bu_1\trasp, \bu_2\trasp)\trasp$ and
\begin{align*}
S_{1,n}(\bu) &= \wELE_n\left (\bbe_{\act} + \frac{\bu_1}{\sqrt{n}}, \frac{\bu_2}{\sqrt{n}}\right ) - \wELE_n\left (\bbe_{\act} + \frac{\bu_1}{\sqrt{n}}, \bcero_{q-k}\right ),\\
S_{2,n}(\bu) &=  \itJ_{\blach}\left (\bbe_{\act} + \frac{\bu_1}{\sqrt{n}}, \frac{\bu_2}{\sqrt{n}}\right ) -\itJ_{\blach}\left (\bbe_{\act} + \frac{\bu_1}{\sqrt{n}}, \bcero_{q-k}\right )\,.
\end{align*}
\underline{Let us begin by deriving (a).} For that purpose, we will first provide a lower bound for $S_{2,n}(\bu)$.  
Using \eqref{eq:condicionpl}, we obtain that there exist   $n_1=n_{1,C}\in \natu$ and  $K=K_C>0$ such that for any $n \geq n_{1}$, $\prob(\itA_n)> 1-\tau/4$, where 
$$\itA_n= \left\{p_{\lambda_s}\left(\frac{|u|}{\sqrt{n}}\right)\ge K \lambda_s \frac{|u|}{\sqrt{n}}\; \mbox{for any $|u|\le C$, $k+1\le s\le q$}\right\}\,.$$ 
Then,  if we denote $u_s$ the $s-$th component of the vector $\bu$, using that $|u_s|\le \|\bu\|\le C$,  we have that in $\itA_n$
\begin{align*}
S_{2,n}(\bu)&= \sum_{s=1}^k p_{\lambda_{n,s}}\left(\left|\beta_s+ \frac{u_{s}}{\sqrt n}\right|\right)+\sum_{s=k+1}^{q} p_{\lambda_{n,s}}\left(\left| \frac{u_{s}}{\sqrt n}\right|\right)- \sum_{s=1}^k p_{\lambda_{n,s}}\left(\left|\beta_s+ \frac{u_{s}}{\sqrt n}\right|\right)\\
& = \sum_{s=k+1}^{q} p_{\lambda_{n,s}}\left(\frac{|u_{s}|}{\sqrt n} \right) \\
& \geq K  \sum_{s=k+1}^q  \lambda_s \frac{|u_s|}{\sqrt{n}} \geq K \frac{\min_{k+1\le s\le q} \lambda_s}{\sqrt{n}}  \sum_{s=k+1}^q     |u_s|  \ge K \frac{\min_{k+1\le s\le q} \lambda_s}{\sqrt{n}} \|\bu_2\|\,,
\end{align*} 
 where the last inequality follows from the fact that for any $\bu\in \real^d$, $\|\bu\|\le \sum_{j=1}^d |u_j|$.

We will now bound $S_{1,n}(\bu)$. Let $\bu_{n}^{(0)}=(1/\sqrt{n})\left (\bcero_{k}\trasp,  \bu_2\trasp\right )\trasp=(1/\sqrt{n}) \bu_0$. As in the proof of Theorem \ref{teo:rate}, using a Taylor's expansion of order two, we obtain that
$$
S_{1,n}(\bu) = \nabla\wELE_n\left (\bbe_{\act} + \frac{\bu_1}{\sqrt{n}}, \bcero_{q-k}\right )\trasp \bu_{n}^{(0)}   
 + \frac{1}{2} (\bu_{n}^{(0)})\trasp \bA_n(\wtbb_{n})\bu_{n}^{(0)}  \,,$$
 where $\wtbb_{n}  =(\wtbb_{1,n}\trasp , \wtbb_{2,n}\trasp )\trasp$ with $\wtbb_{1,n}=\bbe_{\act} +    {\bu_1}/{\sqrt{n}}$ and $\wtbb_{2,n}= \alpha_{n,1}  {\bu_2}/{\sqrt{n}}$, 
for some $\alpha_{n,1} \in [0,1]$ is an intermediate point, $\nabla \wELE_n(\bb)$ is the gradient of the function $\wELE_n(\bb)$  given by 
$$ \nabla \wELE_n(\bb)= - \frac{1}{\wsigma}\frac 1n \sum_{i=1}^n \psi_1\left(\frac{Y_i-\wmu-\sum_{j=1}^p \weta_j(X_{ij})-\bZ_i\trasp \bb}{\wsigma}\right)\,  \bZ_i $$
 and   $\bA_n(\bb)$ is defined in \eqref{eq:matrizAn} and corresponds to  the  Hessian of $\wELE_n(\bb)$. 
  
  Note that the Mean Value Theorem entails that
$$\nabla\wELE_n\left (\bbe_{\act} + \frac{\bu_1}{\sqrt{n}}, \bcero_{q-k}\right )\trasp \bu_{n}^{(0)} = \nabla\wELE_n\left (\bbe\right )\trasp \bu_{n}^{(0)}   
  + \left(\begin{array}{c}
\dst  \frac{\bu_1}{\sqrt{n}}\\
  \bcero_{q-k}
  \end{array}\right)\trasp \bA_n(\wtbb_{n}^{\star}) \bu_{n}^{(0)}$$
with $\wtbb_{n}^{\star}= \alpha_{n,2}  ( {\bu_1}\trasp ,  \bcero_{q-k}\trasp )\trasp/{\sqrt{n}}$ with $\alpha_{n,2} \in [0,1]$.

 Thus, we can write $S_{1,n}(\bu)=S_{11,n}(\bu)+S_{12,n}(\bu)+S_{13,n}(\bu)$ where
\begin{align*}
S_{11,n}(\bu) &=  \nabla\wELE_n\left (\bbe\right )\trasp \bu_{n}^{(0)}    = - \frac{1}{\wsigma}\frac 1n \frac{1}{\sqrt{n}}\sum_{i=1}^n \psi_1\left(\frac{Y_i-\wmu-\sum_{j=1}^p \weta_j(X_{ij})-\bZ_i\trasp \bbe}{\wsigma}\right)\,  \bZ_i\trasp   \bu_0\,,\\
S_{12,n}(\bu) &= \frac{1}{n}\left(\bu_1\trasp, \bcero_{q-k}\trasp\right) \bA_n(\wtbb_{n}^{\star}) \bu_{0} \,,\\
S_{13,n}(\bu)& = \frac{1}{2} \frac{1}{n}  \bu_{0} \trasp \bA_n(\wtbb_{n})\bu_{0} \,.
\end{align*}
Then, 
\begin{align*}
 n |S_{11,n}(\bu)|& \le \left\|\sqrt{n} \;\nabla\wELE_n\left (\bbe\right ) \right\| \| \bu_0\|= \left\|\sqrt{n} \;\nabla\wELE_n\left (\bbe\right ) \right\| \| \bu_2\| \,,\\ 
n |S_{12,n}(\bu)|& \le \|\bu_1\|\left\|\bA_n(\wtbb_{n}^{\star})\right\| \| \bu_0\|\le C\, \left\|\bA_n(\wtbb_{n}^{\star})\right\| \| \bu_2\|\,,\\ 
n |S_{13,n}(\bu)|& \le \frac{1}{2}\|\bu_0\|\left\|\bA_n(\wtbb_{n})\right\| \| \bu_0\|\le \frac{C}{2}\, \left\|\bA_n(\wtbb_{n})\right\| \| \bu_2\|\,.
\end{align*}
 Lemma \ref{lema:An} entails that $\bA_n(\wtbb_{n}^{\star})\convprob \bA$ and $\bA_n(\wtbb_{n})\convprob \bA$, then, $n  (|S_{12,n}(\bu)+S_{13,n}(\bu)|)\le C  \, \|\bu_2\|\, W_{n,0}$ with $W_{n,0}= O_{\prob}(1)$. Besides,   Lemma \ref{lema:5.1} entails that $\sqrt{n} \;\nabla \wELE_n(\bbe)  = O_\prob(1)$, so  $n    |S_{11,n}(\bu)| =  C\, \|\bu_2\|\, W_{n,1}$ with $W_{n,1}= O_{\prob}(1)$ leading to $ S_{1,n}(\bu)= C\, W_n \|\bu_2\|/n$, where $W_n=O_{\prob}(1)$.
  
Let $M=M_C > 0$ be such that $ \prob\left(n\,|S_{1,n}(\bu)| >  M    \|\bu_2\|\right)<\tau/4$. Then, the set $\itB_n=\{n\, S_{1,n}(\bu) > - M  \|\bu_2\|\}$ is such that $ \prob(\itB_n) \geq 1 - \tau / 4 $.

Therefore, in $\itA_n\cap \itB_n$, we get that 
$$S_{1,n}(\bu) + S_{2,n} (\bu) \geq  - \frac{1}{n}\, M    \|\bu_2\| + K \frac{\min_{k+1\le s\le q} \lambda_s}{\sqrt{n}} \|\bu_2\| = 
 \|\bu_2\| \frac{1}{n} \left(K\sqrt{n}  \min_{k+1\le s\le q} \lambda_s  -   \, M\right)$$
Taking into account that $ \sqrt{n}  \min_{k+1\le s\le q} \lambda_s \convprob \infty$, if we define  $\itL_n=\{ \sqrt{n}  \min_{k+1\le s\le q} \lambda_s > (M+1)/K \}$ we have that $\lim_{n\to \infty}\prob(\itL_n)= 1$. Thus, there exists $n_2\in \natu$, such that for $n\ge n_2$, $\prob(\itL_n)>1-\tau/4$.

Take  $\itD_n=\itA_n\cap \itB_n\cap \itL_n\cap \itC_n$, then,  for $n>\max(n_1,n_2)$, $\prob(\itD_n)\ge 1-\tau$. Furthermore, for any $\omega \in \itD_n$, 
we have that
$$n\left(V_n(\bu_1, \bu_2) - V_n(\bu_1,\bcero_{q-k})\right)=n \left(S_{1,n}(\bu) + S_{2,n} (\bu) \right) \geq  \|\bu_2\|>0\,,$$
for any $\bu=(\bu_1\trasp, \bu_2\trasp)\trasp$, such that $\|\bu\|\le C$  and  $\bu_2\ne \bcero_{q-k}$, so $V_n(\bu_1, \bu_2)> V_n(\bu_1,\bcero_{q-k})$. 
Besides, from \eqref{eq:wbbeenCn} and \eqref{eq:un}, we also have that $\wbbe  = \bbe+ \wbu/{\sqrt{n}}$ with   $\|\wbu\|\le C$ and  
$$ \wbu =\argmin_{\|\bu\|\le C: \bu=(\bu_1\trasp, \bu_2\trasp)\trasp}V_n(\bu_1, \bu_2) \,,$$ 
implying that $\wbu_{2}= \bcero_{q-k}$ in $\itD_n$ and concluding the proof. 

\underline{Let us proceed to derive (b).} As in the proof of (a), we give a lower bound for $S_{2,n}(\bu)$.  
Using that $\sqrt{n}\|\wbbe_{\ini}-\bbe\|=O_{\prob}(1)$, we obtain that there exists     $A>0$ such that for any $n\in \natu$, $\prob(\itA_n)> 1-\tau/4$, where $\itA_n= \left\{ \sqrt{n} \max_{k+1\le j\le q}|\wbeta_{\ini,s}|\le A\right\}$. 
Then,  if as above we denote $u_s$ the $s-$th component of the vector $\bu$, using that $|u_s|\le \|\bu\|\le C$,  we have that in $\itA_n$
\begin{align*}
S_{2,n}(\bu)
& = \sum_{s=k+1}^{q} p_{\lambda_{n,s}}\left(\frac{|u_{s}|}{\sqrt n} \right) = \iota_n \sum_{s=k+1}^{q} \frac{|u_{s}|}{\sqrt n} \frac{1}{|\wbeta_{\ini,s}|}  \geq \frac{1}{A} \iota_n \sum_{s=k+1}^{q}  {|u_{s}|} \ge    \frac{1}{  A} \iota_n   \|\bu_2\| \,,
\end{align*} 
where we have used again that $ \sum_{s=k+1}^{q}  {|u_{s}|}\ge \|\bu_2\|$.

As in (a), we have that $ S_{1,n}(\bu)= C\, W_n \|\bu_2\|/n$, where $W_n=O_{\prob}(1)$, so let  $M=M_C > 0$ be such that $ \prob\left(n\,|S_{1,n}(\bu)| >  M    \|\bu_2\|\right)<\tau/4$. Then, the set $\itB_n=\{n\, S_{1,n}(\bu) > - M  \|\bu_2\|\}$ is such that $ \prob(\itB_n) \geq 1 - \tau / 4 $.

Therefore, in $\itA_n\cap \itB_n$, we get that 
$$S_{1,n}(\bu) + S_{2,n} (\bu) \geq  - \frac{1}{n}\, M    \|\bu_2\| +\frac{1}{  A} \iota_n   \|\bu_2\|  = 
 \|\bu_2\| \frac{1}{n} \left( \frac{1}{  A} \,n\,\iota_n   -   \, M\right)$$
Taking into account that $  n\,\iota_n \convprob \infty$, if we define  $\itL_n=\{  n\,\iota_n >   A\, (M+1)\}$ we have that $\lim_{n\to \infty}\prob(\itL_n)= 1$. Thus, there exists $n_0\in \natu$, such that for $n\ge n_0$, $\prob(\itL_n)>1-\tau/4$.

The proof follows now as in (a) defining  $\itD_n=\itA_n\cap \itB_n\cap \itL_n\cap \itC_n$, and taking into account that,  for $n>n_0$, $\prob(\itD_n)\ge 1-\tau$ and for any $\omega \in \itD_n$, we have that for any $\|\bu\|\le C$  such that $\bu_2\ne \bcero_{q-k}$,
$n\left(V_n(\bu_1, \bu_2) - V_n(\bu_1,\bcero_{q-k})\right)  \geq  \|\bu_2\| >0$.
\end{proof}

\small

\bibliographystyle{apalike}
 
\bibliography{referencias2}

\end{document}